\documentclass{article}

\usepackage[english]{babel}
\usepackage{amsthm,amsmath,lipsum}
\usepackage{xcolor}
\usepackage{fullpage}
\usepackage{amssymb}
\usepackage{centernot}
\usepackage{tikz}
\usepackage{soul}
\usetikzlibrary{shapes.geometric, arrows}
\usetikzlibrary{positioning}
\usetikzlibrary{calc}

\DeclareFontFamily{U}{mathb}{\hyphenchar\font45}
\DeclareFontShape{U}{mathb}{m}{n}{
      <5> <6> <7> <8> <9> <10> gen * mathb
      <10.95> mathb10 <12> <14.4> <17.28> <20.74> <24.88> mathb12
}{}
\DeclareSymbolFont{mathb}{U}{mathb}{m}{n}
\DeclareMathSymbol{\llcurly}{3}{mathb}{"CE}
\DeclareMathSymbol{\ggcurly}{3}{mathb}{"CF}

\usepackage{thm-restate}

\theoremstyle{plain}

\newtheorem{example}{Example}
\newtheorem{conjecture}{Conjecture}
\newtheorem{theorem}{Theorem}
\newtheorem{lemma}{Lemma}
\newtheorem{definition}{Definition}

\newtheorem{corollary}{Corollary}
\newtheorem{claim}{Claim}
\newtheorem{proposition}{Proposition}

\usepackage{times}
\usepackage{soul}
\usepackage{url}
\usepackage[hidelinks]{hyperref}
\usepackage[utf8]{inputenc}
\usepackage[small]{caption}
\usepackage{graphicx}
\usepackage{amsmath}
\usepackage{amsthm}
\usepackage{amssymb}
\usepackage{booktabs}
\usepackage{booktabs}
\usepackage{algorithm}
\usepackage{algorithmic}
\usepackage{color}
\usepackage[switch]{lineno}

\usepackage[most]{tcolorbox}

\title{Exact and approximate maximin share allocations in multi-graphs}

\author{George Christodoulou\thanks{Aristotle University of
    Thessaloniki and Archimedes, Athena Research Center. Email: \texttt{\{gichristo,smastra\}@csd.auth.gr} }
\and {Symeon Mastrakoulis\footnotemark[1]}}
\date{}

\begin{document}
\usetikzlibrary {patterns,patterns.meta}

\maketitle
\begin{abstract}
We study the problem of (approximate) \emph{maximin share} (MMS) allocation of indivisible items among a set of agents.
We focus on the graphical valuation model, in which the input is given by a graph where edges correspond to items, and vertices correspond to agents. An edge may have non-zero marginal value only for its incident vertices. We study additive, XOS and subadditive valuations and we present positive and negative results for (approximate) MMS fairness, and also for (approximate) pairwise maximin share (PMMS) fairness.
    
\end{abstract}
\section{Introduction}
\label{sec:introduction}

The fair allocation of \textit{indivisible} goods is a fundamental problem that arises in various fields, including game theory, social choice theory, and multi-agent systems. The objective is to allocate a set of $m$ \textit{indivisible} items among $n$ agents in a way that satisfies a predefined notion of fairness. Over time, various fairness criteria have been explored, each capturing a distinct interpretation of what constitutes a ``fair'' allocation. In the case of \textit{divisible} goods---studied in the context of cake-cutting---the key fairness concepts are  {\em envy-freeness} and {\em proportionality}. An allocation is called envy-free if no agent envies the portion allocated to another agent, while it is called proportional if every agent receives her proportional share, which is at least $1/n$ of her total value for the whole cake. Unfortunately, it is well known that both notions may fail to exist in the discrete setting. It is therefore natural to employ relaxed or approximate fairness notions when dealing with indivisible goods. 

In this work, we consider \textit{maximin share (MMS)}, the most prominent relaxation of proportionality in the context of indivisible goods, introduced by \cite{BudishMMS}. 
Each agent $i$ has an associated threshold, called her maximin share $\mu_i$, which is equal to the maximum value agent $i$ can secure if she partitions the set of items into $n$ bundles and receives the lowest-value bundle. Under this notion, an allocation is considered fair if every agent receives at least her MMS value. 

\cite{KurokawaProcacciaWang18} first showed that, unfortunately, it is not always possible to guarantee the MMS value for every agent, even under additive valuations. Consequently, research has shifted toward approximate MMS fairness, where each agent is guaranteed a given fraction of her MMS value. This has led to a surge of research in the past decade, primarily for additive valuations \cite{AmanatidisMNS17, AkramiGargSODA24, GargTaki21, GhodsiHSSY21} but also for more general valuation classes \cite{GhodsiHSSY22, BarmanKrishnamurthy20, SeddighinSeddighin24, AkramiNeurIPS23, FeigeGrindberg25,FeigeHuang25}. Determining the best possible guarantees for important valuation classes, including additive, XOS, and subadditive, remains an active area of research.

\paragraph{Graphical Valuations.} While exact MMS allocations are not always achievable in general, they are known to exist for the case of two agents. A natural extension of the two-agent setting to the multi-agent setting, is the graphical model introduced by \cite{Christodoulou} in the context of EFX.\footnote{EFX, proposed in \cite{CaragiannisKurokawaACM19} is an extensively studied relaxation of envy-freeness.} In this model, agent valuations are represented by a graph where vertices correspond to agents, and edges correspond to items. Each edge $\{i,j\}$ may have a positive marginal value for agents $i$ and $j$ only; for all other agents, the item has zero marginal value. 
In the case of {\em multi-graphs}, multiple parallel edges between any pair of vertices are possible. The graphical model has recently received considerable attention in the study of EFX, with significant progress made for various classes of (multi-)graphs (see e.g. \cite{HsuCoRR24,BhaskarPanditCoRR24,Afshinmehr24,ZhouIJCAI24,MirsaSethia24, SgouritsSotiriouEFXMULTI2025,BSRS25}) drawing interesting connections between graph theory and fair division. However, while the existence of EFX allocations in simple graphs is known~\cite{Christodoulou}, the question of EFX existence for general multi-graphs remains open.

In this work, we study approximate MMS fairness, and other related notions of fairness, for the case of multi-graphs.

The fair allocation of \textit{indivisible} goods is a fundamental problem that arises in various fields, including game theory, social choice theory, and multi-agent systems. The objective is to allocate a set of $m$ \textit{indivisible} items among $n$ agents in a way that satisfies a predefined notion of fairness. Over time, various fairness criteria have been explored, each capturing a distinct interpretation of what constitutes a ``fair'' allocation.

Unlike the case of \textit{divisible} goods---studied in the context of cake-cutting---key fairness concepts such as envy-freeness and proportionality may not exist in the discrete setting. It is therefore natural to employ relaxed or approximate fairness notions when dealing with indivisible goods. Two of the most extensively studied relaxations of envy-freeness are EF1, introduced by Budish \cite{BudishMMS}, and EFX, proposed by Caragiannis et al. \cite{CaragiannisKurokawaACM19}.

Beyond envy-freeness relaxations, another prominent approach to fairness in the allocation of indivisible goods is the \textit{maximin share (MMS)}, introduced by Budish in \cite{BudishMMS}. MMS, which is the focus of this work, has garnered significant attention in recent years and is one of the most well-studied fairness notions for indivisible goods. This notion serves as an important relaxation of proportionality in the context of indivisible goods. Each agent $i$ has an associated threshold, called her maximin share $\mu_i$, which is equal to the maximum value agent $i$ can secure if she partitions the set of items into $n$ bundles and receives the lowest-value bundle. An allocation is considered fair if every agent receives at least her MMS value. 

Kurokawa, Procaccia, and Wang \cite{KurokawaProcacciaWang18} showed that it is not always possible to guarantee the MMS value for every agent, even for additive valuations. Consequently, research has shifted toward approximate MMS fairness, where each agent is guaranteed a given fraction of her MMS value. This has led to a surge of research in the past decade, primarily for additive valuations (e.g., \cite{AmanatidisMNS17, AkramiGargSODA24, GargTaki21, GhodsiHSSY21}) but also for more general valuation classes (e.g., \cite{GhodsiHSSY22, BarmanKrishnamurthy20, SeddighinSeddighin24, AkramiNeurIPS23, FeigeGrindberg25,FeigeHuang25}). Finding the best possible guarantees for important valuation classes, including additive, XOS, and subadditive, remains an active area of research. 

While exact MMS allocations are not always achievable in general, they are known to exist for two agents. A natural extension of the 2-agent setting is the graph model introduced by Christodoulou et al.~\cite{Christodoulou} in the context of EFX. In this model, agent valuations are represented by a graph where vertices correspond to agents, and edges correspond to items. Each edge $\{i,j\}$ may have a positive marginal value for agents $i$ and $j$ only; for all other agents, the item has zero marginal value. We study approximate MMS fairness for \textit{multi-graphs},\footnote{For simple graphs, each agent's MMS value is 0, making any allocation trivially MMS-fair.} where multiple parallel edges may connect a pair of vertices. Multi-graphs have recently received considerable attention in the study of EFX, with significant progress made for various classes of multi-graphs (e.g., \cite{HsuCoRR24,BhaskarPanditCoRR24,Afshinmehr24,ZhouIJCAI24,MirsaSethia24, SgouritsSotiriouEFXMULTI2025}). However, the question of EFX existence for general multi-graphs remains open.

\subsection{Our Contribution}

We study (approximate) MMS, the predominant notion of share-based
fairness for the case of indivisible goods, for graphical valuations which has been previously studied in the context of EFX. We focus on {\em multi-graphs}\footnote{We remark that for simple graphs, each agent's MMS value is $0$, making any allocation trivially MMS fair.} and
explore additive, XOS, and subadditive valuations.\footnote{We refer the reader to Section~\ref{sec:preliminaries} for formal definitions of fairness notions and valuation functions.} Importantly, our results apply to general {\em multi-graphs} without any additional assumptions on the graph properties.
%
We also obtain results on {\em pairwise maximin share (PMMS)} \cite{CaragiannisKurokawaACM19} as well as on \emph{ordinal approximation ($1$-out-of-$d$)} \cite{BudishMMS} which is strongly related to MMS. An allocation is PMMS fair, if for any pair of agents $i$ and $j$, it guarantees that agent $i$ receives at least the MMS threshold on the restricted set of items she shares with $j$. An allocation is $1$-out-of-$d$ fair if it ensures to each agent the minimum value she could guarantee if she could partition all the items in $d$ parts, assuming she receives her least preferred bundle.

We explore additive, XOS, and subadditive valuations.\footnote{We refer the reader to Section~\ref{sec:preliminaries} for formal definitions of fairness notions and valuation functions.}

\paragraph{Additive Valuations.} In Section~\ref{sec:additive}, we study additive valuations and show that in multi-graphs, an exact MMS allocation~\footnote{For MMS, we can focus on orientations---where each edge is allocated to an incident node--- as for every allocation, there exists an equally good orientation while the orientation is not wasteful.} always exists for any number of agents (Theorem
\ref{thm:multiaddcutandchoose}). This contrasts with the general (non-graphical) case, where exact MMS cannot be achieved, for $n \geq 3$~\cite{KurokawaProcacciaWang18}. In fact, we establish two stronger results: first, the existence of an allocation that is both MMS and PMMS (Theorem~\ref{thm:multiaddcutandchoose}), and second, the existence of a $1$-out-of-$3$ allocation, and hence MMS for $n\ge 3$ (Theorem~\ref{thm:multiaddroundrobin}). 

We also show that even in graphs, MMS does not imply $\alpha$-approximation of PMMS up to any factor $\alpha$. This was previously known for general
instances \cite{CaragiannisKurokawaACM19}.  
Interestingly, our impossibility result holds even for the case of {\em symmetric} multi-graphs, where each item has the same value for both endpoints (agents).

\paragraph{XOS Valuations.}In Section~\ref{sec:XOS}, we consider XOS valuations and prove the existence of a $2/3$-MMS allocation for any $n \geq 3$ (Theorem~\ref{thm:XOS4+}). This contrasts with
general (non-graphical) XOS valuations, where an upper bound of $1/2$
is known~\cite{GhodsiHSSY22}. We complement this positive
result by providing a $\left(1-\frac{1}{\left\lceil n/2 \right\rceil+1}\right)$-MMS 
approximation upper bound for $n$ agents. (Theorem~\ref{thm:XOSupper}).
For $3$ and $4$ agents, this bound matches our lower bound of $2/3$.
Both the upper and lower bounds rely on an interesting connection with extremal combinatorics, particularly the existence of independent transversals in multipartite graphs. This suggests that the primary bottleneck in obtaining near-optimal fairness guarantees for XOS valuations is combinatorial in nature, rather than an inherent limitation of the graphical model itself.

Moreover, we show a $(1-1/d)$-out-of-$d$ approximation when the number of agents is
$n \le 3$ and a $\frac{1}{2}$-out-of-$2$ approximation for an arbitrary number of agents
$n$. We explicitly show that there exists a graph with XOS agents in which
no $(1/2+\varepsilon)$-PMMS \emph{orientation} exists, but an exact PMMS allocation does exist (Theorem~\ref{thm:subadditivePMMS}).

In Section~\ref{sec:SBXOS}, we explore the possibility of achieving approximation guarantees beyond the $2/3$ factor for XOS agents using tools from extremal graph theory. We prove the existence of a $2/3$-out-of-$8\sqrt{n}$-MMS orientation (Theorem~\ref{thm:sqrtXOS}). Furthermore, we establish a formal connection between the existence of $(1-1/\ell)$-MMS orientations and a well-known conjecture in extremal graph theory regarding independent transversal sets. Specifically, we show that if Conjecture~\ref{thm:ITSconj} holds, it implies that every multigraph admits a $(1-1/\ell)$-out-of-$\mathcal{O}(n^{1-\epsilon})$-MMS orientation and, if there exists a constant $\alpha < 1$ such that an $\alpha$-MMS allocation is not guaranteed to exist for any instance with $n$ XOS agents with graphical valuations, then Conjecture~\ref{thm:ITSconj} does not hold (Theorem~\ref{thm:connectionITS}). 

\paragraph{Subadditive Valuations.} Finally, in Section \ref{sec:subadditive} we consider the case of subadditive valuations\footnote{Independently, \cite{Feige25} showed the existence of $1/2$-out-of-$2$ for subadditive agents and multi-graphs, using different techniques.} and we show a \emph{tight} $1/2$-MMS approximation for any $n\ge 2$ (Theorem \ref{Thm:subadditivemulti}). We emphasize that our results imply a separation between subadditive and XOS valuations with respect to approximate MMS in the graph model. Notably, such a separation remains an open problem in the general (non graphical) setting. Furthermore, for orientations, we show the existence of a $1/2$-PMMS orientation and provide a matching impossibility result that is tight even for XOS.

\paragraph{Discussion.} We remark that while at first look PMMS may seem to "fit" better within the graphical model, our results show that it is {\em wasteful}\footnote{This is similar to the case of EFX in graphs.} i.e. there exist instances in which in order to achieve PMMS fairness, we may need to allocate to some agent an item with no value for her.  In contrast, we show that this is not the case for MMS fairness. That is, the MMS fairness of an allocation from the perspective of agent $i$ does not depend on the bundle or the valuation of any other agent $j$. Thus, it is reasonable to study MMS in graphical model and as a result, our techniques and our tools can be also used in the general model.

We also remark that, when all items are goods (i.e., have non-negative values), the $1$-out-of-$d$ threshold is (weakly) monotonically decreasing as the number of partitions $d$ increases. Hence, we obtain stronger results using $d < n$ such as when $d=2$.

\begin{table}
    \centering
        \caption{Best known approximate MMS for general model and for the graph model. Our contributions appear in \textbf{\color{blue} bold blue}. One fraction indicates a tight approximation and two fractions with brackets indicate lower and upper bounds. $^{\dagger}$\cite{HuangZHouAdditiveMMS}; $^{\dagger\dagger}$\cite{FeigeSapirTauberWINE21};
    $^{\ddagger}$\cite{christodoulouIJCAI25};  $^{\mathparagraph}$ \cite{FeigeGrindberg25}; $^{\mathsection}$ \cite{GhodsiHSSY22}; $^{**}$\cite{Feige25}.}
    \begin{tabular}{c c c }
        \toprule
        {\bf Valuation Class} & { \bf General Model} & { \bf Graph Model}  \\
        \midrule
        Additive \\
        \midrule
         $n \ge 2$ &$ \left [ 7/9^{\dagger},1-\frac{1}{n^4}^{\dagger\dagger}\right]$ &  {\bf \color{blue} Exact MMS [Thm \ref{thm:multiaddcutandchoose}]}\\
        \midrule
        XOS \\
        \midrule
        $n=2$ & $1/2 ^{\ddagger}$ & $1/2^{\ddagger}$ \\
        $n \ge 3$ & $\left [4/17^{\mathparagraph},1/2^{\mathsection} \right ]$  & {\bf \color{blue}$\boldsymbol{\left[2/3 [Thm~\ref{thm:XOS4+}],\left(1-\frac{1}{\left\lceil n/2 \right\rceil+1}\right)[Thm~\ref{thm:XOSupper}]\right]}$}\\
        \midrule
        Subadditive \\
        \midrule
        $n \ge 2$ & $\left[  \frac{1}{\mathcal{O}(\log\log n)}^{**},1/2^{\mathsection}\right]$ & {\bf \color{blue}$\boldsymbol{1/2}$  [Thm \ref{Thm:subadditivemulti}]}  \\
        \bottomrule
    \end{tabular}
    \label{table:Results}

\end{table}

Overall, we show that the graph setting yields provably improved results for additive and XOS valuations, but not for subadditive valuations, where we match the best known upper bound. We also investigate the relationship between orientations versus allocations under the PMMS notion. Additionally, we study other important share-based notions of fairness, such as PMMS and 1-out-of-$d$ MMS fairness.

\subsection{Further related work}
In this paper we focus on approximate MMS fairness. We also provide
results for PMMS and $1$-out-of-$d$ fair notions which are strongly
correlated with MMS. We refer the interested reader to the recent
survey of \cite{AmanatidisABFLMVW23Survey} covering a wide variety of
discrete fair division settings along with the main fairness notions
and their properties.

\paragraph{Maximin Share and $\alpha$-MMS}  During the past years MMS has seen significant progress for all valuation classes. Although for two agents with additive valuations, MMS allocations always exist, it is known that there exist instances where it is impossible to allocate items in such a way to guarantee for every agent her exact MMS value. Kurokawa, Procaccia and Wang proved an upper bound of $1 - \mathcal{O}(\frac{1}{2^n})$ \cite{KurokawaProcacciaWang18} for additive values and $n$ agents. Later Feige, Sapir and Tauber provided an improved bound of $1-\frac{1}{n^4}$ in \cite{FeigeSapirTauberWINE21}. For additive valuations there is an abundance of works which has led to strong approximation guarantees. Among all the works we note the existence of $2/3$-MMS allocation for $n$ agents in \cite{KurokawaProcacciaWang18}, the existence of $3/4$-MMS allocation for $n$ agents from Ghodsi et al. in \cite{GhodsiHSSY21}, the existence of $7/8$ for $n=3$ agents in \cite{AmanatidisMNS17} by Amanatidis et al. and the existence of $3/4+1/12n$ for $n$ agents by Garg and Taki in \cite{GargTaki21}. Akrami et al.~\cite{AkramiGST23} improved the approximation ratio to $\frac{3}{4} + \min\left\{\frac{1}{36}, \frac{3}{16n-4}\right\}$, which was subsequently refined by Akrami and Garg~\cite{AkramiGargSODA24} to $\frac{3}{4} + \frac{3}{3836}$. Recently, Heidari et. al. \cite{HeidariSodaAdditiveMMS} prove the existence of $10/13$-MMS allocations and, using a similar method, Huang and Zhoue \cite{HuangZHouAdditiveMMS} prove the existence of $7/9$-MMS allocations which is the best approximation we know so far for additive valuations

Our understanding of valuation classes beyond additive is limited. For submodular functions, a tight $2/3$- approximation for MMS allocation with two agents is known \cite{Coauthor, Ruta}. For more agents, Uziahu and Feige showed a lower bound of $10/27$ \cite{UziahuFeigeSubmodular}, and Ghodsi et al. established an upper bound of $3/4$ \cite{GhodsiHSSY22}. Moreover, Ghodsi et al.~\cite{GhodsiHSSY22} proved the existence of a $1/5$-MMS allocation for fractionally subadditive valuations and a $1/2$ upper bound for approximate MMS. Later the lower bound improved to $3/13$ \cite{AkramiNeurIPS23} and recently to $4/17$ \cite{FeigeGrindberg25}. 

The gaps are larger for the class of subadditive valuations; there always exists a $\frac{1}{\mathcal{O}(\log\log n)}$-MMS allocation \cite{Feige25} while the best known upper bound is $1/2$ \cite{GhodsiHSSY22}.

Furthermore, Chekuri et al. proved that $1/2$-MMS allocations exist for {\em SPLC} valuations \cite{ChekuriKKM24}. The same guarantee holds for the case of hereditary set systems as Hummel shows in \cite{Hummel:HSS24}. Additionally, for valuations which are both {\em leveled} and submodular a lower bound of $2/3$-MMS allocation is known, and for valuations which are both leveled and subadditive we know a $1/2$-MMS allocation \cite{Coauthor}.

\paragraph{Pairwise maximin share and $\alpha$-PMMS} Another fairness notion which we will consider in this work, is the pairwise maximin share (PMMS) introduced by Caragiannis in \cite{CaragiannisKurokawaACM19}.
It is an open problem if we can guarantee exact PMMS value for all agents even when their valuations are additive. In \cite{CaragiannisKurokawaACM19}, it is shown that a PMMS allocation is also EFX (if each item has strictly positive value for both agents). In the same work they show that an approximate $0.618$-PMMS allocation always exists and they also establish that a PMMS allocation is also $1/2$-MMS allocation for additive valuations despite the fact that neither PMMS implies MMS nor the opposite. Later, Amanatidis, Birmpas and Markakis in \cite{AmanatidisBirmpasMarkakisIjcai18} proved an implication of $4/7$-MMS from a PMMS allocation, as well as an inapproximability result for $0.5914$-MMS. For approximate PMMS, the best known result is 0.781 by Kurokawa \cite{KurokavaThesis17} for additive valuations. Zhang et al.  \cite{Zhang24} studied the efficiency of PMMS computation and the existence for additive valuations when all agents agree on the ordinal ranking of the goods. Barman and Verma \cite{BarVermAAMAS20} studied the PMMS notion under matroid-rank valuations.

\paragraph{$1$-out-of-$d$ MMS}We will also obtain results of the $1$-out-of-$d$ fair notion (\emph{ordinal approximation}) which was introduced by Budish in \cite{BudishMMS} and showed the existence of $1$-out-of-$(n+1)$-MMS, by adding a small number of excess goods. This is the value that an agent can secure by partitioning the items into $d$ bundles, assuming she will receive the least preferred bundle. We note that when all items are goods (i.e., have non-negative values), the $1$-out-of-$d$ threshold is (weakly) monotonically decreasing as the number of partitions $d$ increases. For additive valuations the main open problem is the minimum $d$ for  such that we can guarantee at least $1$-out-of-$d$ for all agents.  In \cite{Aigner-HorevSH22} Aigner and Horev obtained a bound of $d=2n-2$ on the existence of $1$-out-of-$d$ allocations under additive valuations without adding excess goods. Later the approximation was subsequently improved to $1$-out-of-$\lceil 3n/2 \rceil$ by Hosseini and Searns \cite{HosseiniSearns21} and after to $1$-out-of-$\lfloor 3n/2 \rfloor$ by \cite{HosseiniSSH22}. The minimum such $d$ we know so far is proved by Akrami et al. \cite{AkramiGST24}. Whether $1$-out-of-$(n+1)$ MMS allocations always exist remains an open question. Babaioff, Nisan and Talgam{-}Cohen \cite{BabaioffNT21} introduced the $\ell$-out-of-$d$ maximin share, which corresponds to the maximum value an agent can guarantee to herself if she partitions the items into $d$ bundles and  then being allocated the worst $\ell$ of them. We note that in all these works, ordinal approximation was used with $n < d$. For more general valuations Hosseini, Searns and Segal show that an exact $1$-out-of-$d$ MMS allocation may not exist under submodular valuations \cite{HosseiniSSH22} for any $d \ge 1$ even for $2$ agents.

\paragraph{Graphical Valuations} Christodoulou et. al. \cite{Christodoulou} first studied the graphical model considering the EFX property and show that, while an EFX allocation exists for simple graphs, an EFX orientation is not guarantee to exist and deciding if it exists is NP complete. Zhou et. al. \cite{ZhouIJCAI24} studied the graphs with mixed mana; items with either positive or negative value. \cite{ZM24} showed that EFX orientations are not guaranteed to exist in graphs with chromatic number greater than $3$, and they always exist when the chromatic number is at most $2$. In \cite{DEGK24}, Deligkas et. al.  showed that EFX orientations need not exist in multi-graphs with $10$ vertices.
\cite{Afshinmehr24, BhaskarPanditCoRR24, SgouritsSotiriouEFXMULTI2025, DEGK24, HsuCoRR24, BSRS25, MirsaSethia24} studied the existence of EFX and also the complexity of finding EFX orientations for additive and more general valuations in multi-graphs under special conditions i.e. bipartite graphs, length of the shortest cycle, bounded neighbors, multi-trees etc. Mirsa and Sethia \cite{MirsaSethia24} studied the graph model assuming binary valuations.

We also note that this graphical model has been studied in other settings, including machine scheduling~\cite{EbeKrcJgall14, VershaeWiese10} and mechanism design~\cite{ChristodKoutsKov21}.

\section{Preliminaries}
\label{sec:preliminaries}
In this section, we introduce the main concepts and notation. We consider the graphical valuation. In this model, there is a given undirected graph $G=(V,E)$, where the vertices correspond to a set of $n$ agents $V=\{1,\ldots,n\}$ and the set of edges corresponds to a set $E=\{1,\ldots, m\}$ of $m$ indivisible goods. Each vertex (agent) $i$ is equipped with a valuation function $v_i:2^E\rightarrow \mathbb{R}^+$, where $v_i\left(X\right)$ is the value of agent $i$ for the subset $X \subseteq E$. An item $e\in E$ is \emph{irrelevant} to agent $i\in V$ if for all subsets $X\subseteq E$, it holds that $v_i(X\cup \{e\}) = v_i(X)$. If an item is not irrelevant for agent $i$, it is called {\em relevant} for agent $i$. The graph $G=(V,E)$ induces a special structure for the valuations as follows; an edge $e=(i,j)$ is only relevant to at most agents $i$ and $j$.\footnote{We allow the edge $e$ to be irrelevant for one of these agents (say for agent $j$), but in this case this edge could be equally modeled by a self-loop $e=(i,i)$.} We allow multiple edges between two agents $i$ and $j$ which we denote by $E_{i,j}$,\footnote{We allow $i=j$, $E_{i,i}$ for the case of possibly multiple self-loops and an agent $i$.} and we denote by $E_i$ the set of all edges adjacent to agent $i$, i.e., $E_{i}=\bigcup_jE_{i,j}$.

\paragraph{Valuation Classes.}
We consider valuation classes that are {\em monotone}, i.e., $v_i(S)\leq v_i(T)$ whenever $S\subseteq T\subseteq E$ and are {\em normalized} i.e., $v_i(\emptyset)=0$, for all $i\in V$. We study several important classes separately of complement-free valuations, such as additive, XOS, and subadditive valuation functions which we define below. It is well known that these classes form a hierarchy, namely $\text{additive } \subsetneq \text{ XOS } \subsetneq \text{ subadditive}$. A valuation function $v$ is \begin{itemize}
  
\item {\em additive,} if $v(S) = \sum_{g \in S} v(g)$ for any $S \subseteq E$.

\item  {\em fractionally subadditive (a.k.a. XOS)} if there exists a set of additive functions $a_1, \dots, a_k$ such that $v(S) = \max_{l \in [k]}a_l(S)$ for any $S \subseteq E$.  
\item {\em subadditive} if $v(S) + v(T) \ge v(S \cup T)$ for any $S, T \subseteq E$.
\end{itemize}

It is well-known that these classes form a hierarchy, namely $\text{additive } \subsetneq \text{ XOS } \subsetneq \text{ subadditive}$.

\paragraph{Allocations and Orientations.} We are interested in allocating (a subset of) edges in $E$ into $n$ mutually disjoint sets $X_1,\ldots, X_n$, where $X_i$ is the set of edges assigned to agent $i$. We denote the respective (partial) allocation by $X=(X_1, \ldots, X_n)$. An allocation $X$ is called a {\em partition} if $\bigcup_{i\in V} X_i=E$ and {\em partial} if $\bigcup_{i\in V} X_i \subset E$. An allocation $X$ is an \emph{orientation} if every edge is allocated to an agent for which it is relevant; that is, for each edge $e=(i,j)$, it must hold $e\in X_i\cup X_j$.

\paragraph{Fairness Notions.} 
Here, we provide the definitions for the fairness notions that we consider. For a positive integer $d$, let $P=(P_1, \ldots, P_d)$ be a partition of $E$ into $d$ parts (we also call it $d$-partition), and let $\Pi_d(E)$ be the set of all possible $d$-partitions of $E$.

We define the {\em $1$-out-of-$d$ maximin share value} of agent $i$, denoted by $\mu_i^d\left(E\right)$, to be the maximum over
    all $d$-partitions of $E$, of the minimum value under $v_i$ of a part in the $d$-partition, i.e.
    \[
    \mu^d_i\left(E\right) = \max_{P \in \Pi_d\left(E\right)}\min^{d}_{j=1}v_i\left(P_j\right).
  \]
  Of special interest is the case of $d=n$ which is known as the \emph{maximin share (MMS)} value of agent $i$, denoted by $\mu_i(E):=\mu_i^n(E)$. When $E$ is clear from the context, we use the simpler notation $\mu_i^d$ and $\mu_i$ respectively.
  
We are now ready to define the notion of (approximate) MMS allocations. Fix an $\alpha \in (0,1]$. We call $X=(X_1,\ldots, X_n)$ an {\em $\alpha$-out-of-$d$ MMS allocation} if $v_i\left(X_i\right) \ge \alpha \mu_i^d$ for all agents $i$. Of special interest is the case of $d=n$, and then $X$ is called an {\em $\alpha$-MMS allocation}. When  $\alpha =1$, $X$ is simply called an {\em (exact) MMS allocation}. Regarding MMS approximations, we can focus on orientations, since for every allocation there exists an equally good orientation.

\paragraph{Canonical Partitions.} A partition $P$ of $E$ into $d$ parts, which maximizes the $\min^{d}_{j=1}v_i\left(P_j\right)$ for agent $i$ is called a $1$-out-of-$d$ MMS partition of agent $i$. Note that there may be more than one such partition. We fix one (arbitrary) partition $\hat{P}^d_i$ and define the respective partition of the set $E_i$. We refer to this partition of $E_i$ as the {\em canonical} $1$-out-of-$d$ partition MMS of agent $i$, denoted by $B_i^d=\left(B_{i,1}^d,B_{i,2}^d,\ldots,B_{i,d}^d\right)$, where $B_{i,t}^d$ is the $t$-th {\em canonical bundle}
\footnote{Since every agent $i$ is only interested in relevant items in our setting, we have $v_i(\hat{P}_{i,t}) = v_i(\hat{P}_{i,t} \cap E_i) = v_i(B_{i,t})$. Therefore, it is convenient to define the canonical bundles $B_{i,t}$, which contain only edges relevant to agent $i$.} of agent $i$. When it is clear from the context, we drop all superscripts.

We also study the notion of \emph{pairwise MMS} (PMMS) fairness.\footnote{Note that PMMS does not imply EFX in the graph setting.} An allocation $X = \left(X_1,\ldots, X_n\right)$ is a {\em pairwise MMS} allocation if for every agent $i$
\[
v_i\left(X_i\right) \ge \max_{j \in V\setminus\{i\}}\mu_i^2 \left(X_i \cup X_j\right)=: \operatorname{PMMS}_i\left(X\right).
\]

Note that the PMMS threshold depends on the allocation $X$. When the allocation $X$ is clear from the context we will write $\operatorname{PMMS}_i$ instead of $\operatorname{PMMS}_i(X)$.

\section{Additive valuations}
\label{sec:additive}
In this section, we focus on multi-graphs with additive agents. The main result of this section asserts that there is always an MMS orientation in multi-graphs. We present two different algorithms that both guarantee exact MMS together with an additional property. The first algorithm computes an orientation that is both MMS and PMMS~\footnote{We remark that PMMS does not necessarily imply MMS in multi-graphs. In the Section~\ref{sec:addPMMSrelateMMS} we elaborate on the relation between MMS and PMMS and $1$-out-of-$d$ MMS in multi-graphs.}  (Theorem~\ref{thm:multiaddcutandchoose}), while the second guarantees an $1$-out-of-$3$ (and thus exact MMS for $n\ge 3$) (Theorem~\ref{thm:multiaddroundrobin}). These results are in contrast to the general (non-graphical) case where even for the case of $3$ additive agents, exact MMS allocations need not exist~\cite{KurokawaProcacciaWang18}. We also construct instances showing the limitations of our algorithms (Section~\ref{sec:addPMMSrelateMMS}).

\begin{theorem} \label{thm:multiaddcutandchoose}
    In every multi-graph with $n$ additive agents there exists an orientation which is MMS and also PMMS.
\end{theorem}
\begin{proof}
   We construct the allocation with the desired properties via an
  algorithm. We process the agents in an increasing order of their indices. For every pair of agents $i,j$ we allocate to them the set of
  their common edges $E_{i,j}$ in a cut-and-choose fashion. This step
  asserts that the final allocation is PMMS, and is also an orientation. In particular, suppose w.l.o.g. that $i<j$ then $j$
  cuts $E_{i,j}$ into two parts, such that the minimum of the two
  guarantees her the value $\mu^2_j(E_{i,j})$. Then $i$ is allocated
the most preferable of the two parts according to her valuation, denoted
  by $X_{i,j}$ and $j$ gets the remaining part, denoted by
  $X_{j,i}$. The final set $X_i$ allocated to agent $i$ is the union of all these subsets, and
  additionally any possible self-loops denoted by $X_{i,i}=E_{i,i}$,
  i.e.
  $X_i=\bigcup_{j=1}^nX_{i,j}$. 

The next claim asserts some useful properties of the way $E_{i,j}$ is split.

  \begin{claim} \label{claim:multiadd}
    Fix an agent $i$. Then, for any other agent $j\neq i$ it is either $v_i(X_{i,j})\geq \frac{1}{3}v_i(E_{i,j})$ or $ |X_{j,i}|=1.$
\end{claim}
\begin{proof} [Proof of claim \ref{claim:multiadd}]

  If $i<j$ then $j$ cuts, $i$ chooses her favorite bundle, hence $v_i(X_{i,j})\geq \frac{1}{2}v_i(E_{i,j})$. If $i>j$, then $i$ cuts and $j$ chooses the best of the two parts.
  Let's assume that $|X_{j,i}|\geq 2$, and also that $v_i(X_{i,j})<v_i(X_{j,i})$, otherwise the claim follows. By construction, it holds that $\min \left\{v_i(X_{i,j}),v_i(X_{j,i})\right\}=\mu^2_j(E_{i,j})$,  so we conclude that for any item $e\in X_{j,i}$ it must hold $\max\{v_i(\{e\}),v_i(X_{j,i}\setminus\{e\})\}\leq v_i(X_{i,j}),$
otherwise either the partition $(X_{i,j}\cup \{e\},X_{ji}\setminus\{e\})$ or the partition $(\{e\},X_{i,j}\cup X_{j,i}\setminus\{e\})$ would guarantee a greater minimum value. Therefore, we obtain that $v_i(X_{ij})\geq \frac{1}{2}v_i(X_{j,i})$ which implies that $v_i(X_{i,j})\geq \frac{1}{3}v_i(E_{i,j})$ as needed.
\end{proof}
  
We now argue that $v_i(X_i)\geq \mu_i$. Let $K_i$ be the set of agents that arrive before $i$ in the order and choose a {\em singleton} part from $i$'s partition of $E_{i,j}$, i.e., $K_i=\{j:|X_{j,i}|=1\}$. Let $G'=(V',E')$ be the multi-graph produced if we remove all agents in $K_i$ along with their allocated bundles,\footnote{Technically, any edge $(i,j)$ which was assigned to $j$ becomes a self-loop of  $j$ in $E'$.} i.e., $V'=V\setminus K_i, E'=E\setminus \left(\cup_{j\in K_i}X_j\right)$. We claim that in $G'$ the MMS guarantee of agent $i$ can only become higher, i.e., $\mu'_i\geq \mu_i$. First observe that  $v_i(X_j)=v_i(X_{j,i})$, due to the graphical property of the valuations. So since $X_{j,i}$ is a singleton, for any $j\in K_i$, we can always assume that the canonical MMS partition of agent $i$ for the original set of relevant edges $E_i$, is $B_{i,1},\ldots, B_{i,n}$ such that $X_{j,i}\subseteq B_{i,t}$ for some $t$. Therefore the set of items removed intersect with at most $\lvert K_i\rvert$ such bundles. This leaves in $E'$ at least $n'=n-\lvert K_i\rvert$ bundles $B_{i,t}$ intact, each of them guaranteeing value at least $\mu_i$.

Next we show, that if we reshuffle the items in $E'$ into $n'$
bundles, we guarantee $\mu_i'$ which can only be higher.
If $n'=1$, then $X_i=E'$, hence $v_i(X_i)\geq \mu_i$.  If $n'=2$, then
regardless of whether $i<j$ or $j<i$, $X_{ij}$ guarantees the MMS
value of the set $E^{\prime}$, hence $v_i(X_i)\geq \mu'_i\geq \mu_i$.
If $n'\ge 3$ then by using Claim~\ref{claim:multiadd} and the fact
that we have discarded all agents $j<i$ in $K_i$ who take singletons
of $E_{i,j}$, we conclude
$$v_i(X_i)\geq \sum_{j\not \in {K_i}}v_i(X_{i,j})\geq
\frac{1}{3}\sum_{j\not \in {K_i}}v_i(E_{i,j})\geq
\frac{1}{3}v_i(E')\geq \frac{1}{3}n'\mu'_i\geq \mu_i,$$ where the penultimate inequality holds due to additivity of $v_i$.
\end{proof}

The results of Theorem~\ref{thm:multiaddcutandchoose} are nontrivial. In general instances, PMMS does not imply an MMS allocation, nor does the converse hold~\cite{CaragiannisKurokawaACM19}. The implications are even stronger, as there exists an orientation that satisfies PMMS. This positive result does not hold for more general valuations.

The next theorem shows an alternative algorithm that achieves an $1$-out-of-$3$ MMS orientation which, in the case of $3$ or more agents, implies MMS.\footnote{Recall that when $d < n$, the existence of a $1$-out-of-$d$ MMS allocation is more restrictive than in general harder to achieve. Indeed, in Appendix 
we remark that the algorithm of Theorem~\ref{thm:multiaddcutandchoose} cannot guarantee $1$-out-of-$d$ MMS for any $d<n$.}

\begin{theorem} \label{thm:multiaddroundrobin}
    In every multi-graph with $n$ additive agents there exists an $1$-out-of-$3$ MMS orientation.
\end{theorem}

\begin{proof}

  We present a greedy algorithm that for $n\ge 2$ produces a $1$-out-of-$3$ MMS allocation (hence an MMS allocation for $n \ge 3$).

  The
  algorithm runs in $m$ rounds. In the first round,  we
  pick an arbitrary agent $i_1$, who selects the relevant edge with
  the maximum value for her, edge $e=(i_1,i_2)\in E^1$. Then we
  proceed with agent $i_2$ that shares this edge (who may be agent
  $i_1$ in case of a self-loop). In round $k$, we allocate to the
  selected agent $i_k$, the best available edge
  $e'=(i_k,i_{k+1})\in E^k$ i.e.,
  $e' \in \arg \max_{e \in E^k} \{v_i(e)\}$. We update the set of
  edges $E^{k+1}=E^k\setminus \{e'\}$ and proceed to the next round
  $k+1$ with agent $i_{k+1}$. If in some round $k$ the set $E^{k}_{i_k} = \emptyset$
  for the selected agent $i_k$, then we continue with an arbitrary
  agent $i'$ such that $E^{k}_{i'}\ne \emptyset$. The algorithm
  terminates after $m=\lvert E \rvert$ iterations and each edge is
  allocated exactly to one agent.

  In what follows we will show that the resulting allocation is a
  $1$-out-of-$3$ MMS allocation. Take an agent $i$ and assume
  w.l.o.g. that the algorithm selects agent $i$ for the first time in
  round $k$. Let $e^*_{i} = \arg \max_{e \in E_i}\{v_i(e)\}$ be the edge with the maximum value for agent $i$. Let $P=(A,B,C)$ be a partition of
  $E$ in three (possibly empty) parts which is defined as follows: $A$ is
  the set of edges allocated before round $k$, $B$ is the set of edges
  allocated to agent $i$ and $C$ contains the rest of the edges, i.e.,
  the edges allocated after step $k$ to agents other than $i$. It
  suffices to shown that $v_i(B) \ge \mu_i^3$.
  
  First, observe that $\lvert A \cap E_i\rvert \le 1$. Indeed, if
  agent $j$ gets an item $e=(i,j)\in A\cap E_i$ in round $k'<k$ by definition agent $i$ will be selected in the next round, hence
  $k'=k-1$. Therefore, $v_i(A) \le v_i(e^*_{i})$. Next we claim that
  $v_i(C)\le v_i(B)$. It suffices to observe that for each edge
  $e_C \in C\cap E_i$ there exists a corresponding edge $e_B \in B$,
  with higher value for $i$; $i$ selected $e_B$ over $e_C$, hence
  $v_i(e_B)\geq v_i(e_C)$. We conclude that $v_i(B)\geq v_i(C)$ by noting that all edges $e\in C\setminus E_i$ are irrelevant for $i$.

  Let the canonical $1$-out-of-$3$ MMS partition of agent $i$ be
  $B_i^3=(B_{i,1},B_{i,2},B_{i,3})$ and w.l.o.g. assume that
  $v_i(B_{i,1}) \ge v_i(B_{i,2}) \ge v_i(B_{i,3})=\mu_i^3$. Note that
  $v_i(B_{i,1}) \ge v_i(e^*_{i}) \ge v_i(A)$. Since $P$ is partition
  of $E$ and $B_i$ is partition of $E_i$ the total valuation of
  agent $i$ over all parts sums up to $v_i(E_i)$ in both cases.
 
 We conclude that $v_i(B) \ge \left(v_i(B_{i,2})+v_i(B_{i,3})\right)/{2} \ge  \mu_i^3$ as needed. \end{proof}

\subsection{Limitation of Our Algorithms and Relations Among Fair Notions in Multi-Graphs} \label{sec:addPMMSrelateMMS}
As discussed in this section, we focus on \emph{symmetric} multi-graphs~\cite{Christodoulou}, where each edge $e = (i, j)$ has the same value $w_e$ for both agents $i$ and $j$ i.e. $v_i(e) = v_j(e) = w_e$. This symmetry allows us to establish our non-existence arguments in a structured manner.

\begin{example}\label{exa:improundrobin}
    There exists a symmetric multi-graph with $n$ additive agent where the allocation derived from the algorithm in Theorem~\ref{thm:multiaddroundrobin} can not imply a $(1/2+\epsilon)$-PMMS allocation for any $\epsilon > 0$.
\end{example}
\begin{proof}
    Assume the symmetric graph $G = (V,E)$ with $n$ agents and $\lvert E \rvert = 2(n-1)+1$ be the set of edges with one large edge (we denote it with $e_L$) and $2(n-1)$ small edges (we denote them with $e_i, i \le 2(n-1)$). We will focus on two agents, let them be agent $1$ and agent $2$. Each edge has endpoints $(1,2)$ and weight $w_{e_L}=n$ for the large edge and $w_{e_i}=1/2$ for the small edges. Both agents have MMS value equal to one and the canonical partition $B_i^n=(B_{i,1},\ldots B_{i,n})$ with $B_{i,t}= \{e_t,e_{2t}\},t \le n -1$ and $B_{i,n}=\{e_L\},i \in \{1,2\}$.
    The algorithm assigns to the first agent her most valuable edge, the large edge,  and then agents will select small edges sequentially, $n-1$ small edges each agent. Let $X=(X_1,X_2,\emptyset,\ldots,\emptyset)$ be the resulting allocation. For agent $2$ her PMMS value compare to agent $1$ is $\operatorname{PMMS}_2(X) = (n-1)$ if she partitions the edges in $X_1 \cup X_2$ into two bundles, one with the large edge and the other with the remaining small edges. She achieved $v_2(X_2) = \frac{1}{2}(n-1)  < n-1$ which $\frac{1}{2}$-$\operatorname{PMMS}_2$.
\end{proof}

The following example shows that the allocation derived from the algorithm  in Theorem~\ref{thm:multiaddcutandchoose} that achieves both MMS and PMMS does not guarantee $1$-out-of-$d$ MMS for any $d<n$.
\begin{example}\label{exa:impcutandchoose}
    There exists a symmetric graph with $n$ additive agent where the allocation derived from the algorithm in  Theorem~\ref{thm:multiaddcutandchoose} can not imply $\alpha$-out-of-$d$ MMS for any $d < n$ up to any factor $\alpha > 0$.
\end{example}
\begin{proof}

  Consider the complete graph $K_n$ with unit weights. Note that each agent's MMS value is equal to zero, as there is the case of graph. However, for each agent, the $1$-out-of-$(n-1)$ MMS value is equal to one, since it is possible to form $n-1$ bundles, each containing a single relevant edge. The last agent in the selection order will partition her edges into two parts, one of which is the empty set, when she applies the cut and choose protocol with every other agent.  As she is always the last to choose, she receives the empty set. As a result, while she trivially satisfies her MMS guarantee, she cannot guarantee $1$-out-of-$d$ up to any factor for $d < n$.

\end{proof}
The following proposition shows that even on symmetric multi-graph an MMS allocation does not guarantee $\alpha$-PMMS up to any factor.
\begin{proposition}\label{prop:MMSnotimplyPMMS}
 There exists a symmetric graph with $n$ additive agent where an MMS allocation can not imply $\alpha$-PMMS up to any constant factor.
\end{proposition}
\begin{proof}
    We will construct a counter example. Assume the instance $G=(V,E)$ where $V=\{1,2,3\}$ is the set of three agents and $\lvert E \rvert = 3$ is the set of $3$ edges, two large edges (we denote them with $e_{L_i}, i \le 2$) and one small edge (we denote it with $e_S$). We will focus on two agents, agent $1$ and agent $2$. Each edge has endpoints $(1,2)$ and weight $w_{e_{L_i}}=M$ for some $M>0$  for the large edges and $w_{e_S}=1$ for the small edge.
    
    The MMS values are $\mu_1(E)=\mu_2(E)=1$ and $\mu_3(E)=0$ as agents $1$ and $2$ can make three bundle each one with one edge and all edges are irrelevant for agent $3$.
    Allocation $X=\left(X_1,X_2,X_3\right)=\left(e_S,E \setminus \{e_S\},\emptyset\right)$ is MMS but is not PMMS as $
    \operatorname{PMMS}_1=\mu_1^3\left(X_1 \cup X_2\right) = M > v_1\left(X_1\right) = 1$ and hence it is $\frac{1}{M}$-PMMS. For any $\epsilon > 0$ there exists an arbitrary large $M$ such that we cannot guarantee more than $\epsilon$-PMMS. 
\end{proof}

\section{XOS valuations}
\label{sec:XOS}
In this section, we study multi-graphs with agents having XOS valuations. First, we provide our tools which are of independent interest. Then, we study the case of a few agents as a warm up and show the existence of a $(1-1/d)$-out-of-$d$ MMS orientation for $n \le 3$ (Theorem~\ref{thm:twothreeXOS}). Our main result demonstrates the existence of a $2/3$-MMS orientation in multi-graphs with three or more agents (Theorem~\ref{thm:XOS4+}). This separates the graphical model from the non-graphical model, for which a $1/2$ upper bound is known~\cite{GhodsiHSSY22}. Next, in Section~\ref{sec:SBXOS}, we establish asymptotically improved bounds $2/3$-out-of-$\mathcal{O(\sqrt{n}})$ (Theorem~\ref{thm:sqrtXOS}) by utilizing recent results from \cite{DaiLiuZhang2025}.

Furthermore, in the same paper, the authors state a conjecture for the existence of independent transversal sets as a function of the average degree of the graph. The main result in Section~\ref{sec:XOSITS}, (Theorem~\ref{thm:connectionITS}) states that, if their conjecture is proved true, for every $1-1/k$ with $k \ge 2$ there exists a positive integer $n_k$ such that for all $n \ge n_k$ the instance with $n$ XOS agents with graphical valuations admits a $(1-1/k)$-MMS orientation. Moreoever, if there exists an upper bound of $\alpha$ such that for every $n$ the instance with $n$ XOS agents with graphical valuations does not admit a $\alpha$-MMS allocation, then their conjecture is falsified. That is, we show a strong connection between MMS allocations for XOS agents in multi-graphs and independent transversal sets, greatly enhancing our understanding of the former via tools from graph theory.

We complement this positive result with impossibility results (upper bounds) in Section~\ref{sec:upper-bound-xos} for any number of agents and any number of bundles. Specifically, we prove upper bounds of $2/3$-MMS for three and four agents and for more agents, we show an upper bound of an $\left(1-\frac{1}{\left\lceil n/2\right\rceil+1}\right)$-MMS (Theorem~\ref{thm:XOSupper}). Moreover, in the same theorem, we provide upper bounds for two and three agents with $d$ bundles matching the corresponding lower bounds. For both the upper and lower bounds, we employ tools from combinatorics, particularly the theory of transversal independent sets in multipartite graphs.

We also focus on $\alpha$-out-of-$d$ MMS approximations for small values of $d$, independent of $n$. Analogous to the $1$-out-of-$3$ result for additive valuations (Theorem~\ref{thm:multiaddroundrobin}), we show that a modified version of the algorithm used in that proof can be adapted for XOS agents, achieving a $1/2$-out-of-$2$ MMS for $n$ agents (Theorem~\ref{thm:XOS1outofd}). This guarantee is optimal for the MMS notion in the graphical setting Theorem~\ref{thm:XOSupper}.
 
For the PMMS notion, we show the existence of a tight $1/2$-PMMS orientation. The same result also holds for subadditive valuations so its presentation is deferred to Section~\ref{sec:subadditive} i.e. the lower bound holds for the general class of subadditive valuations while the impossibility result of the upper bound is shown by construction with XOS valuations. We note that in our construction for the upper bound there exists a PMMS allocation while we cannot guarantee more than $(1/2+\varepsilon)$-PMMS orientation, showing that the fair notion is wasteful i.e. in order to guarantee fairness we must allocate the items in a suboptimal way.

\subsection{Our Tools}\label{sec:tools-XOS}

In this section, we present two lemmas and a definition that are essential to our analysis. We believe that our approach is of independent interest and may find its applications in future work. The following definition introduces a special type of partial allocation, which we call \textit{frugal}, as it assigns each agent subsets from only one of the $1$-out-of-$d$ MMS bundles; it never combines items from different MMS bundles. 

\begin{definition}\label{def:frugal}
Let ${\bf B}^d=\left(B_1^d,\ldots,B_n^d\right)$ be the vector of the canonical $1$-out-of-$d$ MMS partitions. An allocation $X$ is said to be \emph{frugal} with respect to ${\bf B}^d$  if for all $i$ there exists a $t_i\leq d$ such that $X_i \subseteq B_{i,t_i}^d$.  
\end{definition}

For simplicity, we refer to an allocation as frugal (omitting $\mathbf{B}^d$), with the understanding that this refers to the canonical partition vector $\mathbf{B}^d$. A frugal allocation may be partial; however, assigning the remaining edges arbitrarily (e.g., orienting them randomly) does not affect the approximation factor of the allocation, as $1$-out-of-$d$ MMS fairness guarantees are preserved due to monotonicity. 

The following normalization approach, which has been noted in the literature for MMS and XOS valuations (see, e.g., \cite{GhodsiHSSY22}), allows us to assume (by appropriate rescaling) that an agent's $1$-out-of-$d$ MMS bundles all have equal value. We provide a proof sketch for completeness.

\begin{lemma}\label{lem:normalized}
If there exists an $\alpha$-out-of-$d$ MMS allocation for normalized graphical XOS valuations, where $v_i(B_{i,t})=1$ and $v_i(S) \le 1, \forall S \subseteq E$ for each agent $i$ and each canonical bundle $B_{i,t}$, then there exists an $\alpha$-out-of-$d$ MMS allocation for graphical XOS valuations.
\end{lemma}

\begin{proof}
  Let $v_1,\ldots, v_n$ be the valuations of the agents with $v_i$
  being the valuation of agent $i$ and
  ${B}_i^d=(B_{i,1},\ldots,B_{i,d})$ being her canonical
  $1$-out-of-$d$ MMS bundles. We will define scaled valuations $v_i''$ such that
  $v_i''(B_{i,t}) = 1$ for all $i, t$. Moreover for every subset $S\subseteq E$ whenever
  $v_i''(S) \ge \alpha$ in the scaled instance, then it must be $v_i(S) \ge \alpha \mu_i^d(E)$ in the original instance.
  
We can first scale the valuations (by $\mu_i^d(E)$), such that $\mu'^d_i(E)=1$ and hence  $v'_i(B_{i,t})\geq 1$ for all $i, t$. 
  
Then by use of Lemma 5.2 in \cite{GhodsiHSSY22} we can define the capped valuation $v_i''\left(S\right)=\min\left\{1,v_i'\left(S\right)\right\}$,$\forall S \subseteq E$ and maintaining the XOS property. From Observation 5.2 in \cite{GhodsiHSSY22} we have $v''_i\left(S\right)\le v_i'\left(S\right),\forall S \subseteq E$ and hence
$\mu''^d_i \le \mu'^d_i=1$. As a result, for the canonical bundles in $B_{i}^d$ we obtain $v_i''(B_{i,t})=1$ for all $j$ and $\mu''^d_i=1$. Thus, the existence of $\alpha$-out-of-$d$ MMS allocation $X = (X_1, \ldots,X_n)$ for the normalized values $v''_i$ implies the existence of  $\alpha$-out-of-$d$ MMS for the original instance with XOS agents i.e. $\alpha \le v_i''(X_i) \le v_i'(X_i)$.\end{proof}

In light of Lemma~\ref{lem:normalized}, we assume for the rest of this
section that all values are normalized; that is, $\mu_i^d = 1$ for all
$i \in V$, and for the canonical $1$-out-of-$d$ MMS partition
$B_i^d = (B_{i,1}, \ldots, B_{i,d})$ of agent $i$, we have
$v_i(B_{i,t}) = 1,t \le d$.

The next Lemma establishes a useful property of frugal allocations for agents with XOS valuations, and its applicability extends beyond graphical instances. The Lemma state that it suffices to assume additive valuation functions and frugal allocations for XOS agents (in general).

\begin{lemma} \label{lem:XOSadditive} 
The instance with $n$ XOS agents admits an $\alpha$-out-of-$d$ MMS allocation if and only if the instance with $n$ additive agents admits a \textit{frugal} $\alpha$-out-of-$d$ MMS allocation.
\end{lemma}

\begin{proof}
We will first show that for every  instance (and thus also in a graph) with $n$ agents and XOS valuations $v_i$, there exists an equivalent instance with the same items  (i.e.  in the same graph for graphical valuations) with additive functions $v'_i$ and $\mu'_i = \mu_i$ such that for every frugal allocation $X = (X_1, \ldots, X_n)$, we have $v'_i(X_i) \le v_i(X_i)$. Thus, a frugal $\alpha$-MMS allocation in the equivalent instance (with additive agents) yields an $\alpha$-MMS allocation in the original instance (with XOS agents). Then we will show that given an instance with $n$ agents and additive valuations $v_i$, if there exists an upper bound of $\alpha$-out-of-$d$ MMS for frugal allocations then  there exists an equivalent instance with the same items with XOS functions $v'_i$ and $\mu'_i = \mu_i$ such that for every allocation $X = (X_1, \ldots, X_n)$, we have $v'_i(X_i) \le \alpha \mu_i'^d$.

Consider an instance with $m$ items and $n$ agents with XOS valuation functions $v_1, \ldots, v_n$. By definition of XOS, the valuation of each agent $i$ is $v_i(S) = \max_{j \le k}\{a_j(S)\}$, for all $S \subseteq E$, where $a_j$ is a valuation function from a set of $k$ additive functions $a_1, \ldots, a_k$. By Lemma~\ref{lem:normalized}, we can assume that every canonical bundle has a value of 1. We now define additive valuation functions $v'_1, \ldots, v'_n$.

Let $B_i^d = (B_{i,1}, \ldots, B_{i,d})$ be the canonical $1$-out-of-$d$ MMS partition of agent $i$, and let $f^*_{i,t}$ be the additive function corresponding to the bundle $B_{i,t}$. Therefore $v_i(B_{i,t}) = \max_{j \le k}\{a_j(B_{i,t})\} = f^*_{i,t}(B_{i,t}) = 1$.

Observe that every item $e$ belongs to exactly one bundle $B_{i,t_e}$ for some index $t_e$, and thus, we can define the additive function $v'_i$ as follows:
\[
v'_i(e) = f^*_{i,t_e}(e), \quad e \in B_{i,t_e}\]
and
\[
v'_i(S) = \sum_{e \in S} v'_i(e), \quad S \subseteq E.
\]

We now claim that $\mu'^d_i = 1$. To see this, observe that $\mu'^d_i \ge \mu_i^d = 1$; that is, if we consider the same canonical bundles for agent $i$ in the additive setting $B_i^d = (B_{i,1}, \ldots, B_{i,d})$, we achieve a value of $v'_i(B_{i,t}) = 1$ for all $t$. Moreover, due to additivity, we have $v_i'(E) = \sum_{t=1}^d v_i'(B_{i,t}) \ge d \min_t\{v_i'(B_{i,t})\} = d\mu'^d_i$, and thus, $\mu'^d_i \le v_i'(E)/d = 1$.

Now consider a frugal $\alpha$-out-of-$d$ MMS allocation $X = (X_1, \ldots, X_n)$, where $X_i \subseteq B_{i,t_i}$ for some $t_i$. We conclude that $\alpha \le v'_i(X_i) = f^*_{i,t_i}(X_i) \le \max_j\{a_j(X_i)\} = v_i(X_i).$

Assume now an instance with $m$ items and $n$ agents with additive valuation functions $v_1, \ldots, v_n$ in whicht here is not $(\alpha + \varepsilon)$-out-of-$d$ MMS frugal allocation. We will define the XOS valuation functions $v'_1, \ldots, v'_n$.

Let $B_i^d = (B_{i,1}, \ldots, B_{i,d})$ be the canonical $1$-out-of-$d$ MMS partition of agent $i$, and let $a_{i,k}$ be an additive function as follows
$$a_{i,t}(e)=\begin{cases}
    v_i(e), & e \in B_{i,t} \\
    0, & \text{otherwise}
\end{cases}$$

For each agent $i$ we denote the XOS function $v_i'(S)=\max_{j \le d}\{a_{i,j}(S)\}$, for all $S \subseteq E$.

We now claim that $\mu'^d_i = \mu_i^d$. To see this, observe that $\mu'^d_i \ge \mu_i^d$; that is, if we consider the same canonical bundles for agent $i$ in the additive setting $B_i^d = (B_{i,1}, \ldots, B_{i,d})$, we achieve a value of $v'_i(B_{i,t}) \ge \mu_i^d$ for all $t$. Moreover the partition $B_i^d = (B_{i,1}, \dots, B_{i,d})$ is the unique $1$-out-of-$d$ MMS partition of agent $i$, with $\mu_i'^d =\mu_i^d$. To see this, suppose there exists another $1$-out-of-$d$ MMS partition $X = (X_1, \dots, X_d)$ such that $v_i(X_{t_X}) \ge \mu_i^d$ for every $t_X$, but $X$ is not a reordering of $B_i$. Then, for every $t_X \in \{1, \dots, d\}$, there exists a $B_{i,t}$ such that $X_{t_X} \supseteq B_{i,t}$ with the containment being strict for at least some $t^*$. This is impossible.

Now assume for the shake of contradiction that an $(\alpha+\varepsilon)$-out-of-$d$ MMS allocation $X = (X_1, \ldots, X_n)$ exists and thus for each $i$ there exists an additive function $f^*_{i,t_i}=v_i(X_i) \ge \alpha+\varepsilon$. Hence, allocation $X'=(X_1 \cap B_{1,t_1}, \ldots,X_n \cap B_{n,t_n})$ is frugal and moreover by construction $\alpha + \varepsilon \le v_i'(X_i)=f_{i,t_i}^*(X_i \cap B_{i,k_i})=v_i(X'_i)$ which contradicts the fact that there is not $(\alpha + \varepsilon)$-out-of-$d$ MMS frugal allocation in the original instance with the additive functions.
\end{proof}

Due to Lemma~\ref{lem:XOSadditive}, it suffices from now on to assume additive valuation functions and frugal allocations.

\subsection{Two and Three Agents}\label{sec:two-three-XOS}
As a warm up, for the case of two
\footnote{We note that for two agents this result was already known for $d=2$ and for subadditive values~\cite{christodoulouIJCAI25}.}
and three agents we show that there always exists a frugal $(1-1/d)$-out-of-$d$ MMS allocation, for any $d\geq 1$. This is used as a building block for the proof of the main result for many agents. In Section~\ref{sec:upper-bound-xos} we show that the results of Theorem~\ref{thm:twothreeXOS} are tight.

\begin{theorem}\label{thm:twothreeXOS}
      In every multi-graph with $n \le 3$ XOS agents there exists a frugal $(1-1/d)$-out-of-$d$ MMS orientation. 
\end{theorem}

\begin{proof}

We will break down the proof into two lemmas, one for the case of two agents and one for the case of three agents. Combining the two lemmas we prove the desired theorem.
\begin{lemma} \label{lem:2XOSdlower}
    In every multi-graph with $2$ XOS agents there exists a frugal $(1-1/d)$-out-of-$d$ MMS allocation.
\end{lemma}
\begin{proof}

  We focus on the partitioning of agent
  $1$'s canonical bundle $B_{1,1}$ into $d$ disjoint sets according
  to $B^d_2=(B_{2,1},\ldots,B_{2,d})$, the canonical $1$-out-of-$d$ MMS partition of
  agent $2$. Due to Lemma~\ref{lem:normalized} and Lemma~\ref{lem:XOSadditive} we assume that agent $1$ has additive value for this bundle and its value is equal to $1$, therefore 
  there exists at least one canonical bundle of agent $2$, let it be $B_{2,t^*}$, such that
  $v_1(B_{1,1} \cap B_{2,t^*}) \le 1/d$. Due to additivity, we obtain $v_1(B_{1,1} \setminus B_{2,t^*}) \ge (d-1)/d$. We output the allocation
  $X=(B_{1,1} \setminus B_{2,t^*},B_{2,t^*})$ which is is frugal and
  $(1-1/d)$-out-of-$d$. \end{proof}
\begin{lemma}\label{lem:3XOSdlower}
      In every multi-graph with $3$ XOS agents there exists a frugal $(1-1/d)$-out-of-$d$ MMS allocation.
\end{lemma}
\begin{proof}
  We will show that for every $d\geq 1$ there are indices
  $t_1,t_2,t_3 \leq d$ for agents $1, 2, 3$ respectively such that there exists a frugal allocation $X$ where $v_i(X_i) \ge (1-1/d)\mu_i^d$, with respect to $\mathbf{B}^d=(B_1^d,B_2^d,B_3^d)$  with $X_{i} \subseteq B_{i,t_i}$, for $i=1,2,3$.

  Intuitively, a given configuration of indices $t_1, t_2, t_3$ is a
  good candidate frugal solution if the demand on the intersections
  $B_{1,t_1}\cap B_{2,t_2}$, $B_{2,t_2}\cap B_{3,t_3}$ and
  $B_{1,t_1}\cap B_{3,t_3}$ of the canonical bundles satisfy
  certain properties. We call an intersection
  $B_{i,t_i}\cap B_{j,t_j}$ {\em large} for agent $i$, if
  $v_i\left(B_{i,t_i}\cap B_{j,t_j}\right)> 1/d$, otherwise we
  call it {\em small}. For each such intersection, only two agents are potentially interested, so if it is small for one of them, then they are happy to abandon one small intersection, as the value of the remaining bundle is higher than $1-1/d$.

  We will show by downward induction, that if there is a pair of
  agents $i, j$ and an index $t_i$ such that $i$ has $k$ or more
  large intersections of bundle $B_{i,t_i}$ with the canonical bundles
  of $j$, then there exists a frugal $(1-1/d)$-out-of-$d$ MMS allocation.
  Notice that by additivity (Lemma~\ref{lem:XOSadditive}) there exist at most {\em
    at most} $d-1$ large intersections for any such instantiation, i.e. $k\leq d-1$.

  \noindent{\bf Base case, $\bf k=d-1$}.  Assume that there exist two agents
  (say $1, 2$), and some bundle (say $B_{1,1}$) with $d-1$ large
  intersections for $1$ with the bundles of agent $2$ let them be $B_{2,1},\ldots, B_{2,d-1}$. Then, by additivity $v_1(B_{1,1} \cap B_{2,d}) \le \frac{1}{d}$. By similar arguments, regarding bundle $B_{2,d}$ and agent $2$, there
  must exist some bundle of agent $3$, say $B_{3,1}$, such that
  $v_2(B_{2,d} \cap B_{3,1}) \le \frac{1}{d}$. As a result the
  allocation $X=(B_{1,1} \cap(B_{2,1} \cup \ldots\cup B_{2,d-1}), B_{2,d} \setminus B_{3,1},B_{3,1})$ is
  $(1-1/d)$-out-of-$d$ and frugal and the lemma follows i.e. the lemma follows for $t_1=1,t_2=d$ and $t_3=1$. Figure~\ref{fig:all1} illustrates such an allocation for $d=3$.

  \begin{figure}[htbp]
   \begin{tabular}{ccc}
\begin{tikzpicture}[scale=0.6]
\small{
  \draw[step=1cm,] (0,0) grid (3,3);
\draw[pattern={north west lines},pattern color=red]
(2,0) rectangle +(1,3);
\draw[pattern={horizontal lines},pattern color=blue](0,2) rectangle +(2,1);

\node[anchor=north] at (0.3,4) {$B_{2,1}$};
\node[anchor=north] at (1.4,4) {$B_{2,2}$};
\node[anchor=north] at (2.5,4) {$B_{2,3}$};
\node[anchor=west] at (-1.3,0.5) {$B_{1,3}$};
\node[anchor=west] at (-1.3,1.5) {$B_{1,2}$};
\node[anchor=west] at (-1.3,2.5) {$B_{1,1}$};
\node[anchor=west] at (0.8,-0.5) {$E_{1,2}$};
\node[anchor=west] at (1,-1.2) {(a)};}
\end{tikzpicture} &\begin{tikzpicture}[scale=0.6]
\small{
\draw[pattern={crosshatch dots},pattern color=green](0,0) rectangle +(1,3);
  \draw[step=1cm,] (0,0) grid (3,3);

\node[anchor=north] at (0.3,4) {$B_{3,1}$};
\node[anchor=north] at (1.4,4) {$B_{3,2}$};
\node[anchor=north] at (2.5,4) {$B_{3,3}$};
\node[anchor=west] at (-1.3,0.5) {$B_{1,3}$};
\node[anchor=west] at (-1.3,1.5) {$B_{1,2}$};
\node[anchor=west] at (-1.3,2.5) {$B_{1,1}$};
\node[anchor=west] at (0.8,-0.5) {$E_{1,3}$};
\node[anchor=west] at (1,-1.2) {(b)};}
\end{tikzpicture} &
\begin{tikzpicture}[scale=0.6]
\small{
  \draw[step=1cm,] (0,0) grid (3,3);
\draw[pattern={crosshatch dots},pattern color=green](0,0) rectangle +(1,3);
\draw[pattern={north west lines},pattern color=red](1,0) rectangle +(2,1);
\node[anchor=north] at (0.3,4) {$B_{3,1}$};
\node[anchor=north] at (1.4,4) {$B_{3,2}$};
\node[anchor=north] at (2.5,4) {$B_{3,3}$};
\node[anchor=west] at (-1.3,0.5) {$B_{2,3}$};
\node[anchor=west] at (-1.3,1.5) {$B_{2,2}$};
\node[anchor=west] at (-1.3,2.5) {$B_{2,1}$};
\node[anchor=west] at (0.8,-0.5) {$E_{2,3}$};
\node[anchor=west] at (1,-1.2) {(c)};
}
\end{tikzpicture}\\
    \end{tabular}
  \caption{Common edges $E_{1,2}$ (a), $E_{1,3}$ (b) and $E_{2,3}$ (c). In each set $E_{i,j}$ rows represent the bundles $B_{i,t}$ of agent $i$ and columns represent the bundles $B_{j,t}$ of agent $j$. Cell $(t_i,t_j)$ represents the intersection $B_{i,t_i} \cap B_{j,t_j}$. The figure illustrates the frugal orientation $X=(B_{1,1}\cap(B_{2,1} \cup B_{2,2}),B_{2,3}\setminus B_{3,1},B_{3,1})$ of the base case with $d=3$. With blue, red and green being the bundles allocated to agent $1,2$ and $3$ respectively.}
  \label{fig:all1}
\end{figure}

  \noindent{\bf Induction step}. Assume now that there do not exist
  $i,j$ and $t_i$ such that $i$ has strictly more than $k$ large
  intersections of  $B_{i,t_i}$ with the canonical bundles of
  $j$, or otherwise there is a desired allocation. We will show that
  this is also the case for $k$.
  
  Therefore w.l.o.g. let's assume that for $B_{1,1}$ of agent $1$ there are $k$
  large intersections with the bundles of agent $2$ say the first $k$
  bundles $B_{2,1},\ldots, B_{2,k}$.  Agent $1$ would be
  happy to abandon any single intersection $B_{1,1}\cap B_{2,k+l}$ with
  $l=1,\ldots, d-k$.  Now we focus on those bundles of agent $2$ and
  consider the intersections with agent $3$. The following claim
  asserts that there are sufficiently many good pairs of distinct indices that agent 2
  considers small.

  \begin{claim}\label{claim:hall}
    There exist $d-k$ distinct set of indices $r_l, l=1,\ldots, d-k$, such that each intersection $B_{2,k+l}\cap B_{3,r_l}$ is small for agent $2$ for $l=1,\ldots, d-k$.
  \end{claim}
  \begin{proof}
  Consider a bipartite graph $G_b=(U_b,V_b,E_b)$ with parts
  $U_b =\{u_{1},\ldots,u_{k-d}\}$ and $V_b = \{v_1,\ldots v_d\}$. A
  vertex $u_l\in U_b$ corresponds to the bundle $B_{2,k+l}$ of agent $2$
  while a vertex $v_t\in V_b$ corresponds to the bundle $B_{3,t}$ of
  agent $3$. An edge $(u_{l},v_{t})$ exists if and only if the
  corresponding intersection is small for agent $2$, that is
  $v_2\left(B_{2,k+l} \cap B_{3,t}\right) \le 1/d$. By the induction
  hypothesis, for each bundle $B_{2,k+l}$ there exists at most $k$
  large intersections with the bundles of agent $3$ or else there exists a $(1-1/d)$-out-of-$d$ MMS
  allocation. Therefore, there are at least $d-k$ small intersections 
  of $B_{2,k+l}$ for agent $2$ with the bundles of agent $3$. As a result, the degree of each vertex
  in $U_d$ is at least $d-k$. By  Hall's Theorem there exists a matching of size $d-k$. Therefore there exist $d-k$
  pairs with distinct elements and the induction hypothesis holds also
  for $k$. \end{proof}

We note that all these $d-k$ pairs guaranteed to exist by
Claim~\ref{claim:hall} are good for agent $2$, i.e. agent $2$ would be
happy to abandon any one of the respective intersections with agent $2$. To
complete the proof, we need to argue that out of all these $d-k$ pairs of indices,
there is a pair $(k+l,r_l)$ for which agent 1 can afford to abandon the subset $S_l:=B_{1,1}\cap
\left(B_{2,k+l} \cup B_{3,r_l}\right)$ as it has value of at most $1/d$.

First, observe that the $S_l$'s are mutually disjoint for all
$l=1,\ldots, d-k$. Since agent $1$ has exactly $k$ large intersections
with agent $2$, then there are at most $d-k-1$ such $S_l$'s with
$v_i(S_l)> 1/d$ and therefore there is at least one $S_{l^*}$ with
value of at most $1/d$.

Hence, the allocation $X=(B_{1,1}\setminus \left(B_{2,k+l^*}\cup B_{3,r_{l^*}}\right),B_{2,k+l^*}\setminus B_{3,r_{l^*}},B_{3,r_{l^*}})$  is $(1-1/d)$-out-of-$d$ i.e. the lemma follows for $t_1=1,t_2=k+l^*$ and $t_3=r_{l^*}$
Figure~\ref{fig:all2} illustrates such an
allocation for $d=3$ and $l^*=3$.
\begin{figure}[htbp]
    \begin{tabular}{ccc}
\begin{tikzpicture}[scale=0.6]
\small{
  \draw[step=1cm,] (0,0) grid (3,3);
\draw[pattern={north west lines},pattern color=red]
(2,0) rectangle +(1,3);
\draw[pattern={horizontal lines},pattern color=blue](0,2) rectangle +(2,1);

\node[anchor=north] at (0.3,4) {$B_{2,1}$};
\node[anchor=north] at (1.4,4) {$B_{2,2}$};
\node[anchor=north] at (2.5,4) {$B_{2,3}$};
\node[anchor=west] at (-1.3,0.5) {$B_{1,3}$};
\node[anchor=west] at (-1.3,1.5) {$B_{1,2}$};
\node[anchor=west] at (-1.3,2.5) {$B_{1,1}$};
\node[anchor=west] at (0.8,-0.5) {$E_{1,2}$};
\node[anchor=west] at (1,-1.2) {(a)};}
\end{tikzpicture} &\begin{tikzpicture}[scale=0.6]
\small{
\draw[pattern={horizontal lines},pattern color=blue](0,2) rectangle +(1,1);
\draw[pattern={horizontal lines},pattern color=blue](2,2) rectangle +(1,1);
\draw[pattern={crosshatch dots},pattern color=green](1,0) rectangle +(1,3);
  \draw[step=1cm,] (0,0) grid (3,3);

\node[anchor=north] at (0.3,4) {$B_{3,1}$};
\node[anchor=north] at (1.4,4) {$B_{3,2}$};
\node[anchor=north] at (2.5,4) {$B_{3,3}$};
\node[anchor=west] at (-1.3,0.5) {$B_{1,3}$};
\node[anchor=west] at (-1.3,1.5) {$B_{1,2}$};
\node[anchor=west] at (-1.3,2.5) {$B_{1,1}$};
\node[anchor=west] at (0.8,-0.5) {$E_{1,3}$};
\node[anchor=west] at (1,-1.2) {(b)};}
\end{tikzpicture} &
\begin{tikzpicture}[scale=0.6]
\small{
  \draw[step=1cm,] (0,0) grid (3,3);
\draw[pattern={crosshatch dots},pattern color=green](1,0) rectangle +(1,3);
\draw[pattern={north west lines},pattern color=red](0,0) rectangle +(1,1);
\draw[pattern={north west lines},pattern color=red](2,0) rectangle +(1,1);
\node[anchor=north] at (0.3,4) {$B_{3,1}$};
\node[anchor=north] at (1.4,4) {$B_{3,2}$};
\node[anchor=north] at (2.5,4) {$B_{3,3}$};
\node[anchor=west] at (-1.3,0.5) {$B_{2,3}$};
\node[anchor=west] at (-1.3,1.5) {$B_{2,2}$};
\node[anchor=west] at (-1.3,2.5) {$B_{2,1}$};
\node[anchor=west] at (0.8,-0.5) {$E_{2,3}$};
\node[anchor=west] at (1,-1.2) {(c)};}

\end{tikzpicture}\\
    \end{tabular}
  \caption{Common edges $E_{1,2}$ (a), $E_{1,3}$ (b) and $E_{2,3}$ (c). In each set $E_{i,j}$ rows represent the bundles $B_{i,t}$ of agent $i$ and columns represent the bundles $B_{j,t}$ of agent $j$. Cell $(t_i,t_j)$ represents the intersection $B_{i,t_i} \cap B_{j,t_j}$. The figure illustrates the frugal orientation $X=(B_{1,1}\setminus(B_{2,3} \cup B_{3,2}), B_{2,3}\setminus B_{3,2},B_{3,2})$ of the induction step $k=1$, with $d=3$ and $l^*=3$. With blue, red and green being the bundles allocated to agent $1,2$ and $3$ respectively.}
  \label{fig:all2}
\end{figure}
\end{proof}

The proof of Theorem~\ref{thm:twothreeXOS} follows from Lemmas~\ref{lem:2XOSdlower} and \ref{lem:3XOSdlower}.
\end{proof}
\subsection{Many Agents}\label{four-XOS}

Next, we present our main result, asserting that there always exists a $2/3$-MMS allocation for $n\ge 3$.

\begin{theorem}\label{thm:XOS4+}
    In every multi-graph with $n \ge 3$ XOS agents there exists a frugal $2/3$-MMS orientation.
\end{theorem}
\begin{proof}
    Let $G = (V, E)$ be a graph with $n \ge 3$ vertices, and let $\mathbf{B}_i^n = (B_{i,1}, \ldots, B_{i,n})$ denote the canonical MMS partition of agent $i$ for the edge set $E_i$. For some $k \leq n$, let $\mathbf{I}_i(k) \subseteq \{1, 2, \ldots, n\}$ be a subset of $k$ indices, i.e., $|\mathbf{I}_i(k)| = k$. We define $B_i(\mathbf{I}_i(k))$ as the set of the respective bundles $B_{i,t}$ with $t \in \mathbf{I}_i(k)$ i.e. $B_i(\mathbf{I}_i(k))= \{B_{i,t} \mid t\in \mathbf{I}_i(k)\}$.

    We employ Lemma~\ref{lem:XOSadditive}, which asserts that it
    suffices to show that there exists a $2/3$-MMS {\em frugal}
    allocation in every multi-graph with $n \geq 3$ {\em additive}
    agents. We will show by induction on $k\in \{3,\ldots, n\}$, that
    for every subset $\mathbf{P}=\{p_1,\ldots,p_k\} \subseteq V$ of
    $k$ agents and for any fixed collection of $k$ sets $\mathbf{I}_{p_1}(k),\ldots, \mathbf{I}_{p_k}(k)$, of indices of size $k$ there exists a
    frugal orientation $X$ with respect to
    $\mathbf{B}^k=(B_{p_1}(\mathbf{I}_{p_1}(k)),
    \ldots,B_{p_k}(\mathbf{I}_{p_k}(k)))$ such that
    $v_{p_i}(X_{p_i}) \ge \frac{2}{3}\min_{t \in
      \mathbf{I}_{p_i}(k)}\{v_{p_i}(B_{p_i,t})\}\ge
    \frac{2}{3},\forall p_i\in \mathbf{P}$. Note that every frugal
    allocation with respect to bundles $B_{p_i}(\mathbf{I}_{p_i}(k))$
    for some agent $p_i \in \mathbf{P}$ and set of $k$ indices
    $\mathbf{I}_{p_i}(k)$ is also frugal with respect to the canonical
    partition $B_{p_i}^n$. Furthermore, observe that for $k=n$ we have $\boldsymbol{P} =V$. Thus for each agent $i\in V$ we have $v_{i}(X_{i}) \ge \frac{2}{3}\min_{t \in
      \mathbf{I}_{i}(n)}\{v_{i}(B_{i,t})\}\ge
    \frac{2}{3}V$ and the theorem follows.

    By Theorem~\ref{thm:twothreeXOS} we have already established the base
    case of $n=3$. Now, assume that the statement holds for all
    subsets $\mathbf{P}=\{p_1,\ldots,p_k\}\subset V$ of $k$
    agents and for all corresponding vectors of $k$ subsets of $k$ indices
    $\mathbf{I}_{p_1}(k),\ldots, \mathbf{I}_{p_k}(k)$, with $3\leq k <n$. We now show
    that the statement also holds for all possible subsets
    $\mathbf{P}=\{p_1,\ldots,p_{k+1}\}\subseteq V$ of $k+1$ agents and
    for all possible vectors of $k+1$ subsets of indices
    $\mathbf{I}_{p_1}(k+1),\ldots,
    \mathbf{I}_{p_{k+1}}(k+1)$. W.l.o.g. we will establish the induction step for subset 
     $\mathbf{P}=\{1,\ldots,k+1\}$ and set of indices
    $\mathbf{I}_i(k+1)=\{1,\ldots, k+1\}$, for all
    $i \in \{1,\ldots,k+1\}$; then, clearly, the statement holds for any other
    subset of agents and any other vector of indices of size $k+1$ by appropriate renaming of agents and indices.
    
    We divide the proof of the induction step into two key claims. Claim~\ref{clm:XOS1}, using the inductive
    hypothesis, outlines the necessary conditions for the case of
    $k+1$ agents when no $2/3$-MMS frugal orientation
    exists. Definition~\ref{def:overconstrained} summarizes these
    conditions in terms of what we call an \textit{overconstrained
      set}. Claim~\ref{clm:XOS2}  shows that overconstrained sets reveal useful structural properties in the intersections of canonical bundles, and uses this structure to construct a
    $2/3$-MMS orientation. Due to space limitations, we moved the proof to Appendix. 
    We now proceed with the definition of an overconstrained set.
\begin{definition}[Overconstrained set]
    \label{def:overconstrained}
Let  $\mathbf{P}=\{1,\ldots,k+1\}$ be a set of $k+1$ agents and $\mathbf{I}_i(k+1)=\{1,\ldots, k+1\},$ with $i \in \mathbf{P}$ be $k+1$ subsets of indices. We say that set ${\bf P}$ is {\em overconstrained} if for every subset ${\bf P} \setminus \{i\}$ of $k$ agents, there exist two $2/3$-MMS frugal orientations $X^{(i)}, X'^{(i)}$ such that $$X^{(i)}_j \cap X'^{(i)}_j = \emptyset,\forall j \in \mathbf{P}\setminus \{i\}\text{ and}$$ 
\begin{align}\label{eq:XOSstructure}
v_i\left(B_{i,t_i} \cap \left(X^{(i)} \cup X'^{(i)}\right)\right) \geq 2/3, \quad  t_i\in\{1,\ldots, k+1\}.
\end{align}
\end{definition}

\begin{claim} \label{clm:XOS1}
Let  $\mathbf{P}=\{1,\ldots,k+1\}$ be a set of $k+1$ agents and $\mathbf{I}_i(k+1)=\{1,\ldots, k+1\},$ with $i \in \mathbf{P}$ be $k+1$ subsets of indices.  Either (i) there exists a frugal $2/3$-MMS orientation for agents in ${\bf P}$, or (ii) ${\bf P}$ is overconstrained.
\end{claim}

\begin{proof}

  If (i) holds we are done. Otherwise, using the inductive hypothesis, for each agent $i$ we will identify two orientations $X^{(i)}$ and $X'^{(i)}$ for ${\bf P}\setminus\{ i\}$.  Since $(i)$ does not hold, we conclude that we cannot extend any of the existing orientations $X^{(i)}$ or $X'^{(i)}$ to provide a $2/3$-MMS frugal allocation for ${\bf P}\setminus\{ i\}$, establishing that ${\bf P}$ is {\em overconstrained}.

To show this, we focus on agent $k+1$ and we carefully define a partition of her first MMS bundle $B_{k+1,1}$ 
in three parts
$\left(A_{k+1,1}^1, A_{k+1,1}^2, A_{k+1,1}^3\right).
$  
This partition has the property that if at least one of the first two parts has value at most $2/3$ for agent $k+1$ i.e. either $v_{k+1}\left(A_{k+1,1}^1\right) \leq \frac{1}{3}$ or $v_{k+1}\left(A_{k+1,1}^2\right) \leq \frac{1}{3}$, we can extend an existing allocation—guaranteed to exist by the inductive hypothesis—for the set of $k$ agents and $k$ bundles to the set of $k+1$ agents and $k+1$ bundles. If this is not the case, we will show that inequality~(\ref{eq:XOSstructure}) holds for $i=k+1$ and $t_i=1$. We can proceed in a similar fashion by defining partitions for any bundle $B_{i,t_i}$ with $i, t_i \leq k+1$ and show that ${\bf P}$ is overconstrained.

By the inductive hypothesis, let $X^{(k+1)}$ be the promised frugal orientation for agents $1, \dots, k$ with respect to $B_1(\mathbf{I}_1(k)), \dots, B_{k}(\mathbf{I}_{k}(k))$, where $\mathbf{I}_i(k) = \{1, \dots, k\}$. Moreover, we have $v_i\left(X_i^{(k+1)}\right) \geq \frac{2}{3} \min_{t \in \mathbf{I}_i(k)} \left\{ v_i(B_{i,t}) \right\}$. By frugality, each agent $i \leq k$ receives a subset of their canonical bundles. Without loss of generality, assume $X_i^{k+1} \subseteq B_{i,1}$.

Now, consider the indices $\mathbf{I}'_1(k), \dots, \mathbf{I}'_{k}(k)$, where $\mathbf{I}'_i(k) = \{2, \dots, k+1\}$. That is, for each $\mathbf{I}'_i(k)$, we omit index $1$ —the index from which the previous allocated bundle $X_i^{(k+1)}$ was a subset of— and include the remaining index, $k+1$. Let $X'^{(k+1)}$ be the corresponding promised allocation for agents $1, \dots, k$, satisfying  
$
v_i\left(X'^{(k+1)}_i\right) \geq \frac{2}{3} \min_{t \in \mathbf{I}'_i(k)} \left\{ v_i(B_{i,t}) \right\}.
$  
By frugality, each agent $i \leq k$ receives a subset of their canonical bundles. Without loss of generality, assume $X'^{(k+1)}_i \subseteq B_{i,2}$.

Now, consider the partition of $B_{k+1,1}$ in three parts as follows: 
\[
    \left(B_{k+1,1} \cap X^{(k+1)}, B_{k+1,1} \cap X'^{(k+1)}, B_{k+1,1} \setminus \left(X^{(k+1)} \cup X'^{(k+1)}\right)\right).
\]

We argue that this is a feasible partition. Indeed, take any edge $e = (i,j) \in X^{(k+1)} \cap X'^{(k+1)}$ and assume that $e$ is oriented towards $i$ in $X^{(k+1)}$, i.e., $e \in X_i^{(k+1)} \subseteq B_{i,1}$. By the construction of $X'^{(k+1)}$, edge $e$ should be oriented towards agent $j$ in $X'^{(k+1)}$. Therefore, $X^{(k+1)} \cap X'^{(k+1)} \cap B_{k+1,1} = \emptyset$, i.e. the edge $e$ is irrelevant for agent $k+1$. See Figure~\ref{fig:validpartition} for an illustration of such a partition for bundle $B_{k+1,1}$.

\begin{figure*}[htbp]
  \centering
  \begin{tabular}{cccc}
\begin{tikzpicture}[scale=0.6]
\small{
  \draw[step=1cm,] (0,-1) grid (4,3);
    \draw[pattern={crosshatch dots},pattern color=green](2,2) rectangle +(2,1);
\draw[pattern={north west lines},pattern color=red](1,2) rectangle +(1,1);
\draw[pattern={horizontal lines},pattern color=blue](0,2) rectangle +(1,1);

\node[anchor=north] at (0.3,4) {$B_{1,1}$};
\node[anchor=north] at (1.4,4) {$B_{1,2}$};
\node[anchor=north] at (2.35,3.75) {$\ldots$};
\node[anchor=north] at (3.6,4) {$B_{1,k+1}$};

\node[anchor=west] at (-2.5,-0.5) {$B_{k+1,k+1}$};
\node[anchor=west] at (-1.2,0.5) {$\vdots$};
\node[anchor=west] at (-2,1.5) {$B_{k+1,2}$};
\node[anchor=west] at (-2,2.5) {$B_{k+1,1}$};

\node[anchor=west] at (1.2,-1.5) {$E_{1,k+1}$};
\node[anchor=west] at (1.5,-2.2) {(a)};}
\end{tikzpicture} & \begin{tikzpicture}[scale=0.6]
\small{
  \draw[step=1cm,] (0,-1) grid (4,3);
    \draw[pattern={crosshatch dots},pattern color=green](2,2) rectangle +(2,1);
\draw[pattern={north west lines},pattern color=red](1,2) rectangle +(1,1);
\draw[pattern={horizontal lines},pattern color=blue](0,2) rectangle +(1,1);tern color=red](2,0) rectangle +(1,1);
\node[anchor=north] at (0.3,4) {$B_{2,1}$};
\node[anchor=north] at (1.4,4) {$B_{2,2}$};
\node[anchor=north] at (2.35,3.75) {$\ldots$};
\node[anchor=north] at (3.6,4) {$B_{2,k+1}$};

\node[anchor=west] at (-2.5,-0.5) {$B_{k+1,k+1}$};
\node[anchor=west] at (-1.2,0.5) {$\vdots$};
\node[anchor=west] at (-2,1.5) {$B_{k+1,2}$};
\node[anchor=west] at (-2,2.5) {$B_{k+1,1}$};
\node[anchor=west] at (1.2,-1.5) {$E_{2,k+1}$};
\node[anchor=west] at (1.5,-2.2) {(b)};}
\end{tikzpicture} & \begin{tikzpicture}[scale=0.6]
  \draw[step=1cm,] (0,0) grid (0,0);
  \large{
\node[anchor=north] at (0,4) {$\boldsymbol{\ldots}$};}
\end{tikzpicture} &\begin{tikzpicture}[scale=0.6]
\small{
  \draw[step=1cm,] (0,-1) grid (4,3);
    \draw[pattern={crosshatch dots},pattern color=green](2,2) rectangle +(2,1);
\draw[pattern={north west lines},pattern color=red](1,2) rectangle +(1,1);
\draw[pattern={horizontal lines},pattern color=blue](0,2) rectangle +(1,1);tern color=red](2,0) rectangle +(1,1);
\node[anchor=north] at (0.3,4) {$B_{k,1}$};
\node[anchor=north] at (1.4,4) {$B_{k,2}$};
\node[anchor=north] at (2.35,3.75) {$\ldots$};
\node[anchor=north] at (3.6,4) {$B_{k,k+1}$};

\node[anchor=west] at (-2.5,-0.5) {$B_{k+1,k+1}$};
\node[anchor=west] at (-1.2,0.5) {$\vdots$};
\node[anchor=west] at (-2,1.5) {$B_{k+1,2}$};
\node[anchor=west] at (-2,2.5) {$B_{k+1,1}$};

\node[anchor=west] at (1.2,-1.5) {$E_{k,k+1}$};
\node[anchor=west] at (1.5,-2.2) {(c)};}
\end{tikzpicture} \\
    \end{tabular}  
\caption{Common edges $E_{1,k+1}$ (a), $E_{2,k+1}$ (b) and $E_{k,k+1}$ (c). In each set $E_{i,j}$ rows represent the bundles $B_{i,k_i}$ of agent $i$ and columns represent the bundles $B_{j,t_j}$ of agent $j$. Cell $(t_i,t_j)$ represents the intersection $B_{i,t_i} \cap B_{j,t_j}$. The figure illustrates the partition of $B_{k+1,1}$ bundles as follows: $\left(B_{k+1,1} \cap X^{(k+1)}, B_{k+1,1} \cap X'^{(k+1)}, B_{k+1,1} \setminus \left(X^{(k+1)} \cup X'^{(k+1)}\right)\right)$. $X^{(k+1)}$ is a $2/3$-MMS frugal orientation with $X^{(k+1)}_i \subseteq B_{i,1}$ and $X'^{(k+1)}$ is a $\frac{2}{3}$-MMS frugal orientation with $X'^{(k+1)}_i \subseteq B_{i,2}$.  The three parts are color-coded in blue, red, and green, respectively. $X^{(k+1)}$ and $X'^{(k+1)}$ remain disjoint over the edge set $E_{k+1}$.}
\label{fig:validpartition}
\end{figure*}

If for at least one of the first two parts has value no more than $1/3$ i.e. either $v_{k+1}(B_{k+1,1} \cap X^{(k+1)}) \le 1/3$ or $v_{k+1}(B_{k+1,1} \cap X'^{(k+1)}) \ge 1/3$, then $X^{(k+1)}$ or $X'^{(k+1)}$ can be extended with $X_{k+1} = B_{k+1,1} \setminus X^{(k+1)}$ (or resp. $X_{k+1}' = B_{k+1,1} \setminus X'^{(k+1)}$) and then (i) is satisfied. Otherwise, by additivity, 
\[
v_{k+1}\left(B_{k+1,1} \cap \left(X^{(k+1)} \cup X'^{(k+1)}\right)\right) \geq 2/3.
\]  

Since agent $k+1$ and bundle $B_{k+1,1}$ were picked arbitrarily, it also holds for every $i, t_i\in\{1,\ldots, k+1\}$ and we conclude  that condition~(\ref{eq:XOSstructure}) is satisfied with $X^{(i)}$ and $X'^{(i)}$ being the two promised orientations, as per the inductive hypothesis, for ${\bf P} \setminus \{i\}$, and the claim follows. 
\end{proof}

We denote by $S_{i,t_i}$ the high value intersections of an overconstrained set, $ S_{i,t_i} = B_{i,t_i} \cap \left(X^{(i)} \cup X'^{(i)}\right)$ for each agent $i \in \mathbf{P}$ and $ t_i\in\{1,\ldots,  k+1\} $. 
Notably, due to the frugality, each bundle $S_{i,t_i}$ intersects with exactly two canonical MMS bundles of any other agent $j \in \mathbf{P} \setminus \{i\}$; with (i) the bundle $B_{j,t_j}$ such that $X^{(i)}_j\subseteq B_{j,t_j}$ and (ii) the bundle $B_{j,t_j'}$ such that $X'^{(i)}_j\subseteq B_{j,t_j'}$. 
Symmetrically, each bundle $S_{j,t_j}$ intersects with exactly two MMS bundles of agent $i$. 

We say that bundle $S_{i,t_i}$ {\em intersects} with bundle  $S_{j,t_j}$ if and only if $X^{(i)}_j\subseteq B_{j,t_j}$ or $X'^{(i)}_j\subseteq B_{j,t_j}$ and additionally $X^{(j)}_i\subseteq B_{i,t_i}$ or $X'^{(j)}_i\subseteq B_{i,t_i}$. By definition, for each pair of agents there are exactly four such intersections. 

In the next claim, we show that for $|\mathbf{P}| \geq 5$, we can allocate to each agent $i \in \mathbf{P}$ some bundle $S_{i,t_i}$ in a feasible way, i.e., the bundles will not intersect and achieve $2/3$-MMS frugal orientation. We emphasize that for the special case of $|\mathbf{P}| = 4$, such an orientation need not exist (i.e. allocating a bundle $S_{i,t_i}$ to each agent $i$ in a feasible way). In this case, we show that if such an orientation is not possible, we can redistribute the edges in the $S_{i,t_i}$ sets and still achieve a $2/3$-MMS frugal orientation.
\begin{claim}\label{clm:XOS2}
  If ${\bf P}$ is overconstrained and has $5$ or more agents, then, a greedy approach admits a $2/3$-MMS allocation $X = (X_1, \dots,X_{k+1})$ such that for
  every agent $i$, the allocated bundle satisfies
  $X_i = S_{i,t_i}$
  for some $t_i \in \mathbf{I}_i(k+1)$. If $|\bf P|=4$, then we can use a special treatment and achieve $2/3$-MMS frugal orientation.
\end{claim}

Claim~\ref{clm:XOS2} proves the inductive hypothesis and hence completes the proof. As a warm up, we present a simpler proof of a slightly weaker statement, which is of particular interest when
$|{\bf P}| \ge 15$. The existence of such an allocation
can be established through an interesting connection to the theory
of {\em independent transversal sets} in $r$-partite graphs. An independent transversal set is a selection of $r$ vertices, one from each part, such that no two vertices share an edge.

An overconstrained set $\mathbf{P}$ can be modeled by a $|\mathbf{P}|$-partite graph. Part $i$ consists of $k+1$ vertices corresponding to the $S_{i,t_i}$ sets. An edge $(v_{i,t_i},v_{j,t_j})$ represents the intersecting bundles $S_{i,t_i}$ and $S_{j,t_j}$. 
Because overlapping bundles cannot be assigned simultaneously to their respective agents, the goal is to find an independent transversal set. However, an independent transversal set is not guaranteed to exist for a $4$-partite graph with this special structure, (see Figure~\ref{fig:4XOS}), therefore this case needs special treatment. The proof of Claim~\ref{clm:XOS2} shows the existence of an independent transversal set for any $|\mathbf{P}| \ge 5$ with this structure via a greedy procedure. 
We claim that, using Corollary~10 of Wanless and Wood \cite{WanlessWood21}, existence can be established for $|\mathbf{P}| \geq 15$.

\begin{lemma}[ Corollary~10 by \cite{WanlessWood21}] \label{lem:ITSXOS}
    Fix any $t \ge 1$. For a graph $G$, let $V_1,\ldots,V_n$ be a partition of $V(G)$ such that $\lvert V_i \rvert \ge t$ and there are at most $\frac{t}{4}|V_i|$ edges in $G$ with exactly one endpoint in $V_i$ for each $i \in \{1, \ldots,n\}$. Then there exists at least $(\frac{t}{2})^n$ independent transversal of $V_1,\ldots,V_n$.
\end{lemma}

Indeed, by setting  $t = |\mathbf{P}|$ and $n = |\mathbf{P}|$ we observe that when $|\mathbf{P}| \ge 15$, the conditions of Lemma~\ref{lem:ITSXOS} are satisfied for the graph that models the overconstrained set, which guarantees the existence of an independent transversal set and, consequently, a $2/3$-MMS frugal orientation.

\begin{proof}[Proof of Claim~\ref{clm:XOS2}]

We distinguish cases based on the cardinality of the over constrained set.

\noindent{\textbf{Case of} $|\mathbf{P}| \ge 5$.} Let $K_{2,2}^{(i,j)}=\left\{S_{i,t_i},S_{i,t_i'},S_{j,t_j},S_{j,t_j'}\right\}$ be the collection of bundles such that the pairs formed\\ $(S_{i,t_i}, S_{j,t_j}),(S_{i,t_i}, S_{j,t_{j}'})$, $(S_{i,k_{i}'}, S_{j,k_j}),(S_{i,t_{i}'}, S_{j,t_{j}'})$ are all intersected i.e., we cannot allocate two bundles from the same collection. Recall that for two fixed agents $i$ and $j$, collection $K_{2,2}^{(i,j)}$ has exactly for bundles; two bundles of agent $i$ and two of agent $j$. See Figure~\ref{fig:conditions} for an illustration of those bundles. 
\begin{figure}[htbp]
  \centering
\begin{tikzpicture}[scale=0.6]
\small{
  \draw[step=1cm,] (0,0) grid (6,6);
   
\draw[pattern={north west lines},pattern color=red](0,4) rectangle +(6,2);
\draw[pattern={horizontal lines},pattern color=blue](0,0) rectangle +(2,6);
\node[anchor=north] at (0.5,7) {$B_{j,1}$};
\node[anchor=north] at (1.6,7) {$B_{j,2}$};
\node[anchor=north] at (3,6.75) {$\ldots$};
\node[anchor=north] at (4.2,6.75) {$\ldots$};

\node[anchor=north] at (6,7) {$B_{j,k+1}$};

\node[anchor=west] at (-2,0.5) {$B_{i,k+1}$};
\node[anchor=west] at (-1.2,3.5) {$\vdots$};
\node[anchor=west] at (-1.2,2) {$\vdots$};

\node[anchor=west] at (-2,4.5) {$B_{i,2}$};
\node[anchor=west] at (-2,5.5) {$B_{i,1}$};}
\end{tikzpicture}
\caption{Common edges $E_{i,j}$ where the rows correspond to the bundles $B_{i,t_i}$. In each set $E_{i,j}$ rows represent the bundles $B_{i,t_i}$ of agent $i$ and columns represent the bundles $B_{j,t_j}$ of agent $j$. Cell $(t_i,t_j)$ represents the intersection $B_{i,t_i} \cap B_{j,t_j}$. The figure focus on agents $i,j$ of an overconstrained set. With blue are the bundles $S_{i,t_i}=B_{i,t_i} \cap \left(X^{(i)} \cup X'^{(i)}\right)$ with $X^{(i)}_j\subseteq B_{j,1}$ and $X'^{(i)}_j\subseteq B_{j,2}$. With red are the bundles $S_{j,t_j}=B_{j,t_j} \cap \left(X^{(j)} \cup X'^{(j)}\right)$ with $X^{(i)}_i\subseteq B_{i,1}$ and $X'^{(j)}_i\subseteq B_{i,2}$. The intersected bundles are $(S_{i,1},S_{j,1}), (S_{i,2},S_{j,1}), (S_{i,1},S_{j,2})$ and $(S_{i,2},S_{j,2})$}
\label{fig:conditions}
\end{figure}

For a fixed agent $i$ her bundles are in $|\mathbf{P}|-1$ different collections in total, one for every other agent. Thus, by the pigeonhole principle, we can conclude that for every agent $i$:

\begin{itemize}  
    \item[(i)] Either there exists a bundle $S_{i,t_i}$ that is fully compatible with the allocation, meaning that it is not in any $K_{2,2}^{(i,j)}, j \in \mathbf{P}$.
    \item[(ii)] Agent $i$ has at least two distinct bundles, $S_{i,t_i}$ and $S_{i,t_i'}$, such that each one is contained in exactly one $K_{2,2}^{(i,j)}, j \in \mathbf{P}$ respectively.
\end{itemize}  

We refer to bundles $S_{i,t_i}$ with properties (i) or (ii) for all as {\em candidate bundles} with respect to the set of agents $\mathbf{P}$. We establish the result using a greedy approach, prioritizing the allocation of the candidate bundles $S_{i,t_i}$ i.e., the bundles which are in fewer collections.

\noindent{\bf{Case 1.}} If for some agent $i$ there exists a candidate bundle $S_{i,t_i}$ with property (i), that is, none of the collections contain it, then allocating it to agent $i$ does not affect the final allocation i.e., for any other bundle $S_{j,t_j}$ allocated to an agent $j$, there will not be a conflict.
    
\noindent{\bf{Case 2.}} If there exist agents $i,j_1$ with candidate bundles $S_{i,t_i} \in K_{2,2}^{(i,j_1)}$ and $S_{j_1,t_{j_1}} \in K_{2,2}^{(j_1,j_2)}$ with $j_2 \ne i$, then we can allocate $S_{i,t_i}$ and $S_{j_1,t_{j_1}}$ to the corresponding agents and proceed in the same way to agent $j_2$ until we reach an agent $j_{r^*}$ with a candidate bundle $S_{j_{r^*},t_{j_{r^*}}}$ with the property (i) w.r.t the set of the remaining agents and the agent selected in the previous iteration i.e. there is no $j \in \mathbf{P} \setminus \{i, j_1, \dots, j_{r^*-2}\}$ such that $S_{j_{r^*},t_{j_{r^*}}}\in K_{2,2}^{(j_{r^*},j)}$. Note that this is sufficient since $S_{j_{r^*},t_{j_{r^*}}}$ is not in the same collection with any candidate bundle previously selected; otherwise agent $j_{r^*}$ would have been selected earlier.
    
We now argue that, in each iteration, it is possible to allocate a candidate bundle to agent $j_r,r\ge 2$. Indeed, after the selection of agent $j_r$ there are at most $|\mathbf{P}|-r$ collections overall w.r.t. the remaining agents and agent $j_{r-1}$. Due to the fact that each collection obtains two of her bundles by the pigeonhole principle, there exists either at least one candidate bundle with property (ii) i.e., $S_{j_r,t_{j_r}} \in K_{2,2}^{(j_r,j_{r+1})},j_{r+1}\ne j_{r-1}$ and thus we proceed to the next agent or there exists a candidate bundle with property (i) and hence $r=r^*$.  Observe that for $r \ge |\mathbf{P}|/2 + 1$ from the pigeonhole principle a candidate bundle with property (i) is guaranteed to exist for agent $j_r$. 

\noindent{\bf{Case 3.}} If this is not the case then the agents are forming pairs and collection $K_{2,2}^{(i,j)}$ contains the two candidate bundles of agent $i$ with property (ii) and the two candidate bundles of agent $j$ with property (ii). Figure~\ref{fig:4XOS} illustrates this structure for four agents.  

\begin{figure}[htbp]
  \centering
    
    \begin{tikzpicture}[scale=0.85]

  \small{
  \node[anchor=south west] at (-3.75,-0.75) {${\bf V_1}$};
\node[anchor=south west] at (0.25,3) {${\bf V_2}$};
\node[anchor=south west] at (4,-0.75) {${\bf V_3}$};
\node[anchor=south west] at (0.25,-4.5) {${\bf V_4}$};

  \foreach \i in {1, 2, 3, 4} {
    \node[draw, circle, fill=blue!30] (A\i) at (-2.5, 2.5 - 1.25*\i) {$S_{1,\i}$};
  }

  \foreach \i in {1, 2, 3, 4} {
    \node[draw, circle, fill=red!30] (B\i) at (-2.5 + 1.25*\i, 2.5) {$S_{2,\i}$};
  }

  \foreach \i in {1, 2, 3, 4} {
    \node[draw, circle, fill=green!30] (C\i) at (3.5, 2.5 -1.25*5+ 1.25*\i) {$S_{3,\i}$};
  }

  \foreach \i in {1, 2, 3, 4} {
    \node[draw, circle, fill=yellow!30] (D\i) at (-2.5 + 1.25*5 -1.25*\i, -3.5) {$S_{4,\i}$};
  }

  \foreach \i in {1, 2} {
    \foreach \j in {1, 2} {
      \draw[ultra thick, blue] (A\i) -- (B\j);
    }
  }

  \foreach \i in {3, 4} {
    \foreach \j in {3, 4} {
      \draw[ultra thick, blue] (A\i) -- (C\j);
    }
  }

  \foreach \i in {3, 4} {
    \foreach \j in {3, 4} {
      \draw[ultra thick, blue] (A\i) -- (D\j);
    }
  }

  \foreach \i in {3, 4} {
    \foreach \j in {3, 4} {
      \draw[ultra thick, red] (B\i) -- (C\j);
    }
  }

  \foreach \i in {3, 4} {
    \foreach \j in {3, 4} {
      \draw[ultra thick, red] (B\i) -- (D\j);
    }
  }

  \foreach \i in {1, 2} {
    \foreach \j in {1, 2} {
      \draw[ultra thick, green] (C\i) -- (D\j);
    }
  }
  }
\end{tikzpicture}
\caption{The figure represents bundles $S_{i,t_i}$ via a graph. Each part $V_i$ corresponds to an agent $i$, with vertices representing her bundles. An edge between two vertices $v_{i,t_i},v_{j,t_j}$ exists if and only if the corresponding bundles $S_{i,t_i},S_{j,t_j}$ may intersect. Finding an independent transversal set ensures an allocation of bundles $S_{i,t_i}$. In the case where $\lvert \mathbf{P} \rvert = 4$, no such independent set guarantee to exist, as shown in the figure, requiring a redistribution of edges among the agents.}
\label{fig:4XOS}
\end{figure}

For each such pair of agents we allocate arbitrarily the bundle $S_{i,t_i}$ to agent $i$ and then proceed to the next pair. 

After the processing of all pairs we argue that for $|\mathbf{P}| \ge 5$ we can allocate to each of the remaining agents (omitting their initial candidate bundles) either a candidate bundle having the property (i) w.r.t. the remaining agents or we can iterate as in Case 2 and thus the claim follows. Indeed, for each of the remaining agents, there are at most $\left\lfloor |\mathbf{P}|/2 \right\rfloor -1$ collections i.e. one of with each remaining agent. Omitting their initial candidate bundles, those two which forms intersecting sets with the paired agent to which we allocated bundle arbitrarily, each agent has $|\mathbf{P}|-2$ bundles. Each collection contains $2$ bundles and by the pigeonhole principle the argument holds.

\noindent{\textbf{Case of} $|\mathbf{P}| = 4$.} Unfortunately, the pigeonhole principle argument used in the last case, does not hold in the case of $4$ agents. More generally, there need not exist a feasible allocation of the $S_{i,t_{i}}$ sets. However, the 4-partite graph that models the conflicts of the $S_{i,t_i}$ sets must have a very specific structure (isomorphic to Figure~\ref{fig:4XOS}). In this case, we will provide a different orientation. The proof is applied for the structure represented in Figure~\ref{fig:4XOS}. Clearly, by renaming we can apply the same proof in every case with $4$ agents in which an independent transversal set is not guaranteed to exist. 

Consider the $2/3$-MMS frugal orientations among the first three agents: $X''^{(4)} = \{S_{1,1}, S_{2,3}, S_{3,1}\}$ and $X'''^{(4)} = \{S_{1,2}, S_{2,4}, S_{3,2}\}$.

Now, consider the partition of $B_{4,4}$ in three parts as follows:
$$ \left(B_{4,4} \cap X''^{(4)}, B_{4,4} \cap X'''^{(4)}, B_{4,4} \setminus \left(X''^{(4)} \cup X'''^{(4)}\right)\right).
$$

By the construction of orientations $X''^{(4)}$ and $X'''^{(4)}$ respectively, each agent $i \in \{1,2,3\}$ gets the edges from $S_{i,t_i}$ and $S_{i,t_i+1}$. We argue that agent $4$ derives no value from the intersection of these orientations, i.e.,  $v_4\left(X''^{(4)} \cap X'''^{(4)}\right) = 0$. Indeed, take any edge $e = (i,j) \in X''^{(4)} \cap X'''^{(4)}$ and assume that $e$ is oriented towards $i$ in $X''^{(4)}$, i.e., $e \in X_i''^{(4)} \subseteq S_{i,t_i}$. By the construction of $X'''^{(4)}$, edge $e$ should be oriented towards agent $j$ in $X'''^{(4)}$. Therefore edge $e$ is irrelevant for agent $4$ and thus the partition is feasible.

If either $v_{4}(B_{4,4} \cap X''^{(4)})$ or $v_{4}(B_{4,4} \cap X'''^{(4)})$ are at most $1/3$ then $X''^{(4)}$ or $X'''^{(4)}$ can be extended with $X_{4}'' = B_{4,4} \setminus X''^{(4)}$ (or resp. $X_{4}''' = B_{4,4} \setminus X'''^{(4)}$). Otherwise, by additivity, 
\[
v_{4}\left(B_{4,4} \cap \left(X''^{(4)} \cup X'''^{(4)}\right)\right) \geq 2/3.
\] 

From the structure of bundles $S_{i,t_i}$ (Figure~\ref{fig:4XOS}) we have  
$v_4\left(B_{4,4} \cap \left(X''^{(4)} \cup X'''^{(4)}\right)\right) = v_4\left(B_{4,4} \cap \left(S_{2,4} \cup S_{2,3} \right)\right)\\ \geq 2/3$. Therefore, 
$X = \left(S_{1,3}, S_{2,1}, S_{3,2}, B_{4,4} \cap \left(S_{2,3} \cup S_{2,4} \right)\right)$ admits a $2/3$-MMS orientation, ensuring that each agent receives a subset of her MMS bundle.
\end{proof}
\end{proof}

In the following theorem, we show that a $1/2$-out-of-$2$ MMS orientation exists for $n$ XOS agents. The proof is similar to that of Theorem~\ref{thm:multiaddroundrobin}. We note that this result is tight due to the upper bound of \cite{GhodsiHSSY22} for two agents i.e., we can extend the counter example by adding $n-2$ disconnected vertices and the inapproximability follows.

\begin{theorem}\label{thm:XOS1outofd}
In every multi-graph with $n$ XOS agents, there exists a frugal $1/2$-out-of-$2$ MMS orientation.
\end{theorem}
\begin{proof}
Let the graph $G=(V,E)$ where $V$ is the set of $n$ agents and $E$ is the set of edges. Let $B_i^2=(B_{i,1},B_{i,2})$ be the $1$-out-of-$2$ MMS partition of agent $i$. We will show that we can find a frugal $1/2$-out-of-$2$ allocation considering additive valuations which using Lemma~\ref{lem:XOSadditive} admits a $1/2$-out-of-$2$ MMS allocation with XOS agents. We present a greedy algorithm that for $n\ge 2$ produces a $1/2$-out-of-$2$ MMS allocation. 

The algorithm runs in at most $m$ rounds. In the first round,  we pick an arbitrary agent $i_1$, who selects the relevant edge with the maximum value for her, edge $e=(i_1,i_2)\in E^1$ and we associate agent $i_1$ with her canonical bundle $B_{i_1,t_{i_1}}$ which contains the edge she selected i.e., $e \in B_{i_1,t_{i_1}}$. Then we proceed with agent $i_2$ that shares this edge (who may be agent $i_1$ in case of a self-loop). Let $i_k$ be the selected agent in round $k$. If agent $i_k$ is not associated with some bundle $B_{i_k,t_{i_k}}$ then we associate the agent with her canonical bundle $B_{i_k,t_{i_k}}$ such that $B_{i_k,t_{i_k}} \subseteq E^k$ i.e. none of the edges is oriented before turn $k$.  Then we orientate to the selected agent $i_k$, the best available edge $e'=(i_k,i_{k+1})\in E^k\cap B_{i_k,t_{i_k}}$ i.e., $e' \in \arg \max_{e \in E^k\cap B_{i_k,t_{i_k}}} \{v_i(e)\}$.  We update the set of edges $E^{k+1}=E^k\setminus \{e'\}$ and proceed to the next round $k+1$ with agent $i_{k+1}$. If in some round $t$ the set $E^{t}_{i_t} = \emptyset$ for the selected agent $i_t$, then we continue with an arbitrary agent $i'$ such that $E^{t}_{i'}\ne \emptyset$. The algorithm terminates after at most $m=\lvert E \rvert$ iterations and each edge is allocated at most to one agent.

In what follows we will show that the resulting allocation is a $1/2$-out-of-$2$ MMS allocation. Take an agent $i$ and  assume w.l.o.g. that agent $i$ is selected for the first time at some step $k$ and hence we have to define the fixed canonical bundle $B_{i,t_i},t_i \in \{1,2\}$. We claim that there exists a canonical bundle such that $B_{i,t_i}\subseteq E^k$.  Indeed, let $A$ be the set of edges allocated before round $k$.  If agent $j$ gets an edge $e=(i,j)\in A\cap E_i$ in round $k'<k$ by
  definition agent $i$ will be selected in the next round, hence  $k'=k-1$. Therefore, $\lvert A \cap E_i\rvert \le 1$. If $A \cap E_i \subseteq B_{i,1}$ (resp. $A \cap E_i \subseteq B_{i,2}$) then the associated bundle is $B_{i,2}$ (resp. $B_{i,1}$).

  Let $P=(B,C)$ be a partition of $B_{i,t_i}$ in two (possibly empty) parts which is defined as follows: $B$ is the set of edges allocated to agent $i$ and $C$ contains the rest of the edges, i.e., the edges allocated after step $k$ to agents other than $i$.  It suffices to show that $v_i(B) \ge 1/2 \mu_i^2$.

  Next we claim that
  $v_i(C)\le v_i(B)$. It suffices to observe that for each edge
  $e_C \in C\cap E_i$ there exists a corresponding edge $e_B \in B$, with higher value for $i$; $i$ selected $e_B$ over $e_C$, hence
  $v_i(e_B)\geq v_i(e_C)$. We conclude that $v_i(B)\geq v_i(C)$ by noting that all edges $e\in C\setminus E_i$ are irrelevant for $i$. Note that $v_i(B)+v_i(C) = v_i(B_{i,t_i}) \ge \mu_i^2$ and as a result $v_i(B)  \ge  1/2\mu_i^2$ and the theorem follows. 

\end{proof}

\subsection{Asymptotically Improved Bounds for XOS Valuations}
\label{sec:SBXOS}
In this section, we provide asymptotically improved lower bounds for XOS agents. The main theorem of this section (Theorem~\ref{thm:sqrtXOS}) states that it is sufficient for each agent to have $d \in \mathcal{O}(\sqrt{n})$ bundles in order to achieve a $2/3$-out-of-$d$ orientation. To show this, we utilize some recent results on the number of independent transversal sets from Dai, Liu and Zhang \cite{DaiLiuZhang2025}. 


\begin{theorem}\label{thm:sqrtXOS}
    In every multi-graph with $n$ XOS agents, there exists a frugal $2/3$-out-of-$8\sqrt{n}$ orientation. 
\end{theorem}

Before we continue, we introduce the necessary notation. Given an $r$-partite graph $G$ with a partition $P = (V_1, V_2, \ldots, V_r)$ such that $V_i$ is an independent set for each $i \le r$, let $d_G(v)$ denote the degree of vertex $v \in V_i$ in the graph $G$ for some $i \le r$. We define the \emph{average degree} of $V_i$ in $G$ as:
$$\bar{d}_G(V_i) =\frac{1}{|V_i|} \sum_{v\in V_i} d_G(v).$$
Additionally, we define the \emph{maximum block average degree} as: 

$$\bar{D}_P(G) = \max_{V_i \in P} \bar{d}_G(V_i)$$
which represents the highest average degree among all the independent sets $V_i$ in the partition $P$. A partition $P = (V_1, V_2, \ldots , V_r)$ is $s$-thick if each $V_i$ has size at least $s$. In an $r$-partite graph, an independent transversal set of size $s$ (ITS) consists of $s$ vertices from each part, forming an independent set. We are now ready to present the proof of Theorem~\ref{thm:sqrtXOS}.

\begin{proof}
Let $G=(V,E)$ be a graph with $n \ge 2$ vertices, and let $\mathbf{B}^{8\sqrt{n}}_i=(B_{i,1},\ldots,B_{i,8\sqrt{n}})$ denote the canonical $1$-out-of-$8\sqrt{n}$-MMS partition of agent $i$ for the edge set $E_i$. For some $k \le 8\sqrt{n}$, let $\mathbf{I}_i(k) \subseteq\{1,2,\ldots,8\sqrt{n}\}$ be a subset of $k$ indices, i.e., $|\mathbf{I}_i(k)|=k$. We define $B_i(\mathbf{I}_i(k))$ as the set of the respective bundles $B_{i,t}$ with $t \in \mathbf{I}_i(k)$ i.e., $B_{i}(\mathbf{I}_i(k))=\{B_{i,t}|t \in \mathbf{I}_i(k)\}$.  For some subset of agents $\mathbf{P} \subseteq V$ we denote the set of agents $\mathbf{P}_{-i} =\mathbf{P} \setminus \{i\}, i \in \mathbf{P}$. Also, let the sets $E_{i}^{\mathbf{P}}=E_{i} \cap \left(\bigcup_{j \in \mathbf{P}}E_{i,j}\right)$ and $E^{\mathbf{P}}= \bigcup_{i \in \mathbf{P}} E_{i}^{\mathbf{P}}$ be \emph{the set of the common edges of agent $i$} with agents in $\mathbf{P}$ and \emph{the set of the common edges} of agents in $\mathbf{P}$ respectively. We will refer to those sets as the set of the common edges, omitting $\mathbf{P}$ and agent $i$ when they are clear from context.

We employ Lemma~\ref{lem:XOSadditive}, which asserts that it suffices to show the existence of a  $2/3$-out-of-$8\sqrt{n}$-MMS frugal orientation in every multi-graph with $n \ge 2$ additive agents. We prove by induction on $k \in \{2, \ldots, n\}$ that for every subset $\boldsymbol{P}=(p_1,\ldots,p_k) \subseteq V$ of $k$ agents and for any fixed $k$ sets $\mathbf{I}_{p_1}(8\sqrt{k}), \ldots,\mathbf{I}_{p_k}(8\sqrt{k})$ of indices of size $8\sqrt{k}$, there is a frugal orientation $X$, restricted to the set of the common edges,  with respect to $\mathbf{B}^k=(B_{p_1}(\mathbf{I}_{p_1}(8\sqrt{k})), \ldots, B_{p_k}(\mathbf{I}_{p_k}(8\sqrt{k})))$  (the allocated bundle of each agent is a subset of some canonical bundle with index in the corresponding set) such that each agent is guaranteed at least $2/3$ of the restricted value that the corresponding canonical bundle has over the set of the common edges. That is, for all agents $p_i \in \mathbf{P}$, there exists a canonical bundle $B \in B_{p_i} ( \mathbf{I}_{p_i}(8\sqrt{k}))$ such that $$v_{p_i}(X_{p_i}) \ge \frac{2}{3}v_{p_i}(B \cap E_{i}^{\mathbf{P}}), \quad X_{p_i} \subseteq B \cap E_{p_i}^{\mathbf{P}}$$

We note that every frugal orientation with respect to bundles $B_{p_i}(\mathbf{I}_{p_i}(8\sqrt{k}))$ for some agent $p_i \in \mathbf{P}$ and set of indices $\mathbf{I}_{p_i}(8\sqrt{k})$ is also a frugal orientation with respect to the canonical partition $B_{p_i}^{8\sqrt{n}}$. Furthermore, observe that in contrast with the proof of Theorem~\ref{thm:XOS4+}, each agent of the set achieves $2/3$ of the value that some bundle with admissible index has over the set of the common edges with the other agents in the set. This is crucial for our analysis as we cannot guarantee the existence of two $2/3$ frugal orientations in which each agent selects from a different canonical bundle. Intuitively, if the induction holds for some $k$, then we append the set with one more agent, agent $i$. Now, the allocation guaranteed to exist for the $k$ agents, has to be also extended with the set of the common edges of those agents and the new agent, agent $i$. We give priority to the first $k$ agents of the group over their common edges with agent $i$ in order to guarantee at least $2/3$ of their corresponding selected bundle also in the new set of items.

   For $k=n$, we have $\boldsymbol{P} =V$ and $E_{i}^{V}=E_i$, and as a result, for each agent $i\in V$ we have: $$v_{i}(X_{i}) \ge \frac{2}{3}\min_{B \in B_i( \mathbf{I}_{i}(8\sqrt{n}))}\left\{v_{i}\left(B \cap E_i\right)\right\}=\frac{2}{3}\min_{B \in \mathbf{B}_i^{8\sqrt{n}}}\left\{v_{i}\left(B \right)\right\}=\frac{2}{3}\text{-out-of-}8\sqrt{n}\text{, }\forall i \in V$$ and the theorem follows. 

    \noindent{\bf Base case, $\bf k=2$}. By Lemma~\ref{lem:2XOSdlower} we have that for $n=2$ there exists a $2/3$-out-of-$3$ MMS orientation. That is, for every pair of agents $(i,j)$ and every two sets of $\left\lceil 8\sqrt{2} \right\rceil$ indices (note that $3$ indices are enough for the case of $2$ agents) $\mathbf{I}_i\left(\left\lceil 8\sqrt{2} \right\rceil\right)$ and $\mathbf{I}_j\left(\left\lceil 8\sqrt{2} \right\rceil\right)$, scaling their value in each bundle over the set of the common edges to one i.e., the scaling is applied in each canonical bundle and $v_i'(e) = \frac{v_i(e)}{v_i(B_{i,k_i} \cap E_{i,j})}$ for $e \in B_{i,k_i}$ and $$v'_i(B_{i,k_i} \cap E_{i,j})=v'_j(B_{i,k_j} \cap E_{i,j})=1\text{, }\forall k_i \in \mathbf{I}_i\left(\left\lceil 8\sqrt{2} \right\rceil\right),k_j \in \mathbf{I}_j\left(\left\lceil 8\sqrt{2} \right\rceil\right)$$ Applying Lemma~\ref{lem:2XOSdlower} we can find a frugal orientation for which the inductive hypothesis holds. That is, $2/3$ value in the scaled bundle admits $2/3$ value of the bundle restricted to the value of their common edges as for the allocated bundle $X_{i} \subseteq B_{i,k_i} \cap E_{i,j}$ we have $$v_i(X_i)=v_i'(X_i)\cdot v_i(B_{i,k_i} \cap E_{i,j})=2/3\cdot v_i(B_{i,k_i} \cap E_{i,j})$$

\noindent{\bf Induction step}. Now, let's assume that the statement holds for all
    possible subsets $\mathbf{P}=\{p_1,\ldots,p_k\}\subset V$ of $k$
    agents and for all possible vectors of $k$ subsets of indices
    $\mathbf{I}_{p_1}(8\sqrt{k}),\ldots, \mathbf{I}_{p_k}(8\sqrt{k})$, with $2\leq k <n$. We will show
    that the statement also holds for all possible subsets
    $\mathbf{P}=\{p_1,\ldots,p_{k+1}\}\subseteq V$ of $k+1$ agents and
    for all possible vectors of $k+1$ subsets of indices
    $\mathbf{I}_{p_1}(8\sqrt{k+1}),\ldots,
    \mathbf{I}_{p_{k+1}}(8\sqrt{k+1})$. W.l.o.g., we will establish the induction step for subset $\mathbf{P}=\{1,\ldots,k+1\}$ and set of indices $\mathbf{I}_i(8\sqrt{k+1})=\{1,\ldots, 8\sqrt{k+1}\}$, for all $i \in \{1,\ldots,k+1\}$; clearly then the statement holds for any other subset of agents and any other vector of indices of size $8\sqrt{k+1}$ by renaming. As in the case of $k=2$, for brevity, we assume that for each agent and each of her canonical bundles, its value over the common set of edges is scaled to one i.e., $v_i\left(B_{i,t} \cap E_{i}^{\mathbf{P}}\right) = 1, \forall i \in \mathbf{P},t \in \mathbf{I}_i(8\sqrt{k+1})$. Clearly, for some agent $i$, a frugal orientation $X_{i} \subseteq B_{i,t} \cap E_{i}^{\mathbf{P}}\text{, }t
    \in \mathbf{I}_i(8\sqrt{k+1})$ with value at least $2/3$ in the scaled instance admits (by rescaling) a frugal orientation with value at least $2/3$ of the total value over $E_{i}^{\mathbf{P}}$ that agent $i$ has for her selected bundle $B_{i,t}$ and hence the inductive hypothesis will hold. 

    We divide the proof of the induction step into two key claims. The first claim (Claim~\ref{clm:XOS1sqrt}), using the inductive hypothesis, outlines the necessary conditions for the case of $k+1$ agents when no frugal orientation in which every agent achieves at least $2/3$ of the value exists. Definition~\ref{def:overconstrainedsqrt} summarizes these
    conditions in terms of what we call an \textit{overconstrained
      set}. The second claim (Claim~\ref{clm:XOS2sqrt}) demonstrates how
    overconstrained sets reveal a useful structure in the intersection
    of the canonical bundles, and uses results on the existence of independent transversal sets to construct a frugal orientation where each agent gets at least $2/3$ of the value.

    We proceed with the definition of an overconstrained set.
\begin{definition}[Overconstrained set]
    \label{def:overconstrainedsqrt}
Let $\mathbf{P}=\{1,\ldots,k+1\}$ be a set of $k+1$ agents and let $\mathbf{I}_i(8\sqrt{k+1})=\{1,\ldots, 8\sqrt{k+1}\}$ with $i \in \mathbf{P}$ be $k+1$ subsets of indices. We say that the set ${\bf P}$ is {\em overconstrained} if for every agent $i$ there exists a set of edges $S_i = \bigcup_{t \in \mathbf{I}_i(8\sqrt{k+1})} S_{i,t}$ with $S_{i,t}\subseteq B_{i,t} \cap E_{i}^{\mathbf{P}}$, such that:

\begin{enumerate}
\item $
v_i\left(S_{i,t}\right) \geq 1/3$, for all $t \in \mathbf{I}_i(8\sqrt{k+1})$ and
\item For every $j \in \mathbf{P}_{-i}$ there exist exactly two indices $t_1^i,t_2^i \in \mathbf{I}_i(8\sqrt{k+1})$ such that $S_{i,t_1^i} \cap E_{i,j} \ne \emptyset$ and $S_{i,t_2^i} \cap E_{i,j} \ne \emptyset$.

\end{enumerate}

\end{definition}

Intuitively, the first property guarantees that for every frugal orientation of edges in $S_{i,t}$ for some index $t$ towards agent $i$ the value she achieves is high enough. Moreover, those bundles $S_{i,t}$ arise from the intersection of the corresponding canonical bundle $B_{i,t}$ and the extended orientation for the set of $k$ agents $\mathbf{P}_{-i}$. The second property guarantees that those edges cannot intersect "much" the remaining agents in $\mathbf{P}_{-i}$ as they are uniformly distributed over their common edges. We are now ready to state the following claim:

\begin{claim} \label{clm:XOS1sqrt}
Let $\mathbf{P}=\{1,\ldots,k+1\}$ be a set of $k+1$ agents and let $\mathbf{I}_i(8\sqrt{k+1})=\{1,\ldots,8\sqrt{k+1}\}$ with $i \in \mathbf{P}$ be $k+1$ subsets of indices. Either (i) there exists a frugal orientation in which each agent gets at least $2/3$, or (ii) ${\bf P}$ is overconstrained.
\end{claim}

\begin{proof}[Proof of Claim~\ref{clm:XOS1sqrt}]
  If (i) holds we are done i.e., by rescaling the values the hypothesis holds for the given orientation. Otherwise, using the inductive hypothesis, for each agent $i$ we identify a set of edges $S_i \subseteq E_{i}^{\mathbf{P}}$ for which ${\bf P}$ is {\em overconstrained}.

To prove this, we focus on agent $k+1$ and we carefully define a partition of her bundles over the set of her relevant edges in two parts. We show that the first one will serve as the desired union of the high value intersections and the second one will be the union of the remaining edges. That is $$B_{k+1,t} \cap E_{k+1}^{\mathbf{P}}=S_{k+1,t} \cup S_{k+1,t}', \quad S_{k+1,t} \cap S_{k+1,t}'=\emptyset$$ for all $t \in \mathbf{I}_{k+1}(8\sqrt{k+1})$ and we show that the partition has the following property: if $v_{k+1}\left(S_{k+1,t}\right) \leq \frac{1}{3}$ for some $t$ then we can extend an existing allocation—guaranteed to exist by the inductive hypothesis—for the set of $k$ agents and $8\sqrt{k}$ indices to the set of $k+1$ agents and $8\sqrt{k+1}$ indices. If this is not the case, we will show that the two properties of the overconstrained set hold for $i=k+1$ and the set $S_{k+1}=\bigcup_{t \in \mathbf{I}_{k+1}(8\sqrt{k+1})} S_{k+1,t}$. We can proceed in a similar fashion by defining partitions for any bundle $B_{i,t}$ with $i, t\leq k+1$ and show that ${\bf P}$ is overconstrained. Thus, it is enough to show the existence of such a partition.

By the inductive hypothesis, let $X^{(k+1)}$ be the promised frugal orientation for the agents in $\mathbf{P}_{-k+1}$ with respect to $B_1(\mathbf{I}_1(8\sqrt{k})), \dots, B_{k}(\mathbf{I}_{k}(8\sqrt{k}))$, where $\mathbf{I}_i(8\sqrt{k}) = \{1, \dots, 8\sqrt{k}\}\subseteq \mathbf{I}_i(8 \sqrt{k+1})$. By frugality, each agent $i $ receives a subset of their canonical bundles and the orientation is also frugal with respect to the set of indices $\mathbf{I}_i(8\sqrt{k+1})$. Without loss of generality, assume $X_i^{k+1} \subseteq B_{i,1}$ and we get $v_i\left(X_i^{(k+1)}\right) \geq \frac{2}{3} v_i(B_{i,1} \cap E_{i}^{\mathbf{P}_{-k+1}})$ for all $i \in \mathbf{P}_{-k+1}$.

Now, we argue that if we extend the orientation with some bundle $X_{k+1} \subseteq B_{k+1,t}$ for some $t \in \mathbf{I}_{k+1}(8\sqrt{k+1})$ such that $v_{k+1}(X_{k+1}) \ge 2/3$ and moreover $v_i(B_{i,1} \cap X_{k+1} ) \le \frac{1}{3}v_i(B_{i,1} \cap E_{i,k+1})$ for all $i \in \mathbf{P}_{-k+1}$
then we are in case (i) of the Claim and the inductive hypothesis holds for the extended orientation $$X=\left(X^{(k+1)}_1 \cup \left(B_{1,1} \cap \left(E_{1,k+1}\setminus X_{k+1}\right) \right),\ldots,X^{(k+1)}_k \cup \left(B_{k,1} \cap \left(E_{k,k+1}\setminus X_{k+1}\right) \right),X_{k+1}\right)$$. Indeed, for the extended orientation the hypothesis holds for agent $k+1$. Furthermore, for each agent $i \in \mathbf{P} _{-k+1}$ there exists a partition of her selected bundle over the set of the common edges, i.e., set $B_{i,1}\cap E_{i}^{\mathbf{P}}$, into two parts such that the orientation guarantees at least $2/3$ in each part and hence, by additivity, guarantees $2/3$ value i.e., $$v_i(X_i)=v_i(X_i \cap E_{i,k+1})+v_i(X_i\cap E_{i}^{\mathbf{P}_{-k+1}}) \ge \frac{2}{3}v_i(E_{i,k+1} \cap B_{i,1})+v_i(X_i^{(k+1)})= 2/3$$

We try to construct such an extended orientation for all $t \in \mathbf{I}_{k+1}(8\sqrt{k+1})$ in which for the set of the allocated edges $X_{k+1} \subseteq B_{k+1,t}$ none of the agents $i \in \mathbf{P}_{-k+1}$ overestimates it i.e., $$v_i(X_{k+1} \cap B_{i,1}) \le \frac{1}{3}v_i(B_{i,1} \cap E_{i,k+1}).$$ The inability to find such a bundle will imply the desired partition of the corresponding canonical bundle by setting $S'_{k+1,t}=X_{k+1}$ and $S_{k+1,t}=B_{k+1,t} \setminus X_{k+1}$.

In order to extend the orientation it is sufficient to prove that agent $k+1$ values her bundle enough. To do so, each agent $i \in \mathbf{P}_{-k+1}$ marks the two indices  in $I_{k+1}(8\sqrt{k+1})$, let them be $(t_1^i,t_2^i)$, for which the corresponding bundles of agent $k+1$ have the highest value over her selected bundle $B_{i,1}$. That is,
$$t_1^i=\arg \max_{t \in \mathbf{I}_{k+1}(8 \sqrt{k+1})}\{v_i(B_{i,1} \cap B_{k+1,t})\}$$
and

$$t_2^i=\arg \max_{t \in \mathbf{I}_{k+1}(8 \sqrt{k+1})\setminus\{t^i_1\}}\{v_i(B_{i,1} \cap B_{k+1,t})\}$$
See Figure~\ref{fig:markedXOS} for an illustration. 

Let the desired set $S_{k+1}$ be the intersections of the canonical bundle each agent selected in orientation $X^{(k+1)}$ with the union of the marked bundles of the corresponding agent, i.e., $$S_{k+1}= \bigcup_{i \in \mathbf{P}_{-k+1}}B_{i,1} \cap \left(B_{k+1,t_{1}^i} \cup B_{k+1,t^i_2}\right)$$ and
$$S'_{k+1,t}= (B_{k+1,t} \cap E_{k+1}^{\mathbf{P}}) \setminus S_{k+1}$$
See Figure \ref{fig:markedXOS} for an illustration. Now we argue that indeed, none of the agents overestimates any bundle $S'_{k+1,t}$. Assume otherwise and let $t^*$ and $i$ such that $v_i(S'_{k+1,t^*}\cap B_{i,1}) > \frac{1}{3}v_i(B_{i,1} \cap E_{i,k+1})$. There can not be the case $t^* \in \{t_1^i,t_2^i\}$ since, by construction, we omit the corresponding intersections from the bundle i.e., $B_{i,t^*} \cap B_{i,1} \in S_{k+1}$ and thus the corresponding value for the intersection is zero as $S'_{k+1,t^*} \cap B_{i,1}= \emptyset$. Hence, there must be the case  $t^* \notin \{t_1^i,t_2^i\}$. Due to additivity, we have 
\[
v_{i}(B_{i,1} \cap B_{k+1,t_1^i})+v_{i}(B_{i,1} \cap B_{k+1,t_2^i})+v_{i}(B_{i,1} \cap B_{k+1,t^*}) \ge 3v_{i}(B_{i,1} \cap B_{k+1,t^*}) > v_i(B_{i,1} \cap E_{i,k+1})
\] which is a contradiction.

\begin{center}
\begin{figure} 
\begin{center}
\begin{tikzpicture}[scale=0.65]
\small{
  \draw[step=1cm,] (0,-1) grid (8,6);
\draw[pattern={north west lines},pattern color=red](0,5) rectangle +(1,1);
\draw[pattern={north west lines},pattern color=red](0,4) rectangle +(1,1);
\node[anchor=north] at (0.5,6) {$t_1^1$};
\node[anchor=north] at (0.5,5) {$t_2^1$};

\draw[pattern={north west lines},pattern color=red](2,5) rectangle +(1,1);
\draw[pattern={north west lines},pattern color=red](2,2) rectangle +(1,1);
\node[anchor=north] at (2.5,6) {$t_1^2$};
\node[anchor=north] at (2.5,3) {$t_2^2$};

\draw[pattern={north west lines},pattern color=red](4,3) rectangle +(1,1);
\draw[pattern={north west lines},pattern color=red](4,0) rectangle +(1,1);
\node[anchor=north] at (4.5,4) {$t_1^3$};
\node[anchor=north] at (4.5,1) {$t_2^3$};

\draw[pattern={north west lines},pattern color=red](7,4) rectangle +(1,1);
\draw[pattern={north west lines},pattern color=red](7,1) rectangle +(1,1);
\node[anchor=north] at (7.5,5) {$t_1^k$};
\node[anchor=north] at (7.5,2) {$t_2^k$};

\node[anchor=north] at (0.5,7) {$B_{1,1}$};
\node[anchor=north] at (1.6,6.75) {$\ldots$};
\node[anchor=north] at (2.5,7) {$B_{2,1}$};
\node[anchor=north] at (3.6,6.75) {$\ldots$};
\node[anchor=north] at (4.5,7) {$B_{3,1}$};
\node[anchor=north] at (5.6,6.75) {$\ldots$};
\node[anchor=north] at (6.6,6.75) {$\ldots$};
\node[anchor=north] at (7.5,7) {$B_{k,1}$};

\node[anchor=west] at (-3.5,-0.5) {$B_{k+1,8\sqrt{k+1}}$};
\node[anchor=west] at (-3.5,0.5) {$B_{k+1,8\sqrt{k+1}-1}$};
\node[anchor=west] at (-0.75,3.5) {$\vdots$};
\node[anchor=west] at (-0.75,2) {$\vdots$};
\node[anchor=west] at (-2,4.5) {$B_{k+1,2}$};
\node[anchor=west] at (-2,5.5) {$B_{k+1,1}$};}
\end{tikzpicture}
    
\end{center}
\caption{This figure illustrates the set of edges $E_{k+1}^{\mathbf{P}}$. For agent $1$ and her selected bundle, $B_{1,1} \supset X_1^{(k+1)}$, the corresponding indices are $t_{1}^1=1$ and $t_{2}^1=2$. That is, $v_1(B_{k+1,1} \cap B_{1,1})\ge v_1(B_{k+1,2} \cap B_{1,1}) \ge v_1(B_{k+1,t} \cap B_{1,1}),$ $t \in \mathbf{I}_{k+1}(8\sqrt{k+1})\setminus \{1,2\}$. The corresponding indices for the remaining agents are as shown. Note that it may happen that some bundles are marked by more than one agent, bundles $B_{k+1,1}$ and $B_{k+1,2}$, while other may not be marked, bundle $B_{k+1,8\sqrt{k+1}}$. The set $S_{k+1}$ is the union of those intersections, $S_{k+1}=\bigcup_{i \in \mathbf{P}_{-k+1}} B_{i,1} \cap \left(B_{k+1,t^i_1} \cup B_{k+1,t^i_2}\right)$, shown in figure with red diagonal lines. Observe that by construction, $S_{k+1}$ has the second property of overconstrained set.}
\label{fig:markedXOS}
\end{figure}  
\end{center}

As a result, if for some index $t \in \mathbf{I}_{k+1}(8\sqrt{k+1})$ we have $v_{k+1}(S_{k+1,t}') \ge 2/3$ then we can extend the existing orientation and property (i) of the claim holds. If this is not the case, then due to additivity, $v_i(S_{k+1,t}) > 1/3$. Moreover from the construction we have that for every $i \in \mathbf{P}_{-k+1}$ there exist exactly two indices $t_1^i,t_2^i \in \mathbf{I}_i(8\sqrt{k+1})$ such that $S_i \cap B_{i,t_1^i} \cap E_{i,k+1}=S_{i,t_1^i} \cap E_{i,k+1} \ne \emptyset$ and $S_i \cap B_{i,t_2^i} \cap E_{i,k+1}= S_{i,t_2^i} \cap E_{i,k+1} \ne \emptyset$ and the properties of the overconstrained set follow for agent $k+1$ and set of edges $S_{k+1}$. In a similar way, we can prove that either $(i)$ we can extend an existing orientation or $(ii)$ ${\bf P}$ is overconstrained as all the properties hold for all agents.
\end{proof}



To complete the proof, we state the following claim which asserts a strong and interesting connection between independent transversal sets and the MMS problem with XOS agents.

\begin{claim} \label{clm:XOS2sqrt}
    Let $\mathbf{P}=\{1,\ldots,k+1\}$ be a set of $k+1$ agents and let $\mathbf{I}_i(8\sqrt{k+1})=\{1,\ldots,8\sqrt{k+1}\}$ with $i \in \mathbf{P}$ be $k+1$ subsets of indices. If $\mathbf{P}$ is overconstrained then Theorem~\ref{thm:ITS} implies a frugal orientation $X$ such that 
    $$v_{i}(X_{i}) \ge \frac{2}{3}v_{i}(B_{i,k_i} \cap E_{i}^{\mathbf{P}}),X_{i} \subseteq B_{i,k_i} \cap E_{i}^{\mathbf{P}}, \forall i \in \mathbf{P} \text{ and some }k_i \in \mathbf{I}_{k+1}(8\sqrt{k+1}).$$
\end{claim}

Claim~\ref{clm:XOS2sqrt} combined with Claim~\ref{clm:XOS1sqrt} proves the inductive hypothesis and hence the theorem.

\begin{proof}
    We show that the overconstrained set $\mathbf{P}$ can be represented via a multipartite graph with the desired properties such that there exists an $ITS$ of size $2$. This implies that each agent $i \in \mathbf{P}$ can get her corresponding high value intersections $S_{i,k_i}$ for some $k_i$ and, by Theorem~\ref{thm:XOS1outofd}, also at least $1/2$ of the remaining value in the corresponding canonical bundle. This guarantees at least $2/3$ value for each agent i.e.

    $$v_{i}(S_{i,k_i}) + \frac{1}{2}v_i(B_{i,k_i}\setminus S_{i,k_i}) \ge  \frac{1}{3}+\frac{1}{2}\left(1-\frac{1}{3}\right)=2/3$$

We define the $k+1$-partite graph $G_{\mathbf{P}}$, corresponds to the overconstrained set, with the partition $P=(V_1, V_2\ldots,V_{k+1})$ as follows:
\begin{enumerate}
    \item For each part we have $V_i=\left(v_{i,1},v_{i,2},\ldots,v_{i,8\sqrt{k+1}}\right)$ where each vertex $v_{i,t}$ corresponds to the canonical bundle $B_{i,t} \in B_{i}(\mathbf{I}_i(8\sqrt{k+1}))$. 
    \item For every pair of vertices $(v_{i,t},v_{i',t'})$ with $i\ne i'$ there exist an edge if and only if either (i) $S_{i,t} \cap B_{i',t'} \ne \emptyset$ or (ii) $S_{i',t'} \cap B_{i,t} \ne \emptyset$.
\end{enumerate}

That is, an edge $(v_{i,t},v_{i',t'})$ exists if and only if the corresponding set $S_{i,t}$ (or set $S_{i',t'}$) intersects the canonical bundle $B_{i',t'}$ (respectively bundle $B_{i,t}$). Observe that an independent transversal set $(v_{1,t_1},v_{2,t_2},\ldots,v_{k+1,t_{k+1}})$ in $G_{\mathbf{P}}$ yields a frugal orientation of the corresponding high value sets $S_{i,t_i}$ for each agent $i$ and some index $t_{i} \in \mathbf{I}_i(8\sqrt{k+1})$ i.e., the orientation is valid as none of the items is allocated multiple times. In the same way, a $2$-independent transversal set $(v_{1,t_1},v_{1,t_1'},v_{2,t_2},v_{1,t_2'},\ldots,v_{k+1,t_{k+1}},v_{k+11,t_{k+1}'})$ in $G_{\mathbf{P}}$ derives an orientation of the corresponding high value intersections $S_{i,t_i},S_{i,t_i'}$ for each agent $i$ and two indices $t_{i},t_i' \in \mathbf{I}_i(8\sqrt{k+1})$.

Figure~\ref{fig:overconstrainedgraph} presents the graph $G_{\mathbf{P}}$. 

Note that the graph $G_{\mathbf{P}}$ is $8\sqrt{k+1}$ thick and has maximum average degree $\bar{D}_P(G) \le \frac{4k}{8\sqrt{k+1}}$. That is, from the definition of the overconstrained set, for every pair of agents $(i,i')$ there exist two bundles with indices in $\mathbf{I}_i(8\sqrt{k+1})$ (and also two bundles with indices in $\mathbf{I}_{i'}(8\sqrt{k+1})$) intersecting the corresponding two bundles of agent $i'$ (respectively, the bundles of agent $i$) and thus at most $4$ edges between parts $V_i,V_{i'}$. Each agent forms $k$ pairs in total, and as a result, there are at most $4k$ of those intersections that correspond to some edge with one endpoint in $V_i$. See Figure~\ref{fig:overconstrainedgraph} for an illustration.

To prove the Claim, we obtain the following theorem from \cite{DaiLiuZhang2025}:

\begin{theorem}[Theorem~4.3 \cite{DaiLiuZhang2025}]\label{thm:ITS}
    Given a multipartite graph $G$ with a partition $P = (V_1, V_2,\ldots,V_r)$, if $\bar{D}_P(G) \le D$ and $P$ is $4s^2D$-thick, then $G$ has an (ITS) with size $s$ with respect to $P$. 
\end{theorem}

For $s=2$ the theorem~\ref{thm:ITS} follows as $4s^2D=16 \frac{k}{2\sqrt{k+1}} < 8 \sqrt{k+1}$ and thus there exists an ITS of size $2$ implying that each agent can get the high value intersections from at least two canonical bundles and no multi allocation will occur.

\begin{figure}
    \centering
    \begin{tikzpicture}[scale=0.65]

  \tiny{
\node[anchor=south west] at (2,3) {${\bf V_1}$};
\node[anchor=south west] at (-2.5,-2.75) {${\bf V_2}$};
\node[anchor=south west] at (2.25,-2.75) {${\bf V_3}$};

  \foreach \i in {1, 2, 3} {
    \node[draw, circle, fill=yellow!30] (B\i) at (-2.5 + 1.25*\i, 2.5) {$v_{1,\i}$};
  }
  \node[draw,
  circle, fill=yellow] (B5) at (-2 + 1.25*6, 2.5) {};

    \node[draw, circle, fill=yellow!0] (B4) at (-2 + 1.25*5, 2.5) {};
\node[anchor=south west] at (-2.5 + 1.25*4, 2.5) {$\boldsymbol{\ldots}$};

\node[anchor=south west] at (-2.5 + 1.25*5, 3) {$v_{1,8\sqrt{k+1}-1}$};

\node[anchor=south west] at (-2.5 + 1.25*6.8, 2.25) {$v_{1,8\sqrt{k+1}}$};


   Part 4: Bottom (D) - More south
  \foreach \i in {1, 2, 3} {
    \node[draw, circle, fill=red!30] (D\i) at ( 1.25*\i - 5, -1.5) {$v_{2,\i}$};
  }
      \node[anchor=south west ] at ( 1.25*4 - 5.5, -1.5){$\boldsymbol{\ldots}$};
  \draw[ultra thick, red] (B1) -- (D1);
  \draw[ultra thick, red] (B1) -- (D2);
  \draw[ultra thick, red] (B2) -- (D3);
  \draw[ultra thick, red] (B3) -- (D3);

    \foreach \i in {1, 2, 3} {
    \node[draw, circle, fill=blue!30] (C\i) at ( 1.25*\i, -1.5) {$v_{3,\i}$};
  }
      \node[anchor=south west ] at ( 1.25*4-0.5, -1.5){$\boldsymbol{\ldots}$};
  \draw[ultra thick, blue] (C1) -- (B2);
  \draw[ultra thick, blue] (C1) -- (B4);
  \draw[ultra thick, blue] (C2) -- (B5);
  \draw[ultra thick, blue] (C3) -- (B5);
\draw[fill=yellow!30] (-2 + 1.25*5, 2.5) circle (4.5mm); 
\draw[fill=yellow!30] (-2 + 1.25*6, 2.5) circle (4.5mm); 
}

\end{tikzpicture}

    \caption{The figure illustrates a subset of vertices and edges of $G_{\mathbf{P}}$. Parts $V_1,V_2$ and $V_3$ correspond to agents $1,2$ and $3$ respectively and each vertex $v_{i,t}$ corresponds to bundle $B_{i,t}$. As we can see, $X^{(2)}_{1} \subset B_{1,1}$ i.e., in the promised orientation $X^{(2)}$ for agents in $\mathbf{P}_{-2}$ agent $1$ receives a subset of her canonical bundle $B_{1,1}$ while the two indices she marked are $t_1^1=1$ and $t_{2}^1=2$, the corresponding edges from part $V_{1}$ to $V_{2}$. Similarly, $X^{(3)}_1 \subset B_{1,8\sqrt{k+1}}$ with $t_{1}^1=2$ and $t_{2}^1=3$. On the other hand, for the other agents we have $X^{(1)}_2 \subset B_{2,3}$ and $X^{(1)}_3 \subset B_{3,1}$ with the corresponding marked indices $t_{1}^2=t_{1}^3=2,t_{2}^2=3$ and $t_{2}^3=8\sqrt{k+1}-1$. As the set is overconstrained, each of the sets $S_{1,2}=B_{1,2} \cap \left(B_{2,3} \cup B_{3,1}\right),S_{1,3}=B_{1,3}\cap B_{3,1}$ and $S_{1,8\sqrt{k+1}-1} = B_{1,8\sqrt{k+1}-1}\cap B_{3,1}$ has a value of at least $1/3$ for agent $1$. For every independent transversal set, it holds that each of those edges, which is a superset of some $S_{i,t}$, can be oriented towards agent $i$ and guarantee a value of at least $1/3$ value.} 
        \label{fig:overconstrainedgraph}
\end{figure}

Without loss of generality, let $(v_{1,1},v_{1,2},v_{2,1},v_{2,2}, \ldots,v_{k+1,1},v_{k+1,2})$ be the promised ITS of size $2$. The analysis holds for every ITS of size $2$ by renaming. For the corresponding bundles restricted to the set of the common items $E^{\mathbf{P}}$ without the items corresponding to edges in $G_{\mathbf{P}}$, 
i.e., the set of items $A = E^{\mathbf{P}}\setminus \bigcup_{i \in \mathbf{P}}\left(S_{i,1} \cup S_{i,2}\right)$, we scale the value to one. That is, $v_i\left(B_{i,t} \cap A\right)=1\text{, for } i \in \mathbf{P}\text{, }t\in\{1,2\}$ and apply Theorem~\ref{thm:XOS1outofd} which yields a $1/2$-out-of-$2$ frugal orientation and w.l.o.g. let this orientation be $X=\left(X_1,X_2,\ldots,X_{k+1}\right)\text{, with }X_i \subseteq B_{i,1} \cap \left(E_{i}^{\mathbf{P}}\setminus S_{i}\right)$. By rescaling, each agent $i$ achieves at least half the restricted value of her first bundle, i.e., $v_i(X_i) \ge \frac{1}{2}v_i(B_{i,1} \cap A)$. We extend the orientation by oriented the corresponding edges $S_{i,1}$ towards agent $i$ which can be done as the bundles form an ITS in $G_{\mathbf{P}}$. Thus, due to additivity, each agent guarantees a value of at least
$$v_i(X_i \cup S_{i,1}) = v_i(S_{i,1}) + \frac{1}{2}v_i(B_{i,t_i} \cap A)= v_i(S_{i,1}) +\frac{1}{2}(1-v_i(S_{i,1}))=\frac{1}{2}+\frac{1}{2}v_i(S_{i,1}) \ge 2/3$$ As a result, the claim holds, which establishes the inductive hypothesis and also the theorem.

\end{proof}
\end{proof}

\begin{corollary}
In every multi-graph with $n \ge 64$ XOS agents, there exists a $2/3$-MMS orientation.
    
\end{corollary}

We remark that for a small number of agents the previous theorem does not admit a $2/3$-MMS orientation due to the constant factor. For that matter, we provide Theorem~\ref{thm:XOS4+} which states the existence of a $2/3$-MMS for any number of agents $n \ge 3$. On the negative side, the analysis of Theorem~\ref{thm:XOS4+} is much more challenging and moreover, for many agents, the guarantees are worse, i.e., $2/3$-out-of-$n$ instead of $2/3$-out-of-$8\sqrt{n}$. We also note that in a similar but more detailed analysis for the graph $G_{\mathbf{P}}$ we can use a stronger theorem, Theorem~4.8 from \cite{DaiLiuZhang2025}, and derive better bounds for the constant factor. However, we will not obtain the proof as the first half (Claim~\ref{clm:XOS1sqrt}) remains the same and the bounds are still $d \in \mathcal{O}(\sqrt{n})$.

\begin{theorem}[Using Theorem~4.8 from \cite{DaiLiuZhang2025}]
     
     In every multi-graph with $n$ XOS agents, there exists a $2/3$-out-of-$(6+2\sqrt{2n+7})$-MMS orientation.
\end{theorem}

\subsection{Combinatorial Barriers and the ITS Conjecture Connection}
\label{sec:XOSITS}
Given the significant slackness, specifically, a factor of $\sqrt{n}$ instead of $n$, a natural question is whether this method can be leveraged to obtain better results. Dai, Liu, and Zhang \cite{DaiLiuZhang2025} pose an open question regarding independent transversal sets: Do their theorems still hold if $s$ is replaced by $s^2$? In this section, we demonstrate that this question is closely connected to the problem of approximate MMS allocations for XOS agents with graphical valuations.

\begin{conjecture}[Conjecture on Theorem~$4.3$ \cite{DaiLiuZhang2025}]\label{thm:ITSconj}
    Given a multipartite graph $G$ with partition $P = (V_1, V_2, \ldots, V_r)$, if $\bar{D}_P(G) \le D$ and $P$ is $4sD$-thick, then $G$ contains an ITS of size $s$ with respect to $P$. 
\end{conjecture}

We establish a compelling connection between these two problems. Specifically, a positive resolution to this conjecture implies that, for any integer $\ell > 1$, there exists a threshold $n_\ell$ such that every instance involving $n \ge n_\ell$ XOS agents with graphical valuations admits a $(1-1/\ell)$-approximate MMS allocation. Conversely, if such an approximation guarantee is unattainable for an arbitrary number of agents, the conjecture is necessarily false. More formally:
\begin{theorem}\label{thm:connectionITS}
For the approximate MMS problem with $n$ XOS agents it holds that:
     
\begin{enumerate}
    \item If Conjecture~\ref{thm:ITSconj} holds for $4sD$-thick partitions, then for every $\ell$ there exist $c_\ell > 1$ and $ r_\ell \in [0,1)$ such that in every multi-graph with $n$ XOS agents, there exists a $(1-1/\ell)$-out-of-$c_\ell\cdot n^{r_\ell}$-MMS orientation.
    \item If there exists a constant $\alpha < 1$ such that an $\alpha$-MMS allocation is not guaranteed to exist for any instance with $n$ XOS agents with graphical valuations, then Conjecture~\ref{thm:ITSconj} does not hold for $4sD$-thick partitions. 
\end{enumerate}
    
\end{theorem}
\begin{proof}
    It is enough to prove (1); then statement (2) holds, as otherwise it contradicts statement (1). Hence, assume that Conjecture~\ref{thm:ITSconj} holds.
    
    We prove the following statement by induction on $\ell\ge 2$: for every $\ell$ there exist $c_\ell > 1$ and $ r_\ell \in [0,1)$ such that every multi-graph with $n$ agents admits a $(1-1/\ell)$-out-of-$d^n_\ell$-MMS frugal orientation, where $d^n_\ell=c_\ell\cdot n^{r_\ell}$.\footnote{It will always be the case where $c_\ell \cdot n^{r_\ell}$ is greater than $\ell$. We could also define $d^n_\ell=\max \{c_\ell\cdot n^{r_\ell},\ell\},$ which would directly imply the statement for the case of $n \le 3$.} More precisely, $c_\ell = \sqrt{4 \cdot c_{\ell-1} \cdot 2 \cdot (\ell-1)}$ and $r_\ell = \frac{r_{\ell-1}+1}{2}$ where $c_2 = 2$ and $r_2 = 0$. Observe that $d^n_\ell \in \mathcal{O}(n^{1-\epsilon}),$ which implies the existence of $n_{\ell}$ for which $n \ge d^n_{\ell}$ if $n \ge n_\ell$. \footnote{Note that if $r_\ell \ge 1$ then there is always $n<d_\ell$ which cannot imply MMS guarantee.} It is sufficient to set $n_\ell$ be the smallest $n$ such that $d^n_\ell \le n$ or
    $$d^n_\ell=c_\ell \cdot n^{r_\ell}\le n \implies c_\ell \le n^{1-r_\ell} \implies\sqrt[1-r_\ell]{c_\ell}\quad \le n. $$ 
    
    As a result, for every $n$ greater than or equal to $n_\ell=\sqrt[1-r_\ell]{c_\ell}$ there exists a $(1-1/\ell)$-MMS frugal orientation. 

    Note that we have already proven the base case, i.e., for $\ell=2$ this is true for $r_2=0$ and $c_2=2$ (Theorem~\ref{thm:XOS1outofd}). \footnote{For $\ell=3$ we have $c_3=8$ and $r_3=1/2$ and the inductive hypothesis also follows (Theorem~\ref{thm:sqrtXOS}) for a different constant $c_3$.}

    Assume that the inductive hypothesis holds for some $\ell$ and consider the case of $\ell+1$. We show that, following the same analysis as the proof of Theorem~\ref{thm:sqrtXOS} and applying Conjecture~\ref{thm:ITSconj} for a $4sD$ thick graph with $D=2\ell(n-1)$ and $s=c_{\ell+1}n^{r_{\ell+1}}$, we can guarantee that the hypothesis holds for $c_{\ell+1}=\sqrt{4\cdot c_{\ell}\cdot 2\ell}$ and $r_{\ell+1}=\frac{r_\ell+1}{2}$. This implies $n_{\ell+1}= \sqrt[1-r_{\ell+1}]{c_{\ell+1}}$ and thus the inductive hypothesis also holds for $\ell+1$.
    
    Let $G=(V,E)$ be a graph with $n \ge n_{\ell+1}$ vertices, $d^n_{\ell+1}=c_{\ell+1}\cdot n^{r_{\ell+1}}$ where $c_{\ell+1}=\sqrt{4\cdot c_{\ell}\cdot 2\ell}$ and $r_{\ell+1}=\frac{r_\ell+1}{2}$ \footnote{Observe that $d_\ell^k$ is an increasing function of both $\ell$ and $k$.} and $\mathbf{B}^{d_{\ell+1}}_i=(B_{i,1},\ldots,B_{i,d_{\ell+1}})$ denote the canonical $1$-out-of-$d^n_{\ell+1}$-MMS partition of agent $i$ for the edge set $E_i$. For some $k \le d^n_{\ell+1}$ let $\mathbf{I}_i(k) \subseteq\{1,2,\ldots,d^n_{\ell+1}\}$ be a subset of $k$ indices, i.e., $|\mathbf{I}_i(k)|=k$. We define $B_i(\mathbf{I}_i(k))$ as the set of the respective bundles $B_{i,t}$ with $t \in \mathbf{I}_i(k)$ i.e., $B_{i}(\mathbf{I}_i(k))=\{B_{i,t}\mid t \in \mathbf{I}_i(k)\}$.  For some subset of agents $\mathbf{P} \subseteq V$, we denote the set of agents $\mathbf{P}_{-i} =\mathbf{P} \setminus \{i\}\text{ for } i \in \mathbf{P}$. Also, let the sets $E_{i}^{\mathbf{P}}=E_{i} \cap \left(\bigcup_{j \in \mathbf{P}}E_{i,j}\right)$ and $E^{\mathbf{P}}= \bigcup_{i \in \mathbf{P}} E_{i}^{\mathbf{P}}$ be \emph{the set of the common edges of agent $i$} with agents in $\mathbf{P}$ and \emph{the set of the common edges} of agents in $\mathbf{P}$ respectively. We will refer to those sets as the set of the common edges, omitting $\mathbf{P}$ and agent $i$ when they are clear from context.
    
    We employ Lemma~\ref{lem:XOSadditive}, which asserts that it suffices to show the existence of a  $(1-1/(\ell+1))$-out-of-$d^n_{\ell+1}$-MMS frugal orientation in every multi-graph with $n \ge 2$ additive agents. We prove by induction on $k \in \{2, \ldots, n\}$ that for every subset $\boldsymbol{P}=(p_1,\ldots,p_k) \subseteq V$ of $k$ agents and for any fixed $k$ sets $\mathbf{I}_{p_1}(d_{\ell+1}^k), \ldots,\mathbf{I}_{p_k}(d^k_{\ell+1})$ of indices of size $d^k_{\ell+1}$ there is a frugal orientation $X$, restricted to the set of the common edges,  with respect to $\mathbf{B}^k=(B_{p_1}(\mathbf{I}_{p_1}(d^k_{\ell+1})), \ldots, B_{p_k}(\mathbf{I}_{p_k}(d^k_{\ell+1})))$  (the allocated bundle of each agent is a subset of some canonical bundle with index in the corresponding set) such that each agent is guaranteed at least $1-1/\ell$ of the restricted value that the corresponding canonical bundle has over the set of the common edges. That is, for all agents $p_i \in \mathbf{P}$, there exists a canonical bundle $B \in B_{p_i} ( \mathbf{I}_{p_i}(d^k_{\ell+1}))$ such that $$v_{p_i}(X_{p_i}) \ge (1-1/(\ell+1))v_{p_i}(B \cap E_{i}^{\mathbf{P}}),\quad X_{p_i} \subseteq B \cap E_{p_i}^{\mathbf{P}}$$
    
    We note that every  frugal orientation with respect to bundles $B_{p_i}(\mathbf{I}_{p_i}(d_{\ell+1}^k))$ for some agent $p_i \in \mathbf{P}$ and set of indices $\mathbf{I}_{p_i}(d^k_{\ell+1})$ is also a frugal orientation with respect to the canonical partition $B_{p_i}^{d_{\ell+1}^n}$. Furthermore, observe that in contrast with the proof of Theorem~\ref{thm:XOS4+}, each agent of the set achieves $1-1/(\ell+1)$ of the value which some bundle with admissible index has over the set of the common edges with the other agents in the set. This is crucial for our analysis as we cannot guarantee the existence of two $1-1/(\ell+1)$ frugal orientations in which each agent selects from a different canonical bundle. Intuitively, if the induction holds for some $k$, then we append the set with one more agent, agent $i$. Now, the allocation guaranteed to exist for the $k$ agents, has to be also extended with the set of the common edges of those agents and the new agent, agent $i$. We give priority to the first $k$ agents of the group over their common edges with agent $i$ in order to guarantee at least $1-1/(\ell+1)$ of their corresponding selected bundle also in the new set of items.
    
    For $k=n$ we have $\boldsymbol{P} =V$ and $E_{i}^{V}=E_i$ and as a result, for each agent $i\in V$ we have: $$v_{i}(X_{i}) \ge (1-1/(\ell+1))\min_{B \in B_i( \mathbf{I}_{i}(d_{\ell+1}^n))}\left\{v_{i}\left(B \cap E_i\right)\right\}=(1-1/(\ell+1))\min_{B \in \mathbf{B}_i^{d_{\ell+1}^n}}\left\{v_{i}\left(B \right)\right\}=(1-1/(\ell+1))\text{-out-of-}d_{\ell+1}^n,\forall i \in V$$ and the theorem follows. 

     \noindent{\bf Base case, $\bf k=2$}.  By Lemma~\ref{lem:2XOSdlower} we have that for $n=2$ there exists a $(1-1/(\ell+1))$-out-of-$(\ell+1)$ MMS orientation. Those results are tight (Theorem~\ref{thm:XOSupper}). Thus, combining the tightness with the induction hypothesis i.e., the statement is true for $\ell$ and two agents, we get that the corresponding number of bundles $d^2_\ell$ must be at least $\ell$ i.e.,
     $$\ell \le d_\ell^2 = c_\ell \cdot 2^{r_\ell} \le 2 c_\ell \implies c_\ell \ge \ell /2$$
     where the second inequality is true due to the fact that $r_\ell < 1$. This implies that $d_{\ell+1}^2=$ $c_{\ell+1} \cdot n^{r_{\ell + 1}}\ge  \sqrt{4\cdot c_{\ell}\cdot 2\ell}\ge 2 \sqrt{\frac{\ell}{2}\cdot 2\ell} =2\ell>\ell + 1$.  That is, for every pair of agents $(i,j)$ and every two sets of $\left\lceil d^2_{\ell+1} \right\rceil$ indices (note that $(\ell+1)$ indices are enough for the case of $2$ agents) $\mathbf{I}_i\left(\left\lceil d_{\ell+1}^2 \right\rceil\right)$ and $\mathbf{I}_j\left(\left\lceil d^2_{\ell+1} \right\rceil\right)$, scaling their value in each bundle over the set of the common edges to one i.e., the scaling is applied in each canonical bundle and $v_i'(e) = \frac{v_i(e)}{v_i(B_{i,k_i} \cap E_{i,j})}$ for $e \in B_{i,k_i}$ and $$v'_i(B_{i,k_i} \cap E_{i,j})=v'_j(B_{i,k_j} \cap E_{i,j})=1\text{, }\forall k_i \in \mathbf{I}_i\left(\left\lceil d^2_{\ell+1} \right\rceil\right),k_j \in \mathbf{I}_j\left(\left\lceil d^2_{\ell+1} \right\rceil\right)$$ Applying Lemma~\ref{lem:2XOSdlower} we can find a frugal orientation for which the inductive hypothesis holds. That is, $1-1/(\ell+1)$ value in the scaled bundle admits $1-1/(\ell+1)$ value of the bundle restricted to the value of their common edges as for the allocated bundle $X_{i} \subseteq B_{i,k_i} \cap E_{i,j}$ we have $$v_i(X_i)=v_i'(X_i)\cdot v_i(B_{i,k_i} \cap E_{i,j})=(1-1/(\ell+1))\cdot v_i(B_{i,k_i} \cap E_{i,j})$$

  \noindent{\bf Induction step}. Now, let's assume that the statement holds for all possible subsets $\mathbf{P}=\{p_1,\ldots,p_k\}\subset V$ of $k$ agents and for all possible vectors of $k$ subsets of indices $\mathbf{I}_{p_1}(d^k_{\ell+1}),\ldots, \mathbf{I}_{p_k}(d^k_{\ell+1})$, with $2\leq k <n$. We show hat the statement also holds for all possible subsets $\mathbf{P}=\{p_1,\ldots,p_{k+1}\}\subseteq V$ of $k+1$ agents and for all possible vectors of $k+1$ subsets of $d^{k+1}_{\ell+1}$ indices
    $\mathbf{I}_{p_1}(d^{k+1}_{\ell+1}),\ldots,
    \mathbf{I}_{p_{k+1}}(d^{k+1}_{\ell+1})$. W.l.o.g., we establish the induction step for subset $\mathbf{P}=\{1,\ldots,k+1\}$ and set of indices $\mathbf{I}_i(d^{k+1}_{\ell+1})=\{1,\ldots, d^{k+1}_{\ell+1}\}$, for all $i \in \{1,\ldots,k+1\}$; clearly then the statement holds for any other subset of agents and any other vector of indices of size $d^k_{\ell+1}$ by renaming. As in the case of $k=2$, for brevity, we assume that for each agent and each of her canonical bundles, its value over the common set of edges is scaled to one i.e., $v_i\left(B_{i,t} \cap E_{i}^{\mathbf{P}}\right) = 1, \forall i \in \mathbf{P},t \in \mathbf{I}_i(d^{k+1}_{\ell+1})$. Clearly, for some agent $i$, a frugal orientation $X_{i} \subseteq B_{i,t} \cap E_{i}^{\mathbf{P}}\text{, } t \
    \in \mathbf{I}_i(d^{k+1}_{\ell+1})$ with value at least $1-1/(\ell+1)$ in the scaled instance admits (by rescaling) a frugal orientation with value at least $1-1/(\ell+1)$ of the total value over $E_{i}^{\mathbf{P}}$ that agent $i$ has for her selected bundle $B_{i,t}$ and hence the inductive hypothesis will hold. 

  We divide the proof of the induction step into two key claims. The first claim (Claim~\ref{clm:XOS1sqrt}), using the inductive hypothesis, outlines the necessary conditions for the case of $k+1$ agents when no frugal orientation in which every agent achieves at least $1-1/(\ell+1)$ of the value exists. Definition~\ref{def:overconstrainedsqrtconj} summarizes these conditions in terms of what we call an \textit{overconstrained set}. The second claim (Claim~\ref{clm:XOS2sqrtconj}) demonstrates how overconstrained sets reveal a useful structure in the intersection of the canonical bundles, and uses results on the existence of independent transversal sets to construct a frugal orientation where each agent gets at least $1-1/(\ell+1)$ of the value.

We proceed with the definition of an overconstrained set.
 \begin{definition}[Overconstrained set]
     \label{def:overconstrainedsqrtconj}
Let $\mathbf{P}=\{1,\ldots,k+1\}$ be a set of $k+1$ agents and let $\mathbf{I}_i(d^{k+1}_{\ell+1})=\{1,\ldots, d^{k+1}_{\ell+1}\}$ with $i \in \mathbf{P}$ be $k+1$ subsets of indices. We say that the set ${\bf P}$ is {\em overconstrained} if for every agent $i$ there exists a set of edges $S_i = \bigcup_{t \in \mathbf{I}_i(d^{k+1}_{\ell+1})} S_{i,t}$ with $S_{i,t}\subseteq B_{i,t} \cap E_{i}^{\mathbf{P}}$, such that:

 \begin{enumerate}
 \item $v_i\left(S_{i,t}\right) \geq 1/(\ell+1)$, for all $t \in \mathbf{I}_i(d^{k+1}_{\ell+1})$ and
\item For every $j \in \mathbf{P}_{-i}$ there exist exactly $\ell$ indices $t_1^i,t_2^i,\ldots,t^i_{\ell} \in \mathbf{I}_i(d^{k+1}_{\ell+1})$ such that $S_{i,t_\theta^i} \cap E_{i,j} \ne \emptyset$ for all $\theta \le \ell$.

\end{enumerate}

\end{definition}

 Intuitively, the first property guarantees that for every frugal orientation of edges in $S_{i,t}$ for some index $t$ towards agent $i$ the value she achieves is high enough. More over, those bundles $S_{i,t}$ arise from the intersection of the corresponding canonical bundle $B_{i,t}$ and the extended orientation for the set of $k$ agents $\mathbf{P}_{-i}$. The second property guarantees that those edges cannot intersect "much" the remaining agents in $\mathbf{P}_{-i}$ as they are uniformly distributed over their common edges. We are now ready to state the following claim:

 \begin{claim} \label{clm:XOS1sqrtconj}
 Let $\mathbf{P}=\{1,\ldots,k+1\}$ be a set of $k+1$ agents and $\mathbf{I}_i(d^{k+1}_{\ell+1})=\{1,\ldots,d^{k+1}_{\ell+1}\}$ with $i \in \mathbf{P}$ be $k+1$ subsets of indices. Either (i) there exists a frugal orientation in which each agent gets at least $1-1/(\ell+1)$ value, or (ii) ${\bf P}$ is overconstrained.
 \end{claim}

 \begin{proof}[Proof of Claim~\ref{clm:XOS1sqrtconj}]
   If (i) holds we are done i.e., by rescaling the values the hypothesis holds for the given orientation. Otherwise, using the inductive hypothesis, for each agent $i$ we identify a set of edges $S_i \subseteq E_{i}^{\mathbf{P}}$ for which ${\bf P}$ is {\em overconstrained}.

To prove this, we focus on agent $k+1$ and we will carefully define a partition of her bundles over the set of her relevant edges in two parts. We show that the first one will serve as the desired union of the high value intersections and the second one will be the union of the remaining edges. That is, $$B_{k+1,t} \cap E_{k+1}^{\mathbf{P}}=S_{k+1,t} \cup S_{k+1,t}',S_{k+1,t} \cap S_{k+1,t}'=\emptyset$$ for all $t \in \mathbf{I}_{k+1}(d^{k+1}_{\ell+1})$ and we show the partition has the following property: if $v_{k+1}\left(S_{k+1,t}\right) \leq 1/(\ell+1)$ for some $t$ then we can extend an existing allocation—guaranteed to exist by the inductive hypothesis—for the set of $k$ agents and $d^{k}_{\ell+1}$ indices to the set of $k+1$ agents and $d^{k+1}_{\ell+1}$ indices. If this is not the case, we show that the two properties of the overconstrained set hold for $i=k+1$ and the set $S_{k+1}=\bigcup_{t \in \mathbf{I}_{k+1}(d^{k}_{\ell+1})} S_{k+1,t}$. We can proceed in a similar fashion by defining partitions for any bundle $B_{i,t}$ with $i, t\leq k+1$ and show that ${\bf P}$ is overconstrained.

By the inductive hypothesis, let $X^{(k+1)}$ be the promised frugal orientation for agents in $\mathbf{P}_{-k+1}$ with respect to $B_1(\mathbf{I}_1(d^{k}_{\ell+1})), \dots, B_{k}(\mathbf{I}_{k}(d^{k}_{\ell+1}))$, where $\mathbf{I}_i(d^{k}_{\ell+1}) = \{1, \dots, d^{k}_{\ell+1}\}\subseteq \mathbf{I}_i(d^{k+1}_{\ell+1})$. Recall that $d^{k}_{\ell+1}$ is an increasing function of $k$ and thus we can find such a subset. By frugality, each agent $i$ receives a subset of their canonical bundles and the orientation is also frugal with respect to the set of indices $\mathbf{I}_i(d^{k+1}_{\ell+1})$. Without loss of generality, assume $X_i^{k+1} \subseteq B_{i,1}$ and we get $v_i\left(X_i^{(k+1)}\right) \geq (1-1/(\ell+1)) \cdot v_i(B_{i,1} \cap E_{i}^{\mathbf{P}_{-k+1}})$ for all $i \in \mathbf{P}_{-k+1}$.

 Now, we argue that if we extend the orientation with bundle $X_{k+1} \subseteq B_{k+1,t}$ for some $t \in \mathbf{I}_{k+1}(d^{k+1}_{\ell+1})$ such that $v_{k+1}(X_{k+1}) \ge (1-1/(\ell+1))$ and moreover $v_i(B_{i,1} \cap X_{k+1} ) \le (1-1/(\ell+1)) \cdot v_i(B_{i,1} \cap E_{i,k+1})$ for all $i \in \mathbf{P}_{-k+1}$
then we are in case (i) of the Claim and the inductive hypothesis holds for the extended orientation $$X=\left(X^{(k+1)}_1 \cup \left(B_{1,1} \cap \left(E_{1,k+1}\setminus X_{k+1}\right) \right),\ldots,X^{(k+1)}_k \cup \left(B_{k,1} \cap \left(E_{k,k+1}\setminus X_{k+1}\right) \right),X_{k+1}\right).$$ Indeed, for the extended orientation the hypothesis holds for agent $k+1$. Furthermore, for each agent $i \in \mathbf{P} _{-k+1}$ there exists a partition of her selected bundle over the set of the common edges, i.e., set $B_{i,1}\cap E_{i}^{\mathbf{P}}$, into two parts such that the orientation guarantees at least $1-1/(\ell+1)$ in each part and hence, by additivity, guarantees $1-1/(\ell+1)$ value i.e., $$v_i(X_i)=v_i(X_i \cap E_{i,k+1} )+v_i(X_i\cap E_{i}^{\mathbf{P}_{-k+1}}) \ge (1-1/(\ell+1)) \cdot v_i(E_{i,k+1} \cap B_{i,1})+v_i(X_i^{(k+1)})= 1-1/(\ell+1).$$ As a result, for such an extended orientation, we can imply that, if there is no $1-1/(\ell+1)$ orientation then, for every index $t \in I_{k+1}(d^{k+1}_{\ell+1})$ agent $k+1$ cannot "leave" the intersections $B_{k+1,t} \cap B_{i,1}$ with some agent $i \in P_{-k+1}$ without losing much value and neither agent $i$ can do so. 

We try to construct such an extended orientation for all $t \in \mathbf{I}_{k+1}(d^{k+1}_{\ell+1})$ in which for the set of the allocated edges $X_{k+1} \subseteq B_{k+1,t}$ none of the agents $i \in \mathbf{P}_{-k+1}$ overestimates it i.e., $$v_i(X_{k+1} \cap B_{i,1}) \le (1/(\ell+1))\cdot v_i(B_{i,1} \cap E_{i,k+1}).$$  
The inability of finding such a bundle will imply the desired partition of the corresponding canonical bundle by setting $S'_{k+1,t}=X_{k+1}$ and $S_{k+1,t}=B_{k+1,t} \setminus X_{k+1}$. 

In order to extend the orientation it is sufficient to prove that agent $k+1$ values her bundle enough. To do so, each agent $i \in \mathbf{P}_{-k+1}$ marks the $\ell$ indices  in $I_{k+1}(d^{k+1}_{\ell+1})$, let them be $(t_1^i,t_2^i,\ldots ,t^i_\ell)$, for which the corresponding bundles of agent $k+1$ have the highest value over her selected bundle $B_{i,1}$. That is,
 $$t_1^i=\arg \max_{t \in \mathbf{I}_{k+1}(d^{k+1}_{\ell+1})}\{v_i(B_{i,1} \cap B_{k+1,t})\}$$
 and

$$t_\theta^i=\arg \max_{t \in \mathbf{I}_{k+1}(d^{k+1}_{\ell+1})\setminus\{t^i_1,\ldots,t^i_{\theta -1}\}}\{v_i(B_{i,1} \cap B_{k+1,t})\}$$
See Figure~\ref{fig:markedXOSconj} for an illustration. 

Let the desired set $S_{k+1}$ be the intersections of the canonical bundle each agent selected in orientation $X^{(k+1)}$ with the union of the marked bundles of the corresponding agent, i.e., $$S_{k+1}= \bigcup_{i \in \mathbf{P}_{-k+1},\theta \in [\ell]}B_{i,1} \cap B_{k+1,t_{\theta}^i}$$ and
 $$S'_{k+1,t}= (B_{k+1,t} \cap E_{k+1}^{\mathbf{P}}) \setminus S_{k+1}$$
 See Figure~\ref{fig:markedXOSconj} for an illustration. Now we argue that indeed, none of the agents overestimates any bundle $S'_{k+1,t}$. Assume otherwise and let $t^*$ and $i$ such that $v_i(S'_{k+1,t^*}\cap B_{i,1}) > (1/\ell ) \cdot v_i(B_{i,1} \cap E_{i,k+1})$. There cannot be the case $t^* \in \{t_1^i,\ldots,t_\ell^i\}$ since, by construction, we omit the corresponding intersections from the bundle i.e., $B_{i,t^*} \cap B_{i,1} \in S_{k+1}$ and thus the corresponding value for the intersection is zero as $S'_{k+1,t^*} \cap B_{i,1}= \emptyset$. Hence, there must be the case  $t^* \notin \{t_1^i,\ldots,t_\ell^i\}$. Due to additivity,  we have 
 \[
 \sum_{\theta = 1}^\ell v_{i}(B_{i,1} \cap B_{k+1,t_\theta^i})+v_{i}(B_{i,1} \cap B_{k+1,t^*}) \ge (\ell + 1) \cdot v_{i}(B_{i,1} \cap B_{k+1,t^*}) > v_i(B_{i,1} \cap E_{i,k+1})
 \] which is a contradiction.

\begin{center}
\begin{figure} 
\begin{center}
\begin{tikzpicture}[scale=0.65]
\small{
  \draw[step=1cm,] (0,-1) grid (8,6);
\draw[pattern={north west lines},pattern color=red](0,5) rectangle +(1,1);
\draw[pattern={north west lines},pattern color=red](0,4) rectangle +(1,1);
\node[anchor=north] at (0.5,6) {$t_1^1$};
\node[anchor=north] at (0.5,5) {$t_2^1$};

\draw[pattern={north west lines},pattern color=red](0,0) rectangle +(1,1);
\node[anchor=north] at (0.5,1) {$t_\ell^1$};

\draw[pattern={north west lines},pattern color=red](2,-1) rectangle +(1,1);
\node[anchor=north] at (2.5,0) {$t_\ell^2$};

\draw[pattern={north west lines},pattern color=red](2,5) rectangle +(1,1);
\draw[pattern={north west lines},pattern color=red](2,2) rectangle +(1,1);
\node[anchor=north] at (2.5,6) {$t_1^2$};
\node[anchor=north] at (2.5,3) {$t_2^2$};

\draw[pattern={north west lines},pattern color=red](4,0) rectangle +(1,1);
\draw[pattern={north west lines},pattern color=red](4,-1) rectangle +(1,1);

\node[anchor=north] at (4.5,1) {$t_{\ell - 1}^3$};
\node[anchor=north] at (4.5,0) {$t_{\ell}^3$};

\draw[pattern={north west lines},pattern color=red](7,4) rectangle +(1,1);
\draw[pattern={north west lines},pattern color=red](7,1) rectangle +(1,1);
\node[anchor=north] at (7.5,5) {$t_1^k$};
\node[anchor=north] at (7.5,2) {$t_\ell^k$};

\node[anchor=north] at (0.5,7) {$B_{1,1}$};
\node[anchor=north] at (1.6,6.75) {$\ldots$};
\node[anchor=north] at (2.5,7) {$B_{2,1}$};
\node[anchor=north] at (3.6,6.75) {$\ldots$};
\node[anchor=north] at (4.5,7) {$B_{3,1}$};
\node[anchor=north] at (5.6,6.75) {$\ldots$};
\node[anchor=north] at (6.6,6.75) {$\ldots$};
\node[anchor=north] at (7.5,7) {$B_{k,1}$};

\node[anchor=west] at (-3.5,-0.5) {$B_{k+1,d^{k+1}_{\ell+1}}$};
\node[anchor=west] at (-3.5,0.5) {$B_{k+1,d^{k+1}_{\ell+1}-1}$};
\node[anchor=west] at (-0.75, 2.25) {$\vdots$};
\node[anchor=west] at (-2,3.5) {$B_{k+1,3}$};
\node[anchor=west] at (-2,4.5) {$B_{k+1,2}$};
\node[anchor=west] at (-2,5.5) {$B_{k+1,1}$};}
\end{tikzpicture}
    
\end{center}
\caption{This figure illustrates the set of edges $E_{k+1}^{\mathbf{P}}$. For agent $1$ and her selected bundle, $B_{1,1} \supset X_1^{(k+1)}$, the corresponding indices are $t_{1}^1=1,t_2^1=2$ and $t_{\ell}^1=d^{k+1}_{\ell+1}-1$. That is, $v_1(B_{k+1,1} \cap B_{1,1})\ge v_1(B_{k+1,2} \cap B_{1,1}) \ge v_1(B_{k+1,2} \cap B_{1,d^{k+1}_{\ell+1}-1})  \ge v_1(B_{k+1,t} \cap B_{1,1}),t \in \mathbf{I}_{k+1}(d^{k+1}_{\ell+1})\setminus \{t^1_{1}, \ldots, t^1_{\ell}\}$. The corresponding indices for the remaining agents are as shown. Note that may happen that some bundles are marked by more than one agent, bundles $B_{k+1,1}$ and $B_{k+1,2}$, while other may not be marked, bundle $B_{k+1,3}$. The set $S_{k+1}$ is the union of those intersections, $S_{k+1}=\bigcup_{i \in \mathbf{P}_{-k+1}, \theta \in [ \ell ]} B_{i,1} \cap B_{k+1,t^i_\theta} $, shown in Figure with red diagonal lines. Observe that by construction, $S_{k+1}$ has the second property of overconstrained set.}
\label{fig:markedXOSconj}
\end{figure}  
\end{center}

 As a result, if for some index $t \in \mathbf{I}_{k+1}(d^{k+1}_{\ell+1})$ we have $v_{k+1}(S_{k+1,t}') \ge (1-1/(\ell+1))$ then we can extend the existing orientation and property (i) of the claim holds. If this is not the case, then due to additivity, $v_i(S_{k+1,t}) > 1/(\ell+1)$. Moreover from the construction we have that for every $i \in \mathbf{P}_{-k+1}$ there exist exactly $\ell$ indices $t_1^i,\ldots,t_\ell^i \in \mathbf{I}_i(d^{k+1}_{\ell+1})$ such that $S_i \cap B_{i,t_\theta^i} \cap E_{i,k+1}=S_{i,t^i_\theta} \cap E_{i,k+1} \ne \emptyset$ for all $\theta \in [\ell]$ and the properties of the overconstrained set follow for agent $k+1$ and set of edges $S_{k+1}$. In a similar way, we can prove that either $(i)$ we can extend an existing orientation or $(ii)$ ${\bf P}$ is overconstrained as all the properties hold for all agents.
\end{proof}

If Conjecture~\ref{thm:ITSconj} is true then we can prove the following claim.
\begin{claim} \label{clm:XOS2sqrtconj}
    Let $\mathbf{P}=\{1,\ldots,k+1\}$ be a set of $k+1$ agents and $\mathbf{I}_i(d^{k+1}_{\ell+1})=\{1,\ldots,d^{k+1}_{\ell+1}\}$ with $i \in \mathbf{P}$ be $k+1$ subsets of indices. If $\mathbf{P}$ is overconstrained then Conjecture~\ref{thm:ITSconj} admits a frugal orientation $X$ such that 
    $$v_{i}(X_{i}) \ge (1-1/(\ell+1))v_{i}(B_{i,k_i} \cap E_{i}^{\mathbf{P}}),X_{i} \subseteq B_{i,k_i} \cap E_{i}^{\mathbf{P}}, \forall i \in \mathbf{P} \text{ and some }k_i \in \mathbf{I}_{k+1}(d^{k+1}_{\ell+1})$$
\end{claim}

 Claim~\ref{clm:XOS2sqrtconj} combined with Claim~\ref{clm:XOS1sqrtconj} proves the inductive hypothesis 2hich completes the proof.

\begin{proof}
    We show that the overconstrained set $\mathbf{P}$ can be represent via a multipartite graph with the desired properties such that there exists a $ITS$ of size $d^{k+1}_{\ell}$. This implies that each agent $i \in \mathbf{P}$ can get its corresponding high value intersections $S_{i,k_i}$ from some $k_i$ and also, by the inductive hypothesis, at least $1-1/\ell$ of the remaining value in the corresponding canonical bundle. This guarantees at least $1-1/(\ell+1)$ value for each agent i.e.

    $$v_{i}(S_{i,k_i}) + \frac{1}{\ell}v_i(B_{i,k_i}\setminus S_{i,k_i}) \ge  \frac{1}{\ell+1}+\frac{1}{\ell}\left(1-\frac{1}{\ell+1}\right)=1-1/(\ell+1)$$

    We define the $k+1$-partite graph $G_{\mathbf{P}}$, corresponds to the overconstrained set, with the partition $P=(V_1, V_2\ldots,V_{k+1})$ as follows:
    \begin{enumerate}
        \item For each part we have $V_i=\left(v_{i,1},v_{i,2},\ldots,v_{i,d^{k+1}_{\ell+1}}\right)$ where each vertex $v_{i,t}$ corresponds to the canonical bundle $B_{i,t} \in B_{i}(\mathbf{I}_i(d^{k+1}_{\ell+1}))$. 
        \item For every pair of vertices $(v_{i,t},v_{i',t'})$ with $i\ne i'$ there exist an edge if and only if either (i) $S_{i,t} \cap B_{i',t'} \ne \emptyset$ or (ii) $S_{i',t'} \cap B_{i,t} \ne \emptyset$.
    \end{enumerate}

    That is, an edge $(v_{i,t},v_{i',t'})$ exists if and only if the corresponding set $S_{i,t}$ (or set $S_{i',t'}$) intersects the canonical bundle $B_{i',t'}$ (respectively bundle $B_{i,t}$). Observe that an $s$-independent transversal set $(v_{1,t_1^1},v_{1,t_1^2},\ldots,v_{1,t_1^s},v_{2,t_2^1},\ldots,v_{k+1,t_{k+1}}^s)$ in $G_{\mathbf{P}}$ yields an orientation of the corresponding high value sets $S_{i,t_i^1},S_{i,t_i^2},\ldots,S_{i,t_i^s}$ for each agent $i$ and set of $s$ indices $t_{i}^1,\ldots,t_i^s \in \mathbf{I}_i(d^{k+1}_{\ell+1})$ i.e., the orientation is valid as none of the items is allocated multiple times. 

     Note that the graph $G_{\mathbf{P}}$ is $d_{\ell+1}^{k+1}$ thick and has maximum average degree $\bar{D}_P(G) \le \frac{2k\ell}{d_{\ell+1}^{k+1}}$. That is, from the definition of the overconstrained set, for every pair of agents $(i,i')$ there exist $\ell$ bundles with indices in $\mathbf{I}_i(d^{k+1}_{\ell})$ (and also $\ell$ bundles with indices in $\mathbf{I}_{i'}(d^{k+1}_{\ell + 1})$) intersecting the corresponding $\ell$ bundles of agent $i'$ (respectively, the bundles of agent $i$) and thus at most \footnote{If is the case where those edges have an overlapped then there are $2 \ell - 1$ edges.} $2 \ell$ edges between parts $V_i,V_{i'}$. Each agent forms $k$ pairs in total, and as a result, there are at most $2k\ell$ of those intersections that correspond to some edge with one endpoint in $V_i$. 
    
    We claim that for $s=d_{\ell}^{k+1}$ the Conjecture~\ref{thm:ITSconj} follows as $4sD \le 4 \cdot  d_{\ell}^{k+1} \cdot \frac{2 k\ell}{d^{k+1}_{\ell+1}} \le d_{k+1}^{\ell+1}$. Indeed, by multiplying both sides of the inequality with $d_{\ell+1}^{k+1}$ we get
    \[
    \left(d^{k+1}_{\ell+1}\right)^2=\left(c_{\ell+1} \cdot (k+1)^{r_{\ell+1}}\right)^2=\left(\sqrt{4\cdot c_\ell \cdot 2 \cdot \ell} \cdot (k+1)^{(r_{\ell}+1)/2}\right)^2= 4\cdot c_\ell \cdot (k+1)^{r_\ell} \cdot 2 \ell (k+1) \ge 4 \cdot d^{k+1}_\ell \cdot 2\ell(k+1)
    \]
    where the last inequality holds by the inductive hypothesis i.e., $c_\ell \cdot (k+1)^{r_\ell}=d^{k+1}_\ell$. Thus there exist an ITS of size $d_{\ell}^{k+1}$ implying that each agent can get the high value intersections from at least $d_{\ell}^{k+1}$ canonical bundles and no multi allocation will occur.
    
    Without loss of generality let $(v_{1,1},v_{1,2},\ldots,v_{1,d^{k+1}_\ell},v_{2,1},v_{2,2}, \ldots,v_{2,d^{k+1}_\ell},\ldots v_{k+1,d^{k+1}_\ell-1},v_{k+1,d^{k+1}_\ell})$ be the promised ITS of size $d^{k+1}_\ell$. The analysis holds for every ITS of size $d_{\ell}^{k+1}$ by renaming. For the corresponding bundles restricted to the set of the common items $E^{\mathbf{P}}$ without the items corresponding to edges in $G_{\mathbf{P}}$, i.e., the set of items $A = E^{\mathbf{P}}\setminus \bigcup_{i \in \mathbf{P}}\bigcup_{k_i=1}^{d^{k+1}_\ell}S_{i,k_i}$, we scale the value to one. That is, $v_i\left(B_{i,t} \cap A\right)=1, i \in \mathbf{P},t\in\{1,2,\ldots,d^{k+1}_\ell\}$ and apply the inductive hypothesis for $k+1$ agents and $1-1/\ell$ approximation which admits a $(1-1/\ell)$-out-of-$d_{\ell}^{k+1}$ frugal orientation and w.l.o.g. let this orientation be $X=\left(X_1,X_2,\ldots,X_{k+1}\right),X_i \subseteq B_{i,1} \cap A$. By rescaling, each agent $i$ achieves at least half the restricted value of her first bundle, i.e., $v_i(X_i) \ge \frac{1}{2}B_{i,1} \cap \left(E_{i}^{\mathbf{P}}\setminus S_{i}\right)$. We extend the orientation by oriented the corresponding edges $S_{i,1}$ towards agent $i$ which can be done as the bundles form an ITS in $G_{\mathbf{P}}$. Thus, due to additivity, each agent guarantees value at least
$$v_i(X_i \cup S_{i,1}) = v_i(S_{i,1}) + (1-\frac{1}{\ell})v_i(B_{i,t_i} \cap A)= v_i(S_{i,1}) +(1-\frac{1}{\ell})(1-v_i(S_{i,1}))=1-\frac{1}{\ell}+\frac{1}{\ell}v_i(S_{i,1}) \ge 1-1/(\ell+1)$$ As a result, the claim holds, which establishes the inductive hypothesis.
\end{proof}
\end{proof}

Thus, given $n$ we can solve $c_\ell n^{r_\ell} \le n$ which implies the following lower bound for MMS:

\begin{corollary}\label{cor:conj}
    If Conjecture~\ref{thm:ITSconj} is true then there exists a $(1-\frac{1}{\Omega ( \log \log n)})$-MMS approximation.
\end{corollary}

\subsection{Upper Bound}\label{sec:upper-bound-xos}
In this section, we present impossibility results for $n$ agents and $d$ bundles. For the special case of $d=n$, the following upper bound complements the positive results of Theorem~\ref{thm:XOS4+} and provides a tight $2/3$ MMS guarantee for the case of $3$ and $4$ agents. Moreover, for $n \in \{2,3\}$ the upper bounds match the lower bounds established in Theorem~\ref{thm:twothreeXOS} while for $d=2$, they match the bounds from Theorem~\ref{thm:XOS1outofd}.

\begin{theorem}\label{thm:XOSupper}
        A ($1-1/\Delta_n(d)+\varepsilon$)-out-of-$d$ MMS allocation in multi-graphs with $n$ XOS agents is not guaranteed to exist for any $\varepsilon > 0$ where

    \[
    \Delta_n(d) = \begin{cases}
        \lceil \frac{n-1}{2(n-2)}d \rceil & \text{if $n$ is odd} \\
         \lceil \frac{n}{2(n-1)}d \rceil & \text{if $n$ is even}
    \end{cases}
    \]
    In particular, a $\left(1-\frac{1}{\left\lceil n/2 \right\rceil+1}+ \varepsilon\right)$-MMS allocation for $n$ XOS agents, is not guaranteed to exist for any $\varepsilon>0$.
\end{theorem}
\begin{proof}
We construct a counterexample for $n$ agents. The valuations are defined such that a single edge has one of two possible values, i.e., $v_i(e) \in \{0, 1/b\}$ for $b=\Delta_n(d)$. The construction ensures that, if an agent $i$ receives a bundle $X_i$ satisfying $v_i(X_i) \geq 1-1/\Delta_n(d) + \varepsilon$ for some $\varepsilon > 0$, then she receives a whole canonical bundle i.e. $X_i = B_{i,t_i^*}$ for some $t_i^* \le d$. However, the structure of the edges prevents all agents from simultaneously achieving this condition, proving that no allocation can guarantee an approximation factor strictly greater than $1-1/\Delta_n(d)$.

We show this argument using a combinatorial result on independent transversal. In \cite{SzaboTardos2003}, Construction 3.3 is shown that for given integers $k,d$ and $b$, with $b \geq \left\lceil \frac{kd}{2k-1}\right\rceil=\Delta_n(d) $, there exists a construction of a $2k$-partite graph $G_{2k}$ with partition $P = \{V_1, \dots, V_{2k}\}$, where $|V_i| = d$, and maximum degree at most $b$ such that no independent transversal exists. Later, in \cite{PennyHaxellSzabo2006oddtransversals} it is proved that the same maximum degree also holds for $n=2k-1$ parts i.e. if $b \geq \left\lceil\frac{(2k-1)d}{2(2k-2)}\right\rceil=\Delta_n(d)$, there exists a construction of a $(2k-1)$-partite graph $G_{2k-1}$ with partition $P = \{V_1, \dots, V_{2k-1}\}$ and $|V_i| =d$, and maximum degree at most $b$ such that no independent transversal exists. 

In the rest of the proof we show that the absence of an independent transversal in the $2k$-partite graph implies an upper bound of $(1 - 1/b)$-out-of-$d$ MMS for an even number of agents $n=2k$ and the same analysis also holds for an odd number of agents $n=2k-1$. Thus, we can set $b = \Delta_n(d)$ ensuring the desired upper bound.

Let $G=(V,E)$ be a graph with $2k$ agents and a set of edges $E$. Consider the constructed $2k$-partite graph $G_{2k}$ of \cite{SzaboTardos2003}, Construction 3.3 of maximum degree at most $b$, together with a vertex set partition into $2k$ disjoint subsets $V_1,\ldots,V_{2k}$ of size $|V_i|=n=2k,i\in \{1,\ldots,2k\}$ with $V_{i}=\{v_{i,1},\ldots,v_{i,2k}\}$. We associate every edge $(v_{i,t_i}, v_{j,t_j})$ in $G_{2k}$ with an edge $e(i,t_i,j,t_j) = (i,j)$ in $G$. For each vertex $v_{i,t_i}$ in graph $G_{2k}$ with degree $d(v_{i,t_i}) < b$ we additionally add $b-d(v_{i,t_i})$ self-loops $e(i,t_i,i,t_i)=(i,i)$. We denote the additive functions
\[
a_{i^*,t_i^*}(e(i,t_i,j,t_j))=\begin{cases}
    1/b, & \text{if } i=i^* \text{ and } t_i=t^*_i  \\
    0, & \text{otherwise}\\
\end{cases}
\]
that is, $a_{i^*t^*}(e)=1/b$ if and only if edge $e$ corresponds to an edge connecting vertex $v_{i^*,t_i^*}$ in $G_{2k}$ (or to a corresponding self loop). We denote the valuation $v_i(S) = \max_{t\le d}\{a_{i,t}(S)\}, S \subseteq E$ of agent $i$ which is XOS i.e. maximum over additive. For each agent $i$, there are exactly $n$ disjoint bundles of edges $B_{i,t_i},t_i \in \{1, \dots, d\}$, with value $1$, namely $B_{i,t_i} = \Bigl\{e(i,t_i,j,t_j) : j>i, t_j \in \{1, \dots, d\}\Bigr\} \cup  \Bigl\{e(j,t_j,i,t_i) : j < i, t_j \in \{1, \dots, d\}\Bigr\}$.
By construction, each set $B_{i,t_i}$ has exactly $b$ edges. 

The partition $B_i = (B_{i,1}, \dots, B_{i,d})$ is the unique $1$-out-of-$d$ MMS partition for agent $i$, with $\mu_i^d(E) = 1$. To see this, suppose there exists another MMS partition $X = (X_1, \dots, X_d)$ such that $v_i(X_{t_X}) = 1$ for every $t_X$, but $X$ is not a reordering of $B_i$. Then, for every $t_X \in \{1, \dots, d\}$, there exists some $B_{i,t}$ such that $X_t \supseteq B_{i,t}$ with the containment being strict for at least some $t^*$, but this is impossible.

Now, assume there exists a $X = (X_1, \dots, X_d)$ $\left(1-1/b + \varepsilon\right)$-MMS allocation. By construction, for each agent $i$, there must exist a bundle $B_{i,t_i^*} \subseteq X_i$ i.e. agent $i$ receives all edges $e \in B_{i,t_i^*}$. Let $(B_{1,t_1^*}, \dots, B_{n,t_n^*})$ be the set of MMS bundles allocated wholly to agents in $X$. Since such an allocation exists, no edge can be assigned to more than one agent which admits $B_{i,t_i^*} \cap B_{j,t_j^*} = \emptyset$ for all allocated pairs $(B_{i,t_i^*}, B_{j,t_j^*})$. By construction, each MMS bundle corresponds to a vertex in $G_{2k}$, and each shared edge among these bundles corresponds to an edge in $G_{2k}$. Consequently, the existence of such an allocation would imply the existence of an independent transversal in $G_{2k}$ —a contradiction to the result of \cite{SzaboTardos2003}.
\end{proof}

\subsection{PMMS for XOS Valuations}\label{sec:PMMS-XOS}
In this section, we consider PMMS allocations for XOS agents. For the positive result, which also holds for the more general class of subadditive valuations, we refer the reader to Lemma~\ref{lem:subadditivePMMS}. Here, we provide a tight impossibility upper bound of $1/2$ for XOS valuations, which therefore also holds for subadditive valuations.

\begin{lemma}\label{lem:PMMSXOS}
  There exists a multi-graph with XOS agents in which there exists a PMMS allocation but not $(1/2+\varepsilon)$-PMMS orientation for any $\varepsilon > 0$.
\end{lemma}

\begin{proof}
We present the proof for the upper bound for the case of three agents and we can extend the counter example by adding disconnected vertices (agents). Clearly this will not change the result as in every orientation all the disconnected vertices receive the empty bundle.

Let the graph $G=(V,E)$ where $V=\{1,2,3\}$ is the set of three XOS agents and, $E=\{e_{1,1}, e_{1,2}, e_{2,1}, e_{2,2}\}$ is the set of four edges with endpoints $(1,2)$. We define the additive functions
\[
a_{t}(e_{i,j}) = \begin{cases}
1/2, & \text{if } i=t \\
0, & \text{otherwise}
\end{cases} \]
and
\[
b_{t}(e_{i,j}) = \begin{cases}
1/2, & \text{if } j=t \\
0, & \text{otherwise}
\end{cases}
\]
For $S \subseteq E$, the valuations are $v_1(S) = \max_{t \le 2}\{a_{t}(S)\}$, $v_2(S) = \max_{t \le 2}\{b_{t}(S)\}$, and $v_3(S) = 0$, which are XOS (maximum over additive). Assume that there exists a $(1/2 + \varepsilon)$-PMMS orientation $X' = (X'_1, X'_2, X'_3)$; that is, all edges are distributed among agents 1 and 2. For these agents, $PMMS_1 = \mu_1^2(E) = 1$ and $PMMS_2 = \mu_2^2(E) = 1$, so $\{e_{i,1}, e_{i,2}\} \subseteq X'_1$ for $i \in \{1,2\}$ and $\{e_{1,j}, e_{2,j}\} \subseteq X'_2$ for $j \in \{1,2\}$; otherwise, at least one agent will achieve no more than $1/2$-PMMS. In this case, for every possible pair $(i,j)$, edge $e_{i,j}$ is allocated to both agent 1 and agent 2, contradicting the assumption.

The allocation $X = (\{e_{1,1}, e_{1,2}\}, \{e_{2,1}\}, \{e_{2,2}\})$ is PMMS; that is, $v_i(X_i) = PMMS_i$ for all $i \in V$. Therefore, the lemma follows.
\end{proof}

\section{Subadditive valuations}
\label{sec:subadditive}
In this section, we show that for multi-graphs with subadditive valuations, there always exists a $1/2$-MMS allocation, and this guarantee is tight (Theorem~\ref{Thm:subadditivemulti}). In the general (non-graphical) setting, the best previously known lower bound was $1/\mathcal{O}(\log \log n )$, as shown by \cite{Feige25}, while an upper bound of $1/2$ was established in \cite{GhodsiHSSY22}. Furthermore, we show a tight approximation of $1/2$-PMMS orientations (Theorem~\ref{thm:subadditivePMMS}).

\subsection{Approximate MMS for Subadditive Valuations}
In this section, we present the two lemmas (lower and upper bound) that will establish our results concerning the tight bounds for approximate MMS for subadditive agents, culminating in Theorem~\ref{Thm:subadditivemulti}.

\begin{lemma}
\label{lem:subadditivemulti}
In every multi-graph with $n$ subadditive agents, there exists a $1/2$-MMS orientation.
\end{lemma}

\begin{proof}

  Let $B_i=(B_{i,1},\ldots B_{i,n})$ be the canonical MMS partition of agent $i$ for $E_i$. Let $\mathbf{I}_i(k)$ be a subset of $k$ indices, for some $k\leq n$, i.e., $\mathbf{I}_i(k)\subseteq \{1,\ldots, n\}, |\mathbf{I}_i(k)|=k$. We denote by $B_i(\mathbf{I}_i(k))$ as the set of the respective bundles $B_{i,t}$ with $t\in \mathbf{I}_i(k)$, i.e. $B_i(\mathbf{I}_i(k))=\left(B_{i,t}| t\in \mathbf{I}_i(k)\right).$

  We will show by induction on $k\in \{1,\ldots, n\}$, that for agents
  $1,\ldots, k$ and for any $k$ fixed sets of indices of size $k$,
  $\mathbf{I}_1(k),\ldots, \mathbf{I}_k(k)$, there is a frugal (see definition~\ref{def:frugal}) $1/2$-MMS orientation $X$.
  
    \noindent{\bf Base case, $\bf k=1$}. The base case is immediate; let $\mathbf{I}_1(1)=t\in \{1\ldots, n\}$, then we set $X_1=B_{1,t}$.
    
    Now let's assume that the statement holds for the first $k<n$ agents and for all possible vectors of $k$ subsets of indices $\mathbf{I}_1(k),\ldots, \mathbf{I}_k(k)$.
    
  \noindent{\bf Induction step}.  We will show that the statement also holds for the first $k+1$
  agents and for all possible vectors of $k+1$ subsets of indices $\mathbf{I}_1(k+1),\ldots, \mathbf{I}_{k+1}(k+1)$. W.l.o.g. we will focus on
  $\mathbf{I}_i(k+1)=\{1,\ldots, k+1\}$, for all $i \in \{1,\ldots,k+1\}$; the statement holds for any other vector of indices of size $k+1$ by renaming.

  By the induction hypothesis let $X$ be the promised orientation for agents $1,\ldots, k$ on $\mathbf{I}_1(k),\ldots, \mathbf{I}_k(k)$ with $\mathbf{I}_i(k)=\{1,\ldots, k\}$ and $v_i(X_i) \ge \frac{1}{2}\mu_i$ for all agents. By frugality each agent $i\leq k$ gets a subset of their canonical bundles (w.l.o.g. it is
  $X_i\subseteq B_{i,1}$).

  Now, let's focus on indices $\mathbf{I}'_1(k),\ldots, \mathbf{I}'_k(k)$ with
  $\mathbf{I}'_i(k)=\{2,\ldots, k+1\}$ and let $X'$ be the respective promised
  $1/2$-MMS frugal orientation for agents $1,\ldots, k$ with
  $v_i(X'_i) \ge \frac{1}{2}\mu_i$.
  
  We will now show how we can extend one of the partial orientations $X$ or $X'$ by assigning
  also a subset of items to agent $k+1$. Consider the $t$-th
  canonical bundle $B_{k+1,t}$ of agent $k+1$. Take any
  $e=(i,j)\in X\cap X'$ and assume that $e$ is oriented towards $i$ in
  $X$, i.e. $e\in X_i\subseteq B_{i,1}$. By construction of $X'$, item
  $e$ should be oriented towards agent $j$ in $X'$ and therefore
  $X\cap X' \cap B_{k+1,t}=\emptyset$, i.e., the edge is irrelevant to
  agent $k+1$.  If
  $v_{k+1}(B_{k+1,t}\setminus X)\geq v_{k+1}(B_{k+1,t})/2$ then the
  resulting orientation is $X$ extended with
  $X_{k+1}=B_{k+1,t}\setminus X$. Otherwise, by subadditivity $v_{k+1}(B_{k+1,t}\cap X)\geq v_{k+1}(B_{k+1,t})/2$ and the resulting orientation
  is $X'$ extended with $X'_{k+1}=B_{k+1,t}\cap X$. The extension is feasible because $X\cap X' \cap B_{k+1,t}=\emptyset$. 
In either case, it is clear that the resulting orientation is $1/2$-MMS and satisfies the property of frugality. 
 \end{proof}
Next, we complement the lower bound of Lemma~\ref{lem:subadditivemulti} by providing a matching upper bound. We could also obtain the same result using the upper bound of \cite{GhodsiHSSY22} for two agents and extending the counter example by adding $n-2$ disconnected vertices and the inapproximability follows i.e. one agent cannot guarantee more than $1/2$ or her MMS value while the other agents can guarantee their full MMS value. We present our counter example in which at most one agent guarantee her full MMS value while none of the other agents can guarantee more than $1/2$ of their MMS value.

\begin{lemma}\label{lem:subadditiveupper}
        A $(1/2+\varepsilon)$-MMS allocation in multi-graphs with $n$ subadditive agents, is not guaranteed to exist for any $\varepsilon>0$.
\end{lemma}
\begin{proof}
 We construct an instance $G=(V,E)$ with $n$ vertices (agents) and $\binom{n}{k}n^2$ edges. Each edge is denoted by $e(i,t_i,j,t_j)$ for $1\leq i<j\leq n$ and for $t_i,t_j\in \{1,\ldots, n\}$.
    
    For each pair $i<j$, the set $E_{i,j}$ consists of $n^2$ edges i.e., $E_{i,j}=\left\{e(i,t_i,j,t_j): t_i,t_j\in \left\{1,\ldots, n \right\}\right\}$.

    For every agent $i$, there are exactly $n$ disjoint bundles of edges $B_{i,t_i}$, $t_i\in\{1,\ldots, n\}$ that can give her value equal to 1. 
    The set $B_{i,t_i}$ is defined as follows:
    $
    B_{i,k_i}=\Bigl\{e(i,t_i,j,t_j): j>i, t_j=\{1,\ldots, n\}\Bigr\} \bigcup \Bigl\{e(i,t_i,j,t_j): j < i, t_j=\{1,\ldots, n\}\Bigr\}$.
  The valuation $v_i(S)$ of each agent $i\in \{1,\ldots, n\}$ and bundle $S \subseteq E$ is defined as follows
     \[
    v_i(S) =   \begin{cases}
      1, & \text{ if }\exists k_i :S \supseteq B_{i,t_i} \\
      1/2, & \text{otherwise.}
    \end{cases}
  \]

  It is easy to verify that the valuations are subadditive.  Note that
  the partition $B_i=(B_{i,1},\ldots,B_{i,n})$ is the unique MMS partition
  for agent $i$, with
  $\mu_i^n=1$. To see this, suppose that there is another partition $X=(X_1,\ldots, X_n)$, with $v_i(X_{t_X})=1$ for every $t_X$,  which is not a reordering of $B_i$. Then, for every $t_X\in \{1,\ldots, n\}$, there exists some $B_{i,t}$ such that $X_{t_X}\supseteq B_{i,t}$ with the containment being strict for at least some $t^*$. But this is impossible.

Observe that by construction, each pair of bundles $B_{i,t_i}, B_{j,t_j}$ with $i<j$ must intersect i.e., $B_{i,t_i} \cap B_{j,t_j}=e(i,t_i,j,t_j)$. 
We now argue that there is no allocation $X=(X_1,\ldots, X_n) $ that guarantees value $v_i(X_i)>1/2$ for each agent $i$. Take two agents $i<j$.
First note, that $v_i(X_i)>1/2$ if and only if $X_i\supseteq B_{i,t_i}$ for some
$k_i$. We will show that $v_j(X_j)\leq 1/2$. This is because $X_i$ intersects with all good bundles $B_{j,t_j}$ of agent $j$, as item $e(i,t_i,j,t_j)\in B_{i,k_i}\subseteq X_i$ hence it isn't available in $X_j$.
\end{proof}

By combining Lemmas~\ref{lem:subadditivemulti} and \ref{lem:subadditiveupper} we obtain the main result of the section.
\begin{theorem}\label{Thm:subadditivemulti}
    In every multi-graph with $n$ subadditive agents there exists $1/2$-MMS orientation. Furthermore, an $(1/2+\varepsilon)$-MMS allocation is not guaranteed to exist, for any $\varepsilon >0$. 
\end{theorem}

\subsection{Approximate PMMS for Subadditive Valuations}
In this subsection we show that the the algorithm described in Theorem~\ref{thm:multiaddcutandchoose} guarantees a $1/2$-PMMS orientation for subadditive agents. Combined with lemma~\ref{lem:PMMSXOS} derived our main results of a tight $1/2$-PMMS orientation (Theorem~\ref{thm:subadditivePMMS}).
\begin{lemma}\label{lem:subadditivePMMS}
In every multi-graph with $n$ subadditive agents there exists a $1/2$-PMMS orientation.
\end{lemma}

\begin{proof}

Let $X = (X_1, \ldots, X_n)$ be the orientation produced by applying the cut and choose algorithm for each pair of agents $(i,j)$ and their common set of items $E_{i,j}$ (see proof of Theorem~\ref{thm:multiaddcutandchoose} for more details). We will show that for any agents $i, j$, we have $v_i(X_i) \geq \frac{1}{2} \mu_i(X_i \cup X_j)$ and thus $v_i(X_i) \geq \frac{1}{2}\operatorname{PMMS}_i(X)$.  

\noindent \textbf{Case 1:} If $i < j$, then agent $i$ chooses her most valuable bundle and also receives some additional edges (from common edges with other agents) and thus $v_i(X_i) \ge v_i(X_j)$. By subadditivity, agent $i$ guarantees at least half the value of the combined bundle, i.e., 
$v_i(X_i) =\max\{v_i(X_i),v_i(X_j) \} \geq \frac{v_i(X_i \cup X_j)}{2} \geq \frac{\mu_i^2(X_i \cup X_j)}{2}$.
The last inequality follows also from subadditivity, which ensures that the pairwise valuation is at most the total value.

\noindent \textbf{Case 2:} If $j < i$, let $P = (A, B)$ be the $1$-out-of-$2$ MMS partition of the common edges $E_{i,j}$ among agents $i$ and $j$. Without loss of generality, assume that $v_i(A) \geq v_i(B) \geq a$. If the agent $i$ gets bundle $A$, then by the same analysis as in Case 1, we obtain $v_i(X_i) \geq \frac{\mu_i^2(X_i \cup X_j)}{2}$.
Thus, consider the case where agent $i$ is allocated bundle $B$.  

Now, assume for the sake of contradiction that after in the final allocation, there exists an agent $j$ such that $v_i(X_i) < \frac{\mu_i^2(X_i \cup X_j)}{2}$. By monotonicity, we have $v_i(X_i) \geq v_i(X_i \cap E_{i,j}) = a$.

Let $P' = \{A', B'\}$ be the $1$-out-of-$2$ MMS partition of the edges in $X_i \cup X_j$ by agent $i$. Since $\mu_i^2(E_{i,j}) = a$, define the set
$S = \arg \min_{S \in \{A', B'\}} v_i(S \cap E_{i,j})$, which ensures that $v_i(S \cap E_{i,j}) \leq a$. Otherwise, the partition $(A' \cap E_{i,j}, B' \cap E_{i,j})$ would guarantee a value strictly greater than $a$ for the less valuable bundle, contradicting the fact that $P$ is the $1$-out-of-$2$ MMS partition of edges in $E_{i,j}$.  

Consider the partition of $S$ into two bundles: $(S \cap E_{i,j}, S \setminus E_{i,j})$. Since agent $i$ did not achieve at least half of her PMMS value, we obtain 
\begin{align} 
\label{eq:sub1}
    2v_i(X_i) < \min\{v_i(A'), v_i(B')\} \leq v_i(S).
\end{align}
 By subadditivity, it follows that $v_i(S) \leq v_i(S \cap E_{i,j}) + v_i(S \setminus E_{i,j})$.
Since we have already established an upper bound of $a$ for the value of edges in $S \cap E_{i,j}$, we derived that 
\begin{align} \label{eq:sub2}
    v_i(S) \leq a + v_i(S \setminus E_{i,j}).
\end{align}
Note that the $X$ is an orientation, and thus $v_i(X_j) = v_i(X_j \cap E_{i,j}) \le v_i(E_{i,j})$. Hence, we conclude that  
\begin{align}\label{eq:sub3}
 v_i(S \setminus E_{i,j}) \leq v_i(S \setminus X_j) = v_i(S \cap X_i) \leq v_i(X_i).   
\end{align}

Combining equations (\ref{eq:sub1}),(\ref{eq:sub2}) and (\ref{eq:sub3}), we get  
$2v_i(X_i) < a + v_i(X_i)$ which simplifies to  
$v_i(X_i) < a$ which is a contradiction. As a result there does not exist such an agent $j$ and agent $i$ achieves at least half her PMMS value.
\end{proof}
 Combining Lemmas~\ref{lem:PMMSXOS} and \ref{lem:subadditivePMMS} we obtain the main result of this section.
 
\begin{theorem}\label{thm:subadditivePMMS}
In every multi-graph with $n$ subadditive agents, there exists a $1/2$-MMS orientation. However, there exists a graph with XOS agents where there is not a $(1/2 + \varepsilon)$-PMMS orientation for any $\varepsilon > 0$ while an exact PMMS allocation exists.
\end{theorem}

\section{Conclusions and Open Problems}

In this paper, we have provided a comprehensive study of the maximin share (MMS) fairness guarantee within the graphical valuation model. By leveraging the inherent structure of multigraphs—where items are represented by edges and agents by vertices—we established several positive results that distinguish this model from the general setting of indivisible goods. For additive valuations, we proved that an exact MMS orientation always exists and can be computed efficiently. Notably, this orientation also satisfies the stronger pairwise maximin share (PMMS) property, providing a complete solution for the additive case. Furthermore, we established that a $1$-out-of-$3$-MMS is also achievable for any number of additive agents.

In the realm of XOS valuations, we demonstrated a significant separation from general results by proving the existence of a $2/3$-MMS orientation for any number of agents, as well as a $1/2$-out-of-$2$ guarantee. We further bridged the gap between fair division and extremal combinatorics by showing that progress in the theory of independent transversal sets could lead to even tighter approximation guarantees for XOS agents. These tools guaranteed a $2/3$-out-of-$8 \sqrt{n}$-MMS orientation.

Finally, for subadditive valuations, we established a tight $1/2$-MMS approximation, matching the known upper bounds for the general model while showing that no better guarantee is possible even with graphical constraints. 

For both XOS and subadditive, agents we proved a tight $1/2$-PMMS orientation for any number of agents.

Several interesting directions for future research remain:

\begin{itemize}
\item \textbf{Tightening the XOS Gap:} While we proved a $2/3$ existence result, a gap remains between this and our $(1 - \frac{1}{\lceil n/2\rceil+1})$ upper bound for an arbitrary number of agents. Determining the exact value $\alpha$ for which $\alpha$-MMS orientations are guaranteed to exist for all $n$ remains a compelling open problem.

\item \textbf{MMS Approximation for Multi-Hypergraphs:} A natural generalization of our model is the multi-hypergraph setting, where each item (hyperedge) is relevant to a set of at most $p$ agents. Establishing approximation guarantees for hypergraphs remains a significant challenge; we note that for $p=n$, this coincides with the general MMS problem without any structural restrictions.
\end{itemize}

\section*{Acknowledgments}
The research project is implemented in the framework of H.F.R.I call “3rd Call for H.F.R.I.’s Research Projects to Support Faculty Members \& Researchers” (H.F.R.I. Project Number:24896). This work has been partially supported by project MIS 5154714 of the National Recovery and Resilience Plan Greece 2.0 funded by the European Union under the NextGenerationEU Program.


    





\clearpage
\bibliographystyle{plain}
\bibliography{Bibliography}

@article{Feige25,
  author       = {Uriel Feige},
  title        = {From multi-allocations to allocations, with subadditive valuations},
  journal      = {CoRR},
  volume       = {abs/2506.21493},
  year         = {2025}
}

@inproceedings{BSRS25,
  author       = {V{\'{a}}clav Blazej and
                  Sushmita Gupta and
                  M. S. Ramanujan and
                  Peter Strulo},
  title        = {Tractable Graph Structures in {EFX} Orientation},
  booktitle    = {{SAGT}},
  series       = {Lecture Notes in Computer Science},
  volume       = {15953},
  pages        = {175--190},
  publisher    = {Springer},
  year         = {2025}
}

@inproceedings{HeidariSodaAdditiveMMS,
  author       = {Ehsan Heidari and
                  Alireza Kaviani and
                  Masoud Seddighin and
                  AmirMohammad Shahrezaei},
  title        = {Improved Maximin Share Guarantee for Additive Valuations},
  booktitle    = {{SODA}},
  pages        = {2239--2290},
  publisher    = {{SIAM}},
  year         = {2026}
}

@article{HuangZHouAdditiveMMS,
  author       = {Xin Huang and
                  Shengwei Zhou},
  title        = {An {FPTAS} for 7/9-Approximation to Maximin Share Allocations},
  journal      = {CoRR},
  volume       = {abs/2511.13056},
  year         = {2025}
}

@inproceedings{DEGK24,
  author       = {Argyrios Deligkas and
                  Eduard Eiben and
                  Tiger{-}Lily Goldsmith and
                  Viktoriia Korchemna},
  title        = {{EF1} and {EFX} Orientations},
  booktitle    = {{IJCAI}},
  pages        = {56--63},
  publisher    = {ijcai.org},
  year         = {2025}
}

@article{ZM24,
	Author = {Jinghan A Zeng and Ruta Mehta},
	Journal = {CoRR},
	Title = {On the structure of envy-free orientations on graphs},
	Volume = {abs/2404.13527},
	Year = {2024}}

@article{BudishMMS,
	Author = {Budish, Eric},
	Journal = {Journal of Political Economy},
	Number = {6},
	Pages = {1061-1103},
	Title = {The Combinatorial Assignment Problem: Approximate Competitive Equilibrium from Equal Incomes},
	Volume = {119},
	Year = {2011}}

@article{KurokawaProcacciaWang18,
	Author = {David Kurokawa and Ariel D. Procaccia and Junxing Wang},
	Journal = {J. {ACM}},
	Number = {2},
	Pages = {8:1--8:27},
	Title = {Fair Enough: Guaranteeing Approximate Maximin Shares},
	Volume = {65},
	Year = {2018}}

@article{FeigeSapirTauberWINE21,
	Author = {Uriel Feige and Ariel Sapir and Laliv Tauber},
	Journal = {CoRR},
	Title = {A tight negative example for {MMS} fair allocations},
	Volume = {abs/2104.04977},
	Year = {2021}}

@article{AkramiGargSODA24,
	Author = {Hannaneh Akrami and Jugal Garg},
	Journal = {CoRR},
	Title = {Breaking the 3/4 Barrier for Approximate Maximin Share},
	Volume = {abs/2307.07304},
	Year = {2023}}

@article{ChekuriKKM24,
	Author = {Chandra Chekuri and Pooja Kulkarni and Rucha Kulkarni and Ruta Mehta},
	Journal = {CoRR},
	Title = {1/2 Approximate {MMS} Allocation for Separable Piecewise Linear Concave Valuations},
	Volume = {abs/2312.08504},
	Year = {2023}}

@article{Hummel:HSS24,
	Author = {Halvard Hummel},
	Journal = {CoRR},
	Title = {Maximin Shares in Hereditary Set Systems},
	Volume = {abs/2404.11582},
	Year = {2024}}

@article{UziahuFeigeSubmodular,
	Author = {Gilad Ben-Uziahu and Uriel Feige},
	Journal = {CoRR},
	Title = {On Fair Allocation of Indivisible Goods to Submodular Agents},
	Volume = {abs/2303.12444},
	Year = {2023}}

@article{AkramiGST23,
	Author = {Hannaneh Akrami and Jugal Garg and Eklavya Sharma and Setareh Taki},
	Journal = {CoRR},
	Title = {Simplification and Improvement of {MMS} Approximation},
	Volume = {abs/2303.16788},
	Year = {2023}}

@article{BarmanKrishnamurthy20,
	Author = {Siddharth Barman and Sanath Kumar Krishnamurthy},
	Journal = {{ACM} Trans. Economics and Comput.},
	Number = {1},
	Pages = {5:1--5:28},
	Title = {Approximation Algorithms for Maximin Fair Division},
	Volume = {8},
	Year = {2020}}

@article{WanlessWood21,
	Author = {Wanless, Ian M. and Wood, David R.},
	Journal = {SIAM Journal on Discrete Mathematics},
	Number = {3},
	Pages = {1663-1677},
	Title = {A General Framework for Hypergraph Coloring},
	Volume = {36},
	Year = {2022}}

@article{Zhang24,
	Author = {Sijia Dai and Xinru Guo and Huahua Miao and Guichen Gao and Yicheng Xu and Yong Zhang},
	Journal = {Theor. Comput. Sci.},
	Pages = {114388},
	Title = {The existence and efficiency of {PMMS} allocations},
	Volume = {989},
	Year = {2024}}

@article{MirsaSethia24,
	Author = {Neeldhara Misra and Aditi Sethia},
	Journal = {CoRR},
	Title = {Envy-Free and Efficient Allocations for Graphical Valuations},
	Volume = {abs/2410.14272},
	Year = {2024}}

@misc{christodoulouIJCAI25,
	Archiveprefix = {arXiv},
	Author = {George Christodoulou and Vasilis Christoforidis and Symeon Mastrakoulis and Alkmini Sgouritsa},
	Eprint = {2502.05141},
	Primaryclass = {cs.GT},
	Title = {Maximin Share Guarantees for Few Agents with Subadditive Valuations},
	Url = {https://arxiv.org/abs/2502.05141},
	Year = {2025}}

@article{Afshinmehr24,
	Author = {Mahyar Afshinmehr and Alireza Danaei and Mehrafarin Kazemi and Kurt Mehlhorn and Nidhi Rathi},
	Journal = {CoRR},
	Title = {{EFX} Allocations and Orientations on Bipartite Multi-graphs: {A} Complete Picture},
	Volume = {abs/2410.17002},
	Year = {2024}}

@article{SzaboTardos2003,
	Author = {Tibor Szab{\'{o}} and G{\'{a}}bor Tardos},
	Journal = {Comb.},
	Number = {3},
	Pages = {333--351},
	Title = {Extremal Problems For Transversals In Graphs With Bounded Degree},
	Volume = {26},
	Year = {2006}}

@article{CaragiannisKurokawaACM19,
	Author = {Ioannis Caragiannis and David Kurokawa and Herv{\'{e}} Moulin and Ariel D. Procaccia and Nisarg Shah and Junxing Wang},
	Journal = {{ACM} Trans. Economics and Comput.},
	Number = {3},
	Pages = {12:1--12:32},
	Title = {The Unreasonable Fairness of Maximum Nash Welfare},
	Volume = {7},
	Year = {2019}}

@article{Coauthor,
	Author = {George Christodoulou and Vasilis Christoforidis},
	Journal = {CoRR},
	Title = {Fair and Truthful Allocations Under Leveled Valuations},
	Volume = {abs/2407.05891},
	Year = {2024}}

@article{BhaskarPanditCoRR24,
	Author = {Umang Bhaskar and Yeshwant Pandit},
	Journal = {CoRR},
	Title = {{EFX} Allocations on Some Multi-graph Classes},
	Volume = {abs/2412.06513},
	Year = {2024}}

@article{HsuCoRR24,
	Author = {Kevin Hsu},
	Journal = {CoRR},
	Title = {{EFX} Orientations of Multigraphs},
	Volume = {abs/2410.12039},
	Year = {2024}}

@inproceedings{ZhouIJCAI24,
	Author = {Yu Zhou and Tianze Wei and Minming Li and Bo Li},
	Booktitle = {{IJCAI}},
	Pages = {3049--3056},
	Publisher = {ijcai.org},
	Title = {A Complete Landscape of {EFX} Allocations on Graphs: Goods, Chores and Mixed Manna},
	Year = {2024}}

@article{HosseiniSSH22,
	Author = {Hadi Hosseini and Andrew Searns and Erel Segal{-}Halevi},
	Journal = {J. Artif. Intell. Res.},
	Title = {Ordinal Maximin Share Approximation for Goods},
	Volume = {74},
	Year = {2022}}

@article{DaiLiuZhang2025,
  author       = {Tianjiao Dai and
                  Weichan Liu and
                  Xin Zhang},
  title        = {Independent transversal blow-up of graphs},
  journal      = {CoRR},
  volume       = {abs/2502.19682},
  year         = {2025}
}

@article{BarVermAAMAS20,
	Author = {Siddharth Barman and Paritosh Verma},
	Journal = {CoRR},
	Title = {Existence and Computation of Maximin Fair Allocations Under Matroid-Rank Valuations},
	Volume = {abs/2012.12710},
	Year = {2020}}

@article{BabaioffNT21,
	Author = {Moshe Babaioff and Noam Nisan and Inbal Talgam{-}Cohen},
	Journal = {Math. Oper. Res.},
	Number = {1},
	Pages = {382--403},
	Title = {Competitive Equilibrium with Indivisible Goods and Generic Budgets},
	Volume = {46},
	Year = {2021}}

@article{AkramiGST24,
	Author = {Hannaneh Akrami and Jugal Garg and Setareh Taki},
	Journal = {CoRR},
	Title = {Improving Approximation Guarantees for Maximin Share},
	Volume = {abs/2307.12916},
	Year = {2023}}

@article{Aigner-HorevSH22,
	Author = {Elad Aigner{-}Horev and Erel Segal{-}Halevi},
	Journal = {Inf. Sci.},
	Pages = {164--187},
	Title = {Envy-free matchings in bipartite graphs and their applications to fair division},
	Volume = {587},
	Year = {2022}}

@article{HosseiniSearns21,
	Author = {Hadi Hosseini and Andrew Searns},
	Journal = {CoRR},
	Title = {Guaranteeing Maximin Shares: Some Agents Left Behind},
	Volume = {abs/2105.09383},
	Year = {2021}}

@article{AmanatidisBirmpasMarkakisIjcai18,
	Author = {Georgios Amanatidis and Georgios Birmpas and Evangelos Markakis},
	Journal = {CoRR},
	Title = {Comparing Approximate Relaxations of Envy-Freeness},
	Volume = {abs/1806.03114},
	Year = {2018}}

@article{SeddighinSeddighin24,
	Author = {Masoud Seddighin and Saeed Seddighin},
	Journal = {Artif. Intell.},
	Pages = {104049},
	Title = {Improved maximin guarantees for subadditive and fractionally subadditive fair allocation problem},
	Volume = {327},
	Year = {2024}}

@inproceedings{Ruta,
	Author = {Pooja Kulkarni and Rucha Kulkarni and Ruta Mehta},
	Booktitle = {{AAMAS}},
	Pages = {2875--2876},
	Publisher = {{ACM}},
	Title = {Maximin Share Allocations for Assignment Valuations},
	Year = {2023}}

@article{AmanatidisABFLMVW23Survey,
	Author = {Georgios Amanatidis and Haris Aziz and Georgios Birmpas and Aris Filos{-}Ratsikas and Bo Li and Herv{\'{e}} Moulin and Alexandros A. Voudouris and Xiaowei Wu},
	Journal = {Artif. Intell.},
	Pages = {103965},
	Title = {Fair division of indivisible goods: Recent progress and open questions},
	Volume = {322},
	Year = {2023}}

@phdthesis{KurokavaThesis17,
	Author = {David Kurokawa},
	School = {Carnegie Mellon University},
	Title = {Fair Division in Game theoretic Settings},
	Year = {2017}}

@article{AkramiNeurIPS23,
	Author = {Hannaneh Akrami and Masoud Seddighin and Kurt Mehlhorn and Golnoosh Shahkarami},
	Journal = {CoRR},
	Title = {Randomized and Deterministic Maximin-share Approximations for Fractionally Subadditive Valuations},
	Volume = {abs/2308.14545},
	Year = {2023}}

@article{GhodsiHSSY21,
	Author = {Mohammad Ghodsi and Mohammad Taghi Hajiaghayi and Masoud Seddighin and Saeed Seddighin and Hadi Yami},
	Journal = {Math. Oper. Res.},
	Number = {3},
	Pages = {1038--1053},
	Title = {Fair Allocation of Indivisible Goods: Improvement},
	Volume = {46},
	Year = {2021}}

@article{GhodsiHSSY22,
	Author = {Mohammad Ghodsi and Mohammad Taghi Hajiaghayi and Masoud Seddighin and Saeed Seddighin and Hadi Yami},
	Journal = {Artif. Intell.},
	Pages = {103633},
	Title = {Fair allocation of indivisible goods: Beyond additive valuations},
	Volume = {303},
	Year = {2022}}

@article{EbeKrcJgall14,
	Author = {Tom{\'{a}}s Ebenlendr and Marek Krc{\'{a}}l and Jir{\'{\i}} Sgall},
	Journal = {Algorithmica},
	Number = {1},
	Pages = {62--80},
	Title = {Graph Balancing: {A} Special Case of Scheduling Unrelated Parallel Machines},
	Volume = {68},
	Year = {2014}}

@article{VershaeWiese10,
	Author = {Jos{\'{e}} Verschae and Andreas Wiese},
	Journal = {CoRR},
	Title = {On the Configuration-LP for Scheduling on Unrelated Machines},
	Volume = {abs/1011.4957},
	Year = {2010}}

@article{ChristodKoutsKov21,
	Author = {George Christodoulou and Elias Koutsoupias and Annam{\'{a}}ria Kov{\'{a}}cs},
	Journal = {CoRR},
	Title = {Truthful allocation in graphs and hypergraphs},
	Volume = {abs/2106.03724},
	Year = {2021}}

@article{GargTaki21,
	Author = {Jugal Garg and Setareh Taki},
	Journal = {Artif. Intell.},
	Pages = {103547},
	Title = {An improved approximation algorithm for maximin shares},
	Volume = {300},
	Year = {2021}}

@article{AmanatidisMNS17,
	Author = {Georgios Amanatidis and Evangelos Markakis and Afshin Nikzad and Amin Saberi},
	Journal = {{ACM} Trans. Algorithms},
	Number = {4},
	Pages = {52:1--52:28},
	Title = {Approximation Algorithms for Computing Maximin Share Allocations},
	Volume = {13},
	Year = {2017}}

@inproceedings{Christodoulou,
	Author = {George Christodoulou and Amos Fiat and Elias Koutsoupias and Alkmini Sgouritsa},
	Booktitle = {{EC}},
	Pages = {473--488},
	Publisher = {{ACM}},
	Title = {Fair allocation in graphs},
	Year = {2023}}

@misc{FeigeHuang25,
      title={Concentration and maximin fair allocations for subadditive valuations}, 
      author={Uriel Feige and Shengyu Huang},
      year={2025},
      eprint={2502.13541},
      archivePrefix={arXiv},
      primaryClass={cs.GT},
      url={https://arxiv.org/abs/2502.13541}, 
}

@inproceedings{FeigeGrindberg25,
  author       = {Uriel Feige and
                  Vadim Grinberg},
  title        = {Fair allocations with subadditive and {XOS} valuations},
  booktitle    = {{EC}},
  pages        = {160--185},
  publisher    = {{ACM}},
  year         = {2025}
}

@inproceedings{SgouritsSotiriouEFXMULTI2025,
  author       = {Alkmini Sgouritsa and
                  Minas Marios Sotiriou},
  title        = {On the Existence of {EFX} Allocations in Multigraphs},
  booktitle    = {{AAMAS}},
  pages        = {2735--2737},
  publisher    = {International Foundation for Autonomous Agents and Multiagent Systems
                  / {ACM}},
  year         = {2025}
}

@article{PennyHaxellSzabo2006oddtransversals,
  author       = {Penny E. Haxell and
                  Tibor Szab{\'{o}}},
  title        = {Odd Independent Transversals are Odd},
  journal      = {Comb. Probab. Comput.},
  volume       = {15},
  number       = {1-2},
  pages        = {193--211},
  year         = {2006}
}
\end{document}